%% file: main.tex
\def\submission{0}

\ifnum\submission=0
\documentclass[11pt,letterpaper]{article}
\else
\documentclass[envcountsect]{llncs}
\makeatletter
\renewcommand{\@Opargbegintheorem}[4]{%
  #4\trivlist\item[\hskip\labelsep{#3#2\@thmcounterend}]}
\makeatother
\fi

\usepackage[utf8]{inputenc}
\usepackage{amsmath}
\usepackage{amsfonts}
\usepackage{xcolor}
\usepackage{threeparttable}
\usepackage{booktabs}
\usepackage{breakcites}

\ifnum\submission=0
\usepackage{amsthm}
\fi

\usepackage{braket}
\usepackage{authblk}
\usepackage[colorlinks=true,allcolors=blue]{hyperref}

\ifnum\submission=0
\usepackage[margin=1in]{geometry}
\fi

\usepackage{xcolor}
\usepackage{algorithm}
\usepackage{mdframed} 
\mdfdefinestyle{figstyle}{ %
  linecolor=black!7, %
  backgroundcolor=black!7, %
  innertopmargin=10pt, %
  innerleftmargin=\textwidth, %
  innerrightmargin=\textwidth, %
  innerbottommargin=5pt %
}

\usepackage{comment}
\usepackage[capitalise,noabbrev]{cleveref}

\ifnum\submission=0
\newtheorem{theorem}{Theorem}[section]

\newtheorem{lemma}[theorem]{Lemma}

\newtheorem{corollary}[theorem]{Corollary}
\newtheorem{definition}[theorem]{Definition}

\newtheorem{remark}{Remark}

\fi

\newtheorem{myclaim}[theorem]{Claim}

	\crefname{theorem}{Theorem}{Theorems}
	\crefname{assumption}{Assumption}{Assumptions}
	\crefname{construction}{Construction}{Constructions}
	\crefname{corollary}{Corollary}{Corollaries}
	\crefname{conjecture}{Conjecture}{Conjectures}
	\crefname{definition}{Definition}{Definitions}
	\crefname{example}{Example}{Examples}
	\crefname{experiment}{Experiment}{Experiments}
	\crefname{counterexample}{Counterexample}{Counterexamples}
	\crefname{lemma}{Lemma}{Lemmas}
	\crefname{observation}{Observation}{Observations}
	\crefname{proposition}{Proposition}{Propositions}
	\crefname{remark}{Remark}{Remarks}
	\crefname{claim}{Claim}{Claims}
	\crefname{myclaim}{Claim}{Claims}
	\crefname{fact}{Fact}{Facts}
	\crefname{note}{Note}{Notes}
	\crefname{figure}{Figure}{Figures}

\input{macros.tex}

\ifnum\submission=0
\title{Classically Verifiable NIZK for QMA with Preprocessing}
\else
\title{Classically Verifiable NIZK for QMA with Preprocessing\thanks{{\color{red} To reviewers who have reviewerd this paper before: See \cref{sec:differences} for the main differences from the previous version.}}}
\fi
\begin{document}
\ifnum\submission=0
\begin{flushright}
YITP-21-10
\end{flushright}
\fi

\ifnum\submission=0
\newcommand*{\email}[1]{\normalsize\href{mailto:#1}{#1}}
\author[1]{Tomoyuki Morimae}
\author[2]{Takashi Yamakawa}
\affil[1]{Yukawa Institute for Theoretical Physics,
Kyoto University and PRESTO, JST, Japan \email{tomoyuki.morimae@yukawa.kyoto-u.ac.jp}}
\affil[2]{NTT Corporation, Japan \email{ takashi.yamakawa.ga@hco.ntt.co.jp}}
\else
\author{\empty}\institute{\empty} 
\fi

{\let\newpage\relax\maketitle}

\ifnum\submission=1
\vspace{-7mm} 
\fi

\begin{abstract}
We propose three constructions of
classically verifiable 
non-interactive zero-knowledge proofs and arguments (CV-NIZK) for $\QMA$ in various preprocessing models.  
\begin{enumerate}
    \item We construct a CV-NIZK for $\QMA$ in the quantum secret parameter model where a trusted setup sends a quantum proving key 
    to the prover and a classical verification key to the verifier. It is information theoretically sound and zero-knowledge. 
    \item 
    Assuming the quantum hardness of the learning with errors problem,  
    we construct a CV-NIZK for $\QMA$  in a model where a trusted party generates a CRS and the verifier sends an instance-independent quantum message to the prover as preprocessing.
    This model is the same as one considered in the recent work by Coladangelo, Vidick, and Zhang (CRYPTO '20).  
    Our construction has the so-called dual-mode property, which means that there are two computationally indistinguishable modes of generating CRS, and we have information theoretical soundness in one mode and information theoretical zero-knowledge property in the other.
    This answers an open problem left by Coladangelo et al, which is to achieve either of soundness or zero-knowledge information theoretically.
    To the best of our knowledge, ours is the first dual-mode NIZK for $\QMA$ in any kind of model.
    \item
    We construct a CV-NIZK for $\QMA$ with quantum preprocessing in the quantum random oracle model. This quantum preprocessing is the one where the verifier sends a random Pauli-basis states to the prover. Our construction uses the Fiat-Shamir transformation. The quantum preprocessing can be replaced with the setup that distributes Bell pairs among the prover and the verifier, and therefore 
    we solve the open problem by Broadbent and Grilo (FOCS '20) about the possibility of NIZK for $\QMA$ in the shared Bell pair model
    via the Fiat-Shamir transformation.
    
\end{enumerate}
\end{abstract}
\input{Introduction}
\input{Preliminaries}

\input{NIZK_information_theoretical}

\input{NIZK_computational}
\input{Fiat-Shamir}

\input{bib.tex}

\appendix
\ifnum\submission=0
\else
\newpage
 	\setcounter{page}{1}
 	{
	\noindent
 	\begin{center}
	{\Large SUPPLEMENTAL MATERIALS}
	\end{center}
 	}
	\setcounter{tocdepth}{2}
\fi
\ifnum\submission=0
\input{Broadbent-Grilo}
\input{Alternative_Proofs}

\input{Alternative_NIZK}
\input{Verification}

\input{OT}
\input{Appendix_bell}
\else
\input{differences}
\input{Omitted_Related_Works}
\input{Appendix_Preliminaries}

\input{Broadbent-Grilo}
\input{Appendix_proof_amplification}

\input{Alternative_NIZK}
\input{Verification}
\input{Alternative_Proofs}
\input{Appendix_Omitted_Dual-Mode}
\input{OT}
\input{Appendix_Omitted_Fiat-Shamir}
\input{Appendix_bell}
\fi
\newpage
  \tableofcontents
  \thispagestyle{empty}
\end{document}

%% file: macros.tex
\def \sample { \overset{\hspace{0.1em}\mathsf{\scriptscriptstyle\$}}{\leftarrow} }

\newcommand{\bit}{\{0,1\}}
\newcommand{\ra}{\rightarrow}
\newcommand{\la}{\leftarrow}
\newcommand{\secpar}{\lambda}
\newcommand{\negl}{\mathsf{negl}}
\newcommand{\poly}{\mathsf{poly}}
\newcommand{\st}{\mathsf{st}}
\newcommand{\td}{\mathsf{td}}
\newcommand{\concat}{\|}

\newcommand{\whx}{\widehat{x}}
\newcommand{\whz}{\widehat{z}}

\newcommand{\reprogram}{\mathsf{Reprogram}}
\newcommand{\ROdist}{\mathsf{ROdist}}

\newcommand{\commit}{\mathsf{Commit}}
\newcommand{\decom}{d}

\newcommand{\game}{\mathsf{Game}}
\newcommand{\event}{\mathsf{E}}
\newcommand{\Win}{\mathsf{Win}}

\newcommand{\Tr}{\mathrm{Tr}}
\newcommand{\hist}{\mathrm{hist}}
\newcommand{\ham}{\mathcal{H}}

\newcommand*{\regA}{\mathbf{A}}
\newcommand*{\regC}{\mathbf{C}}
\newcommand*{\regB}{\mathbf{B}}
\newcommand*{\regP}{\mathbf{P}}
\newcommand*{\regV}{\mathbf{V}}

\newcommand{\A}{\mathcal{A}}

\newcommand{\dist}{\mathcal{D}}
\newcommand{\ora}{\mathcal{O}}

\newcommand{\NP}{\mathbf{NP}}
\newcommand{\QMA}{\mathbf{QMA}}
\newcommand{\BQP}{\mathbf{BQP}}
\newcommand{\statement}{\mathtt{x}}
\newcommand{\witness}{\mathtt{w}}
\newcommand{\yes}{\mathsf{yes}}
\newcommand{\no}{\mathsf{no}}

\newcommand{\OT}{\mathsf{OT}}
\newcommand{\ot}{\mathsf{ot}}

\newcommand{\otst}{\mathsf{st}}

\newcommand{\sender}{\mathsf{Sender}}
\newcommand{\receiver}{\mathsf{Receiver}}
\newcommand{\derive}{\mathsf{Derive}}

\newcommand{\vecmu}{\boldsymbol{\mu}}
\newcommand{\Open}{\mathsf{Open}}

\newcommand{\sen}{\mathsf{sen}}
\newcommand{\rec}{\mathsf{rec}}

\newcommand{\onen}{1\text{-}n}
\newcommand{\kn}{k\text{-}n}

\newcommand{\com}{\mathsf{com}}

\newcommand{\NIP}{\mathsf{NIP}}

\newcommand{\NIZK}{\mathsf{NIZK}}

\newcommand{\setup}{\mathsf{Setup}}
\newcommand{\prove}{\mathsf{Prove}}
\newcommand{\verify}{\mathsf{Verify}}
\newcommand{\siml}{\mathsf{Sim}}
\newcommand{\crs}{\mathsf{crs}}
\newcommand{\CRS}{\mathsf{CRS}}
\newcommand{\pkey}{k_P}
\newcommand{\vkey}{k_V}

\newcommand{\virone}{\mathsf{vir}\text{-}1}
\newcommand{\virtwo}{\mathsf{vir}\text{-}2}
\newcommand{\CRSVP}{CRS + $({V\ra P})~$}
\newcommand{\QROVP}{QRO + $({V\ra P})~$}

\newcommand{\mode}{\mathsf{mode}}
\newcommand{\binding}{\mathsf{binding}}
\newcommand{\hiding}{\mathsf{hiding}}
\newcommand{\preprocess}{\mathsf{Preprocess}}
\newcommand{\crsgen}{\mathsf{CRSGen}}

\newcommand{\DM}{\mathsf{DM}}

\newcommand{\DEnc}{\mathsf{DEnc}}
\newcommand{\keygen}{\mathsf{KeyGen}}
\newcommand{\enc}{\mathsf{Enc}}
\newcommand{\dec}{\mathsf{Dec}}
\newcommand{\findmessy}{\mathsf{FindMessy}}
\newcommand{\trapkeygen}{\mathsf{TrapKeyGen}}
\newcommand{\messymode}{\mathsf{messy}}
\newcommand{\decmode}{\mathsf{dec}}
\newcommand{\pk}{\mathsf{pk}}
\newcommand{\sk}{\mathsf{sk}}
\newcommand{\ct}{\mathsf{ct}}

\newcommand{\LE}{\mathsf{LE}}
\newcommand{\injgen}{\mathsf{InjGen}}
\newcommand{\lossygen}{\mathsf{LossyGen}}

\newcommand{\yam}[1]{{\color{orange}[Yamakawa: #1]}}
\newcommand{\mor}[1]{{\color{purple}[Morimae: #1]}}

\newcommand{\msg}{\mathsf{msg}}

%% file: Introduction.tex
\section{Introduction}
\subsection{Background}

The zero-knowledge \cite{SICOMP:GolMicRac89}, which ensures that the verifier learns nothing beyond the statement proven by the prover,
is one of the most central concepts in cryptography.
Recently, there have been many works that constructed non-interactive zero-knowledge (NIZK) \cite{STOC:BluFelMic88} proofs or arguments for $\QMA$, which is the ``quantum counterpart" of $\NP$, in various kind of models \cite{TCC:ACGH20,C:ColVidZha20,FOCS:BroGri20,C:Shmueli21,C:BCKM21a,BartusekMalavolta}.   
We note that we require the honest prover to run in quantum polynomial-time receiving  sufficiently many copies of a witness when we consider NIZK proofs or arguments for $\QMA$.
All known protocols except for the protocol of Broadbent and Grilo \cite{FOCS:BroGri20} only satisfy computational soundness. 
The protocol of \cite{FOCS:BroGri20} satisfies information theoretical soundness and zero-knowledge in the \emph{secret parameter (SP) model} \cite{C:PasShe05} where a trusted party generates proving and verification keys and gives them to the corresponding party while keeping it secret to the other party as setup.\footnote{The SP model is also often referred to as \emph{preprocessing model} \cite{C:DeSMicPer88}. 
}
A drawback of their protocol is that  the prover sends a quantum proof to the verifier, and thus the verifier should be quantum. 
Therefore it is natural to ask the following question.
\begin{center}
\emph{Can we construct a
NIZK proof for $\QMA$ with classical verification assuming a trusted party that generates proving and verification keys?} 
\end{center}

In addition, the SP model is not a very desirable model since it assumes a  strong trust in the setup.  
In the classical literature, there are  constructions of NIZK proofs for $\NP$ in the common reference string (CRS) model \cite{STOC:BluFelMic88,FLS99,C:PeiShi19} where the only trust in the setup is that a classical string is chosen according to a certain distribution and then published.
Compared to the SP model, we need to put much less trust in the setup in the CRS model.
Indeed, several works \cite{FOCS:BroGri20,C:ColVidZha20,C:Shmueli21} mention it as an open problem to construct a NIZK proofs (or even arguments) for $\QMA$ in the CRS model. 
Though this is still open, there are several constructions of NIZKs for $\QMA$ in different models that assume less trust in the setup than in the SP model \cite{C:ColVidZha20,C:Shmueli21,C:BCKM21a}.   
However, all of them are arguments.
Therefore, we ask the following question.
\begin{center}
\emph{Can we construct a
NIZK proof for $\QMA$ with classical verification in a model that assumes less trust in the setup than in the SP model?}
\end{center}

The Fiat-Shamir transformation \cite{C:FiaSha86} is one of the most important techniques in cryptography that have many applications.
In particular, NIZK can be constructed from a $\Sigma$ protocol: the prover generates the verifier's challenge $\beta$
by itself by applying a random oracle $H$ on the prover's first message $\alpha$, and then the prover
issues the proof $\pi=(\alpha,\gamma)$, where
$\gamma$ is the third message generated from $\alpha$ and $\beta=H(\alpha)$.
It is known that Fiat-Shamir transform works in the post-quantum setting where we consider classical protocols secure against quantum adversaries \cite{C:LiuZha19,C:DFMS19,C:DonFehMaj20}.
On the other hand, it is often pointed out that (for example, \cite{C:Shmueli21,FOCS:BroGri20}) 
this standard technique cannot be used in the fully quantum setting. In particular, due to 
the no-cloning, the application of random oracle on the first message does not work when the first message is
quantum 
like so-called the $\Xi$-protocol constructed by Broadbent and Grilo~\cite{FOCS:BroGri20}.
Broadbent and Grilo 
left the following open problem:
\begin{center}
\emph{Is it possible to construct NIZK for $\QMA$ in the CRS model (or shared Bell pair model) via
the Fiat-Shamir transformation?}
\end{center}
Note that the shared Bell pair model is the setup model where the setup distributes Bell pairs among the prover and the verifier.
It can be considered as a ``quantum analogue" of the CRS \cite{Kobayashi03}.

\subsection{Our Results}
We answer the above questions affirmatively.
\begin{enumerate}

    \item We construct a classically verifiable NIZK  (CV-NIZK) for $\QMA$ in the QSP model where a trusted party generates a quantum proving key and classical verification key and gives them to the corresponding parties. 
    We do not rely on any computational assumption for this construction either, and thus both soundness and the zero-knowledge property are satisfied information theoretically. 
    This answers our first question. 
    Compared with \cite{FOCS:BroGri20}, ours has an advantage that verification is classical at the cost of making the proving key quantum. 
    The proving key is a very simple state, i.e., a tensor product of randomly chosen
    Pauli $X$, $Y$, or $Z$ basis states.
    We note that we should not let the verifier  play the role of the trusted party for this construction since that would break the zero-knowledge property. 
    \item 
    Assuming the quantum hardness of the learning with errors problem (the LWE assumption)~\cite{JACM:Regev09},  
    we construct a CV-NIZK for $\QMA$ in a model where a trusted party generates a CRS and the verifier sends an instance-independent quantum message to the prover as preprocessing.
    We note that the CRS is reusable for generating multiple proofs but the quantum message in the preprocessing is not reusable. 
    In this model, we only assume a trusted party that just generates a CRS once, and thus this answers our second question. 
    This model is the same as one considered in \cite{C:ColVidZha20} recently, and we call it the \CRSVP model. 
    Compared to their work, our construction has the following advantages.
    \begin{enumerate}
    \item In their protocol, both soundness and the zero-knowledge property hold only against quantum polynomial-time adversaries, and they left it open to achieve either of them information theoretically. 
    We answer the open problem.
    Indeed, our construction has the so-called dual-mode property~\cite{GOS,C:PeiShi19}, which means that there are two computationally indistinguishable modes of generating CRS, and we have information theoretical soundness in one mode and information theoretical zero-knowledge property in the other.
    To the best of our knowledge, ours is the first dual-mode NIZK for $\QMA$ in any kind of model.
    \item Our protocol uses underlying cryptographic primitives (which are lossy encryption and oblivious transfer with certain security) only in a black-box manner whereas their protocol heavily relies on non-black-box usage of the underlying primitives.
    Indeed, their protocol uses 
    fully homomorphic encryption to homomorphically runs the proving algorithm of a NIZK for $\NP$, which would make the protocol extremely inefficient. 
    On the other hand, our construction uses the underlying primitives only in a black-box manner, which results in a much more efficient construction.
    We note that black-box constructions have been considered desirable for both theoretical and practical reasons in the cryptography community (e.g., see introduction of \cite{STOC:IKLP06}). 
    \item The verifier's quantum operation in our preprocessing is simpler than that in theirs: in the preprocessing of our protocol, the verifier has only to do single-qubit gate operations (Hadamard, bit-flip or phase gates), while in
    the preprocessing of their protocol, the verifier has to do five-qubit (entangled) Clifford operations.
    In their paper~\cite{C:ColVidZha20}, they left the following open problem: how far their preprocessing phase could be weakened? Our construction with the weaker verifier therefore partially answers the open problem.
    \end{enumerate}
On the other hand, Coladangelo et al. \cite{C:ColVidZha20} proved that their protocol is also an \emph{argument of quantum knowledge (AoQK)}.
We leave it open to study if ours is also a proof/argument of knowledge.
\item
    We construct a CV-NIZK for $\QMA$ with quantum preprocessing in the quantum random oracle model. This quantum preprocessing is the one where the verifier sends a random Pauli-basis states to the prover. Our construction uses the Fiat-Shamir transformation. Importantly, the quantum preprocessing can be replaced with the setup that distributes Bell pairs among the prover and the verifier. The distribution of Bell pairs by the setup can be considered as a ``quantum analogue" of the CRS. This result gives an answer 
    to our third question (and the second question as well).
    (Note that both the soundness and zero-knowledge property of the construction
    are computational one, but it does not mean that we use some computational assumptions:
    just the oracle query is restricted to be polynomial time.)
\end{enumerate}

\paragraph{Comparison among NIZKs for $\QMA$.}
\input{table}

We give more comparisons among our and known constructions of  NIZKs for $\QMA$.
Since we already discuss comparisons with ours and \cite{FOCS:BroGri20,C:ColVidZha20}, we discuss comparisons with other works.  A summary of the comparisons is given in \cref{tbl:compare_nizk}. 

Alagic et al. \cite{TCC:ACGH20} gave a construction of a NIZK for $\QMA$ in the SP model.  
Their protocol has an advantage that both the trusted party and verifier are completely classical.   
On the other hand, the drawback is that only computational soundness and zero-knowledge are achieved,
whereas our first two constructions achieve (at least) either statistical soundness or zero-knowledge.
Their protocol also uses the Fiat-Shamir transformation with quantum random oracle like our third result, but
their setup is the secret parameter model, whereas ours can be the sharing Bell pair model, which is a quantum analogue of the CRS model.

Shmueli \cite{C:Shmueli21} gave a construction of a NIZK for $\QMA$ in the malicious designated-verifier (MDV) model, where a trusted party generates a CRS and the verifier sends an instance-independent classical message to the prover as preprocessing. 
In this model, the preprocessing is \emph{reusable}, i.e., a single preprocessing can be reused to generate arbitrarily many proofs later. 
This is a crucial advantage of their construction compared to ours. 
On the other hand,  
in their protocol, proofs are quantum and thus 
the verifier should perform quantum computations in the online phase whereas the online phase of the verifier is classical in our constructions.  
Also, their protocol only satisfies computational soundness and zero-knowledge whereas we can achieve (at least) either of them statistically.

Recently, Bartusek et al. \cite{C:BCKM21a} gave another construction of a NIZK for $\QMA$ in the MDV model that has an advantage that the honest prover only uses a single copy of a witness. (Note that all other NIZKs for $\QMA$ including ours require the honest prover to take multiple copies of a witness if we require neglible completeness and soundness errors.)
However, their construction also requires quantum verifier in the online phase and only achieves computational soundness and zero-knowledge similarly to \cite{C:Shmueli21}.

Subsequently to our work, Bartusek and Malavolta~\cite{BartusekMalavolta} recently constructed the first CV-NIZK argument for $\QMA$ in the CRS model assuming the LWE assumption and ideal obfuscation for classical circuits.
An obvious drawback is the usage of ideal obfuscation, which has no provably secure instantiation.\footnote{In the latest version, they give a candidate instantiation based on indistinguishability obfuscation and random oracles. However,
the instantiation is heuristic since they obfuscate circuits that involve the random oracle, which cannot be done in the quantum random oracle model.}
They also construct a witness encryption scheme for $\QMA$ under the same assumptions.
They use the verification protocol of Mahadev~\cite{FOCS:Mahadev18b} and therefore the LWE assumption is necessary.
If our CV-NIZK in the QSP model is used, instead, a witness encryption for $\QMA$ (with quantum ciphertext) would be constructed without the LWE assumption,  
which is one interesting application of our results.

\subsection{Technical Overview}

\paragraph{Classically verifiable NIZK for $\QMA$ in the QSP model.}
Our starting point is the NIZK for $\QMA$ in \cite{FOCS:BroGri20},
which is based on the
fact that 
a $\QMA$ language can be reduced to the $5$-local Hamiltonian problem with \emph{locally simulatable} history 
states \cite{FOCS:BroGri20,FOCS:GriSloYue19}.
(We will explain later the meaning of ``locally simulatable".)
An instance $\statement$ corresponds to an  $N$-qubit Hamiltonian $\ham_\statement$ of the form  
$
\ham_\statement=\sum_{i=1}^M p_i \frac{I+s_i P_i}{2},
$
where $N=\poly(|\statement|)$, $M=\poly(|\statement|)$, $s_i\in\{+1,-1\}$, $p_i> 0$, $\sum_{i=1}^M p_i=1$, and
$P_i$ is a tensor product of Pauli operators $(I,X,Y,Z)$
with at most $5$ nontrivial Pauli operators $(X,Y,Z)$.
There are $0<\alpha<\beta<1$ with $\beta-\alpha = 1/\poly(|\statement|)$ such that 
if $\statement$ is a yes instance, then there exists a state $\rho_\hist$ (called the \emph{history state}) 
such that $\Tr(\rho_\hist \ham_\statement)\leq \alpha$, and
if $\statement$ is a no instance, then for any state $\rho$, we have $\Tr(\rho \ham_\statement)\geq \beta$.

The completeness and the soundness of the NIZK for $\QMA$ in \cite{FOCS:BroGri20} is based on the posthoc 
verification protocol \cite{posthoc}, which is explained as follows.
To prove that $\statement$ is a yes instance, the prover sends the history state to the verifier.
The verifier first chooses $P_i$ with probability $p_i$,   
and measures each qubit in the Pauli basis corresponding to $P_i$.
Let $m_j\in\{0,1\}$ be the measurement result on $j$th qubit.
The verifier accepts if  $(-1)^{\oplus_j m_j}=-s_i$ and rejects otherwise. 
The probability that the verifier accepts is $1-\Tr(\rho \ham_\statement)$ when the prover's quantum message
is $\rho$, and therefore 
the verifier accepts with probability at least $1-\alpha$ if $\statement$ is a yes instance and the prover is honest whereas it accepts with probability at most $1-\beta$ if $\statement$ is a no instance. (See \cref{lem:prob_and_energy} and \cite{posthoc}.)
The gap between completeness and soundness can be amplified by simple parallel repetitions.

The verifier in the posthoc protocol is, however, not classical,
because it has to receive a quantum state and measure each qubit.
Our first idea to make the verifier classical is to use the quantum teleportation.
Suppose that the prover and verifier share sufficiently many Bell pairs at the beginning. 
Then the prover  can send the history state to the verifier with classical communication by the quantum teleportation. 
Though this removes the necessity of quantum communication, the verifier still needs to be quantum since it has to keep halves of Bell pairs and perform a measurement after receiving a proof. 

To solve the problem, we utilize our observation that the verifier's measurement and the prover's measurement commute
with each other, which is our second idea. In other words,  we can let the verifier perform the measurement at the beginning without losing completeness or soundness. 
In the above quantum-teleportation-based protocol, when the prover sends its measurement outcomes $\{(x_j,z_j)\}_{j\in [N]}$ to the verifier,  the verifier's state collapses to $X^xZ^z  \rho_\hist Z^z X^x$  
where $\rho_\hist$ denotes the history state and 
$X^xZ^z$ means $\prod_{j=1}^{N} X_j^{x_j}Z_j^{z_j}$.  
Then the verifier applies the Pauli correction $X^xZ^z$ and then measures each qubit in a Pauli basis. 
We observe that the Pauli correction can be applied even after the verifier measures each qubit 
because $X_j^{x_j}Z_j^{z_j}$ before a Pauli measurement on the $j$th qubit
has the same effect as XOR by $z_j$ or $x_j$ after the measurement (see \cref{lem:XZ_before_measurement}).
Therefore, if a trusted party generates Bell pairs and measures half of them in random Pauli basis and gives the unmeasured halves to the prover as a proving key while the measurement outcomes to the verifier as a verification key, a completely classical verifier can verify the $\QMA$ language.

The last remaining issue is that the 
distribution of bases that appear in $P_i$ depends on the instance $\statement$, and thus we cannot sample the distribution at the setup phase 
where $\statement$ is not decided yet. 
To resolve this issue, we use the following idea (which was also used in \cite{TCC:ACGH20}). 
The trusted party just chooses random bases, and the verifier just accepts if they are inconsistent to $P_i$ chosen by the verifier in the online phase. 
Since there are only $3$ possible choices of the bases and $P_i$ non-trivially acts on  at most $5$ qubits, the probability that the randomly chosen bases are consistent to $P_i$ is at least $3^{-5}$.\footnote{There is a subtle issue that the probability depends on the number of qubits on which $P_i$ non-trivially acts. We adjust this by an additional biased coin flipping.}
Therefore we can still achieve inverse-polynomial gap between completeness and soundness.

The zero-knowledge property of the NIZK for $\QMA$ in \cite{FOCS:BroGri20} uses the local simulatability of the history state.
It roughly means that a classical description of the reduced density matrix of the history state 
for any $5$-qubit subsystem can be efficiently computable without knowing the witness.
Broadbent and Grilo \cite{FOCS:BroGri20} used this local simulatability to achieve the zero-knowledge property
as follows.
A trusted party randomly chooses $(\whx,\whz)\sample \bit^{N}\times \bit^{N}$, and randomly picks a random subset $S_V\subseteq [N]$ such that $1\le|S_V|\leq 5$. 
Then it gives $(\whx,\whz)$ to the prover as a proving key and gives $\{(\whx_j,\whz_j)\}_{j\in S_V}$ to the verifier as a verification key where $\whx_j$ and $\whz_j$ denote the $j$-th bits of $\whx$ and $\whz$, respectively.
The prover generates the history state $\rho_\hist$ and sends $\rho'=X^{\whx}Z^{\whz}\rho_\hist Z^{\whz}X^{\whx}$ to the verifier as a proof. 
The verifier then measures each qubit as is done in the posthoc verification protocol. This needs the quantum verifier, but as we have explained, we can make the verifier classical by using the teleportation technique. 

An intuitive explanation of why it is zero-knowledge is that the verifier can access at most five qubits of the history state, because other qubits are quantum one-time padded.
Due to the local simulatability of the history state, the information that the verifier gets can be classically
simulated without the witness.
This results in our classically verifiable NIZK for $\QMA$ in the QSP model. 
In our QSP model, the trusted setup sends random Pauli basis states to the prover and their classical description
to the verifier. Furthermore, the trusted setup also
sends randomly chosen $(\whx,\whz)\sample \bit^{N}\times \bit^{N}$ to the prover, 
and $\{(\whx_j,\whz_j)\}_{j\in S_V}$ to the verifier with randomly chosen subset $S_V$.

\paragraph{Classically verifiable NIZK for QMA in the \CRSVP model.}
We want to reduce the trust in the setup, so let us first examine what happens if the verifier runs the setup as preprocessing. 
Unfortunately, such a construction is not zero-knowledge since the verifier can know whole bits of $(\whx,\whz)$ and thus it may obtain information of qubits of $\rho_\hist$ that are outside of $S_V$, in which case we cannot rely on the local simulatability.
Therefore, for ensuring the zero-knowledge property, we have to make sure that the verifier only knows $\{(\whx_j,\whz_j)\}_{j\in S_V}$.
Then suppose that the prover chooses $(\whx,\whz)$ whereas other setups are still done by the verifier. 
Here, the problem is how to let the verifier know $\{(\whx_j,\whz_j)\}_{j\in S_V}$.
A naive solution is that the verifier sends  $S_V$ to the prover and then the prover returns $\{(\whx_j,\whz_j)\}_{j\in S_V}$.
However, such a construction is not sound since it is essential that the prover ``commits" to a single quantum state independently of $S_V$ when reducing soundness to the local Hamiltonian problem.
So what we need is a protocol between the prover and verifier where the verifier only gets $\{(\whx_j,\whz_j)\}_{j\in S_V}$ and the prover does not learn $S_V$.
We observe that this is exactly the functionality of \emph{$5$-out-of-$N$ oblivious transfer}~\cite{C:BraCreRob86}.

Though it may sound easy to solve the problem by just using a known two-round $5$-out-of-$N$ oblivious transfer, there is still some subtlety. 
For example, if we use an oblivious transfer that satisfies only indistinguishability-based notion of receiver's security (e.g., \cite{NP01,TCC:BraDot18}),\footnote{The indistinguishability-based  receiver's security is also often referred to as half-simulation security \cite{EC:CamNevshe07}.} which just says that the sender cannot know indices chosen by the receiver, we cannot prove soundness. 
Intuitively, this is because the indistinguishability-based receiver's security does not prevent a malicious sender from generating a malicious message such that the message derived on the receiver's side depends on the chosen indices, which does not force the prover to ``commit" to a single state.  

If we use a 
\emph{fully-simulatable} \cite{RSA:Lindell08a} 
oblivious transfer, the above problem does not arise and we can prove both soundness and zero-knowledge.  
However, the problem is that we are not aware  of any efficient fully-simulatable $5$-out-of-$N$ oblivious transfer based on post-quantum assumptions (in the CRS model).
The LWE-based construction of \cite{C:PeiVaiWat08} does not suffice for our purpose since a CRS can be reused only a bounded number of times in their construction.
Recently, Quach \cite{SCN:Quach20} resolved this issue, and proposed an efficient fully-simulatable $1$-out-of-$2$ oblivious transfer based on the LWE assumption.\footnote{Actually, his construction satisfies a stronger UC-security \cite{Canetti20,C:PeiVaiWat08}.} 
We can extend his construction to a fully-simulatable $1$-out-of-$N$ oblivious transfer efficiently. 
However, we do not know how to convert this into $5$-out-of-$N$ one efficiently without losing the full-simulatability.
We note that a conversion from $1$-out-of-$N$ to $5$-out-of-$N$ oblivious transfer by a simple $5$-parallel repetition  loses the full-simulatability against malicious senders since a malicious sender can send different inconsistent messages in different sessions, which should be considered as an attack against the full-simulatability.
One possible way to prevent such an inconsistent message attack is to let the sender prove that the messages in all sessions are consistent by using (post-quantum) CRS-NIZK for $\NP$ \cite{C:PeiShi19}.
However, such a construction is very inefficient since it uses the underlying $1$-out-of-$N$ oblivious transfer in a non-black-box manner, which we want to avoid.

We note that the parallel repetition construction preserves indistinguishability-based  receiver's security and fully-simulatable sender's security for two-round protocols. Therefore, we have an efficient (black-box) construction of $5$-out-of-$N$ oblivious transfer if we relax the receiver's security to the indistinguishability-based one.
As already explained, such a security does not suffice for proving soundness. 
To resolve this issue, we add an additional mechanism to force the prover to ``commit" to a single state. 
Specifically, instead of directly sending $(x,z)$ by a $5$-out-of-$N$ oblivious transfer, the prover sends a commitment of $(x,z)$ and then sends $(x,z)$ and the corresponding randomness used in the commitment by a $5$-out-of-$N$ oblivious transfer.
When the verifier receives $\{x_j,z_j\}_{j\in S_V}$ and corresponding randomness, it checks if it is consistent to the commitment by recomputing it, and immediately rejects if not.
This additional mechanism prevents a malicious prover's inconsistent behavior, which resolves the problem in the proof of soundness. 

Finally, our construction satisfies the dual-mode property if we assume appropriate dual-mode properties for building blocks.
A dual-mode oblivious transfer (in the CRS model) has two modes of generating a CRS and it satisfies statistical (indistinguishability-based) receiver's security in one mode and statistical (full-simulation-based) sender's security in the other mode.
The construction of \cite{SCN:Quach20} is an instantiation of a $1$-out-of-$2$ oblivious transfer with such a dual-mode property, and this can be converted into $5$-out-of-$N$ one as explained above. We stress again that it is important to relax the receiver's security to  the indistinguishability-based one to make the conversion work. 
A dual-mode commitment (in the CRS model) has two modes of generating a CRS and it is statistically binding in one mode and statistically hiding in the other mode.
We can use lossy encryption \cite{EC:BelHofYil09,JACM:Regev09} as an instantiation of such a dual-mode commitment.
Both of dual-mode $5$-out-of-$N$ oblivious transfer and lossy encryption are 
based on the LWE assumption (with super-polynomial modulus for the former) and fairly efficient in the sense that they do not rely on non-black-box techniques. 
Putting everything together, we obtain a fairly efficient (black-box) construction of a dual-mode NIZK for $\QMA$ in the \CRSVP model.

\paragraph{NIZK for $\QMA$ via Fiat-Shamir transformation.}
Finally, let us explain our construction of NIZK for $\QMA$ via the Fiat-Shamir transformation.
It is based on so-called the $\Xi$-protocol for $\QMA$ \cite{FOCS:BroGri20}, which is equal to the standard
$\Sigma$-protocol except that the first message is quantum.
Because the first message is quantum, the Fiat-Shamir technique cannot be directly applied.
Our idea is again to use the teleportation technique: if we introduce a setup that sends random Pauli basis states
to the prover and their classical description to the verifier, the first message can be classical.
We thus obtain a (classical) $\Sigma$-protocol in the QSP model,
where the trusted setup sends random Pauli basis states to the prover and their classical description
to the verifier.
This task can be, actually, done by the verifier, not the trusted setup,
unlike our first construction.
We therefore obtain a (classical) $\Sigma$-protocol with quantum preprocessing (\cref{def:sigma_q_prepro}),
where the verifier sends random Pauli basis states to the prover as the preprocessing.

We then apply the (classical) Fiat-Shamir transformation to the $\Sigma$-protocol
with quantum preprocessing, and obtain the 
CV-NIZK for $\QMA$ in the quantum random oracle plus $V\to P$ model (\cref{def:qrovp_nizk}), 
where $V\to P$ means the communication from the verifier to the prover as the preprocessing. 
Note that we are considering a classical $\Sigma$-protocol with quantum preprocessing differently from previous works. By a close inspection, we show that an existing security proof for classical $\Sigma$-protocol in the QROM \cite{C:DonFehMaj20} also works in our setting.

Importantly, in this case, unlike the previous two constructions, the quantum preprocessing can be replaced with 
the setup that distributes Bell pairs among the prover and the verifier.
As a corollary, we therefore obtain NIZK for $\QMA$ in the shared Bell pair model (plus quantum random
oracle).
The distribution of Bell pairs by a trusted setup
can be considered as a ``quantum analogue" of the CRS, and therefore 
we can say that we obtain NIZK for $\QMA$ in the ``quantum CRS" model via the Fiat-Sharmir transformation.

\subsection{Related Work}
\ifnum\submission=1
We have already compared our results with previous NIZKs for $\QMA$.
More related works are discussed in \cref{sec:omitted_related_works}.
\else

\paragraph{More related works on quantum NIZKs.}
Kobayashi \cite{Kobayashi03} studied (statistically sound and zero-knowledge) NIZKs in a model where the prover and verifier share Bell pairs, and gave a complete problem in this setting.
It is unlikely that the complete problem contains (even a subclass of) $\NP$ \cite{MW18} and thus even a NIZK for all $\NP$ languages is unlikely to exist in this model.
Note that if we consider the prover and verifier sharing Bell pairs in advance like this model,
the verifier's preprocessing message 
of our protocols (and the protocol of~\cite{C:ColVidZha20}) 
becomes
classical.
Chailloux et al. \cite{TCC:CCKV08} showed that there exists a (statistically sound and zero-knowledge) NIZK for all languages in $\mathbf{QSZK}$ in the help model where a trusted party generates a pure state \emph{depending on the statement to be proven} and gives copies of the state to both prover and verifier.

\paragraph{Interactive zero-knowledge for $\QMA$.}
There are several works of interactive zero-knowledge proofs/arguments for $\QMA$.
The advantage of these constructions compared to non-interactive ones is that they do not require any trusted setup. 
Broadbent, Ji, Song, and Watrous \cite{BJSW20} gave the  first construction of a zero-knowledge proof for $\QMA$.
Broadbent and Grilo \cite{FOCS:BroGri20} gave an alternative simpler construction.
Bitansky and Shmueli \cite{STOC:BitShm20} gave the first constant round zero-knowledge argument for $\QMA$ with negligible soundness error. 
Brakerski and Yuen \cite{BraYue} gave a construction of $3$-round \emph{delayed-input} zero-knowledge proof for $\QMA$ where  the prover needs to know the statement and witness only for generating its last message. 
By considering the first two rounds as preprocessing, we can view this construction as a NIZK in a certain kind of preprocessing model. 
However, their protocol has a constant soundness error, and it seems difficult to prove the zero-knowledge property for the parallel repetition version of it.


\fi

%% file: table.tex
\begin{table}[t]
\setlength\tabcolsep{0.5eM}
\begin{center}
\begin{minipage}[c]{\textwidth} \scriptsize
\begin{center}
\begin{threeparttable}
\caption{Comparison of NIZKs for $\QMA$. }
\label{tbl:compare_nizk}
\begin{tabular}{lllllll}
\toprule
\multicolumn{1}{c}{Reference} & \multicolumn{1}{c}{Soundness} &\multicolumn{1}{c}{ZK}
&\multicolumn{1}{c}{Verification}
&\multicolumn{1}{c}{Model} &\multicolumn{1}{c}{Assumption}
&\multicolumn{1}{c}{Misc}\\
\midrule
\cite{TCC:ACGH20}&comp. & comp. &classical& SP & LWE + QRO \\\hline
\cite{C:ColVidZha20}&comp.&comp.&quantum+classical&\CRSVP&LWE&AoQK\\\hline
\cite{FOCS:BroGri20}&stat. & stat. &quantum &SP &None \\\hline
\cite{C:Shmueli21}&comp.&comp.&quantum& MDV& LWE&reusable\\\hline
\cite{C:BCKM21a}&comp.&comp.&quantum& MDV& LWE&
\begin{tabular}{@{}l@{}} 
reusable and\\
single-witness
\end{tabular}
\\\hline
\cite{BartusekMalavolta}&comp.& stat. &classical& CRS & iO + QRO (heuristic) \\\hline
\cref{sec:CV-NIZK}& stat.&stat.&classical &QSP& None  \\\hline
\cref{sec:Dual-mode}& 
\begin{tabular}{@{}l@{}} 
stat.\\
comp.
\end{tabular}
&
\begin{tabular}{@{}l@{}} 
comp.\\
stat.
\end{tabular}
&quantum+classical &\CRSVP& LWE& dual-mode  \\\hline
\cref{sec:Fiat-Shamir}&comp.(query)&comp.(query)&classical&$V\to P$/Bell pair&QRO\\
\bottomrule\\
\end{tabular}
In column ``Soundness''  (resp. ``ZK''), stat., and comp. mean statistical, and computational soundness (resp. zero-knowledge), respectively. 
Also, comp.(query) means that only the number of queries should be polynomial.
In column ``Verification", ``quantum+classical" means that the verifier needs to send a quantum message in preprocessing but the online phase of verification is classical. 
QRO means the quantum random oracle.
\end{threeparttable}
\end{center}
\end{minipage}
\end{center}
\end{table}

%% file: Preliminaries.tex
\section{Preliminaries}\label{sec:preliminary}
\ifnum\submission=1
Basic notations used in this paper are given in \cref{sec:notations}, which are mostly standard in cryptography.
\else
\paragraph{Notations.}
We use $\secpar$ to denote the security parameter throughout the paper. 
For a positive integer $N$, $[N]$ means the set $\{1,2,...,N\}$.
For a probabilistic classical  or quantum algorithm $\A$, we denote by $y\sample \A(x)$ to mean $\A$ runs on input $x$ and outputs $y$. 
For a finite set $S$ of classical strings, $x\sample S$ means that $x$ is uniformly randomly chosen from $S$.
For a classical string $x$, $x_i$ denotes the $i$-th bit of $x$. 
For classical strings $x$ and $y$, $x\concat y$ denotes the concatenation of $x$ and $y$. 
We write $\poly$ to mean an unspecified polynomial and $\negl$ to mean an unspecified negligible function. 
We use PPT to stand for (classical) probabilistic polynomial time and QPT to stand for quantum polynomial time.
When we say that an algorithm is non-uniform QPT, it is expressed as a family of polynomial size quantum circuits with quantum advice. 
\fi

\subsection{Quantum Computation Preliminaries}\label{sec:quantum_preliminaries}
Here, we briefly review basic notations and facts on quantum computations.  

For any quantum state $\rho$ over registers $\regA$ and $\regB$,
$\Tr_{\regA}(\rho)$ is the partial trace of $\rho$ over $\regA$.
We use $I$ to mean the identity operator. 
(For simplicity, we use the same $I$ for all identity operators with different dimensions, because the
dimension of an identity operator is clear from the context.)
We use $X$, $Y$, and $Z$ to mean Pauli operators  
i.e., 
$  X := \left(
    \begin{array}{cc}
      0 & 1 \\
      1 & 0  
    \end{array}
  \right)$,
  $  Z := \left(
    \begin{array}{cc}
      1 & 0 \\
      0 & -1  
    \end{array}
  \right)$,
and $Y :=iXZ$.
We use $H$ to mean Hadamard operator, i.e., 
$  H := \frac{1}{\sqrt{2}}\left(
    \begin{array}{cc}
      1 & 1 \\
      1 & -1  
    \end{array}
  \right)$.
  We also define the $T$ operator by
 $  T := \left(
    \begin{array}{cc}
      1 & 0 \\
      0 & e^{i\pi/4}  
    \end{array}
  \right)$.
The 
$CNOT:=|0\rangle\langle0|\otimes I+|1\rangle\langle1|\otimes X$
is the controlled-NOT operator.
  
We define 
$V(Z):=I$, $V(X):=H$, and $V(Y):=
\frac{1}{\sqrt{2}}\left(
    \begin{array}{cc}
      1 & 1 \\
      i & -i  
    \end{array}
  \right)$ 
 so that for each $W\in \{X,Y,Z\}$, $V(W)\ket{0}$ and $V(W)\ket{1}$ are the eigenvectors of $W$ with eigenvalues 
 $+1$ and $-1$, respectively.
For each  $W\in \{X,Y,Z\}$, we call $\{V(W)\ket{0},V(W)\ket{1}\}$ the $W$-basis.


When we consider an $N$-qubit system, for a Pauli operator $Q\in \{X,Y,Z\}$, $Q_j$ denotes the operator that acts on $j$-th qubit as $Q$ and trivially acts on all the other qubits.
Similarly, $V_j(W)$ denotes the operator that acts on $j$-th qubit as $V(W)$ and trivially acts on all the other qubits.
For any $x\in \bit^N$ and $z\in \bit^N$, 
$X^xZ^z$ means $\prod_{j=1}^{N} X_j^{x_j}Z_j^{z_j}$. 

We call the state $\frac{1}{\sqrt{2}}\left(\ket{0}\otimes\ket{0}+\ket{1}\otimes\ket{1}\right)$ the Bell pair. We call the set $\{\ket{\phi_{x,z}}\}_{(x,z)\in \bit^{2}}$ the Bell basis where  
$
|\phi_{x,z}\rangle:= (X^xZ^z\otimes
I)\frac{\ket{0}\otimes\ket{0}+\ket{1}\otimes\ket{1}}{\sqrt{2}}.
$
Let us define
$U(X):=V(X)$,
$U(Y):=V(Y)X$,
and
$U(Z):=V(Z)$.

\begin{lemma}[State Collapsing]\label{lem:statecollapsing}
If we project one qubit of a Bell pair onto $V(W)|m\rangle$ with $W\in\{X,Y,Z\}$ and $m\in\{0,1\}$, 
the other qubit collapses to $U(W)|m\rangle$.
\end{lemma}
\begin{lemma}[Effect of $X^xZ^z$ before measurement]\label{lem:XZ_before_measurement}
For any $N$-qubit state $\rho$, $(W_1,...,W_N)\in \{X,Y,Z\}^{N}$, and $(x,z)\in \bit^{N}\times\bit^{N}$, the distributions of $(m'_1,...m'_n)$ sampled in the following two ways are identical.
\begin{enumerate}
    \item 
    For $j\in[N]$, measure $j$-th qubit of $\rho$ in $W_j$ basis, let $m_j\in\{0,1\}$ be the outcome, and
    set 
    \begin{eqnarray*}
m_{j}':= 
\left\{
\begin{array}{cc}
m_{j}\oplus x_{j}&(W_j=Z),\\
m_{j}\oplus z_{j}&(W_j=X),\\
m_{j}\oplus x_{j}\oplus z_j&(W_j=Y).
\end{array}
\right.
\end{eqnarray*}
\item For $j\in[N]$, measure $j$-th qubit of $X^xZ^z\rho Z^zX^x$ in $W_j$ basis and let $m'_j\in\{0,1\}$ be the outcome.
\end{enumerate}
\end{lemma}
The proofs of the above lemmas are straightforward.
\ifnum\submission=1
The following lemma is implicit in previous works, e.g., \cite{MNS,posthoc}. A proof is given in \cref{sec:proof_prob_and_energy}. 
\begin{lemma}\label{lem:prob_and_energy}
Let  
$
\ham:=\frac{1}{2}\left[I+s(\prod_{j\in S_X} X_j)(\prod_{j\in S_Y} Y_j)(\prod_{j\in S_Z}Z_j)\right]
$
be an $N$-qubit projection operator, where $s\in\{+1,-1\}$, 
and $S_X$, $S_Y$, and $S_Z$ are disjoint subsets of $[N]$. 
For any $N$-qubit quantum state $\rho$, suppose that 
for all $j\in S_W$, where $W\in\{X,Y,Z\}$, we measure $j$-th qubit of $\rho$ in the $W$-basis,
and let $m_j\in\{0,1\}$ be the outcome. 
Then we have 
$
\Pr\left[(-1)^{\bigoplus_{j\in S_X \cup S_Y\cup S_Z}m_j}=-s\right]=1-\Tr(\rho \ham).
$
\end{lemma}
Some other basic lemmas that are used in our security proofs are given in \cref{sec:basic_lemmas}.  
\else
\begin{lemma}[Pauli Mixing]\label{lem:Pauli_mixing}
Let $\rho$ be an arbitrary quantum state over registers $\regA$ and $\regB$, and let $N$ be the number of qubits in $\regA$. 
Then we have 
\[
\frac{1}{2^{2N}}
\sum_{x\in \bit^{N},z_\in \bit^{N}}
\left(X^xZ^z \otimes I_\regB\right)\rho \left(Z^zX^x \otimes I_\regB\right)=\frac{1}{2^{N}} I_\regA \otimes \Tr_{\regA}(\rho).
\]
\end{lemma}

This is well-known, and one can find a proof in e.g., \cite{FOCS:Mahadev18b}.

\begin{lemma}[Quantum Teleportation]\label{lem:teleportation}
Suppose that we have $N$ Bell pairs between registers $\regA$ and $\regB$, i.e., $\frac{1}{2^{N/2}}\sum_{s\in \bit^{N}}\ket{s}_\regA\otimes \ket{s}_\regB$, and let $\rho$ be an arbitrary $N$-qubit quantum state in register $\regC$. 
Suppose that we measure $j$-th qubits of $\regC$ and $\regA$ in the Bell basis 
and let $(x_j,z_j)$ be the measurement outcome 
for all $j\in [N]$. 
Let $x:=x_1\concat x_2 \concat...\concat x_N$ and $z:=z_1\concat z_2 \concat...\concat z_N$. 
Then the $(x,z)$ is uniformly distributed over $\bit^N \times \bit^N$.
Moreover, conditioned on the measurement outcome $(x,z)$, the resulting state in $\regB$ is
$X^{x}Z^{z}\rho Z^{z}X^{x}$.
\end{lemma}
This is also well-known, and one can find a proof in e.g., \cite{NC00}.

The following lemma is implicit in previous works e.g., \cite{MNS,posthoc}. 

\begin{lemma}\label{lem:prob_and_energy}
Let  
\[
\ham:=\frac{I+s(\prod_{j\in S_X} X_j)(\prod_{j\in S_Y} Y_j)(\prod_{j\in S_Z}Z_j)}{2}
\]
be an $N$-qubit projection operator, where $s\in\{+1,-1\}$, 
and $S_X$, $S_Y$, and $S_Z$ are disjoint subsets of $[N]$. 
For any $N$-qubit quantum state $\rho$, suppose that 
for all $j\in S_W$, where $W\in\{X,Y,Z\}$, we measure $j$-th qubit of $\rho$ in the $W$-basis,
and let $m_j\in\{0,1\}$ be the outcome. 
Then we have 
\[
\Pr\left[(-1)^{\bigoplus_{j\in S_X \cup S_Y\cup S_Z}m_j}=-s\right]=1-\Tr(\rho \ham).
\]
\end{lemma}

\begin{proof}[Proof of \cref{lem:prob_and_energy}]
Let us define 
$V:=
(\prod_{j\in S_X}V_j(X))
(\prod_{j\in S_Y}V_j(Y))
(\prod_{j\in S_Z}V_j(Z))$, 
and $|m\rangle:=\bigotimes_{j=1}^N|m_j\rangle$.
Then,
\begin{eqnarray*}
\Pr\left[(-1)^{\bigoplus_{j\in S_X \cup S_Y\cup S_Z}m_j}=-s\right]
&=&\sum_{m\in\{0,1\}^N}\langle m|V^\dagger\rho V|m\rangle
\frac{1-s(-1)^{\bigoplus_{j\in S_X\cup S_Y\cup S_Z}m_j}}{2}\\
&=&\sum_{m\in\{0,1\}^N}\langle m|V^\dagger\rho V
\frac{I-s\prod_{j\in S_X\cup S_Y\cup S_Z}Z_j}{2}
|m\rangle\\
&=&\Tr\Big[V^\dagger\rho V
\frac{I-s\prod_{j\in S_X\cup S_Y\cup S_Z}Z_j}{2}\Big]\\
&=&\Tr\Big[\rho V
\frac{I-s\prod_{j\in S_X\cup S_Y\cup S_Z}Z_j}{2}V^\dagger\Big]\\
&=&\Tr\Big[\rho (I-\ham)\Big]\\
&=&1-\Tr(\rho \ham).
\end{eqnarray*}
\end{proof}
\fi

\subsection{$\QMA$ and Local Hamiltonian Problem}

\ifnum\submission=1
The definition of $\QMA$ is given in \cref{def:QMA} in \cref{sec:appendix_QMA}. 
For any $\QMA$ promise problem $L=(L_\yes,L_\no)$ 
and $\statement \in L_\yes$, we denote by $R_L(\statement)$ to mean the (possibly infinite) set of all quantum states 
$\witness$ such that $\Pr[V(\statement,\witness)=1]\geq 2/3$.
\else
\begin{definition}[$\QMA$]\label{def:QMA}
We say that a promise problem $L=(L_\yes,L_\no)$ is in $\QMA$ if there is 
a polynomial $\ell$ and a QPT algorithm $V$ such that the following is satisfied: 
\begin{itemize}
    \item For any $\statement \in L_\yes $, there exists a quantum state $\witness$ of $\ell(|\statement|)$-qubit (called a witness) such that we have $\Pr[V(\statement,\witness)=1]\geq 2/3$.
    \item For any $\statement\in L_\no$ and any quantum state $\witness$ of $\ell(|\statement|)$-qubit, we have $\Pr[V(\statement,\witness)=1]\leq 1/3$.
\end{itemize}
For any $\statement \in L$, we denote by $R_L(\statement)$ to mean the (possibly infinite) set of all quantum states 
$\witness$ such that $\Pr[V(\statement,\witness)=1]\geq 2/3$.
\end{definition}

\fi

Recently, Broadbent and Grilo \cite{FOCS:BroGri20} showed that any $\QMA$ problem can be reduced to a $5$-local Hamiltonian problem with \emph{local simulatability}.
(See also \cite{FOCS:GriSloYue19}.)
Moreover, it is easy to see that we can make the Hamiltonian $\ham_\statement$ be of the form 
$\ham_\statement=\sum_{i=1}^Mp_i \frac{I+s_i P_i}{2}$ where $s_i\in\{+1,-1\}$, $p_i\ge0$, $\sum_{i=1}^Mp_i=1$,
and $P_i$ is a tensor product of Pauli operators $(I,X,Z,Y)$
with at most $5$ nontrivial Pauli operators $(X,Y,Z)$.
See \cref{sec:BG20} for more details. 
Then we have the following lemma.

\begin{lemma}[$\QMA$-completeness of 5-local Hamiltonian problem with local simulatability \cite{FOCS:BroGri20}]\label{lem:five_local_Hamiltonian}
For any $\QMA$ promise problem $L=(L_\yes,L_\no)$, there is a classical polynomial-time computable deterministic function that maps 
$\statement\in \bit^*$ to an $N$-qubit Hamiltonian $\ham_\statement$ of the form  
$
\ham_\statement=\sum_{i=1}^M p_i \frac{I+s_i P_i}{2},
$
where $N=\poly(|\statement|)$, $M=\poly(|\statement|)$, $s_i\in\{+1,-1\}$, $p_i> 0$, $\sum_{i=1}^M p_i=1$, and
$P_i$ is a tensor product of Pauli operators $(I,X,Y,Z)$
with at most $5$ nontrivial Pauli operators $(X,Y,Z)$, and satisfies the following:
There are $0<\alpha<\beta<1$ such that $\beta-\alpha = 1/\poly(|\statement|)$ and
\begin{itemize}
    \item if $\statement \in L_\yes$, then there exists an $N$-qubit state $\rho$ such that $\Tr(\rho \ham_\statement)\leq \alpha$, and
    \item if $\statement \in L_\no$, then for any $N$-qubit state $\rho$, we have $\Tr(\rho \ham_\statement)\geq \beta$.
\end{itemize}
Moreover, for any $\statement \in L_\yes$, 
we can convert any witness $\witness\in R_L(\statement)$ into a state $\rho_{\hist}$, called the history state, such that $\Tr(\rho_\hist \ham_\statement)\leq \alpha$ in quantum polynomial time.
Moreover, 
there exists a classical deterministic polynomial time algorithm $\siml_{\hist}$ such that for 
any $\statement \in L_\yes$ and
any subset $S\subseteq [N]$ with $|S|\leq 5$,  $\siml_{\hist}(\statement,S)$ outputs a classical description of an $|S|$-qubit density matrix $\rho_S$ such that $\|\rho_S-\Tr_{[N]\setminus S}\rho_\hist\|_{tr}=\negl(\secpar)$
where $\Tr_{[N]\setminus S}\rho_\hist$ is the state of $\rho_\hist$ in registers corresponding to $S$ tracing out all other registers.
\end{lemma}

\subsection{Classically-Verifiable Non-Interactive Zero-knowledge Proofs}\label{sec:def_CV-NIZK}

\begin{definition}[CV-NIZK in the QSP model]\label{def:CV-NIZK}
A classically-verifiable non-interactive zero-knowledge proof (CV-NIZK) for a $\QMA$ promise
problem $L=(L_\yes,L_\no)$ in the quantum secret parameter (QSP) model consists of algorithms $\Pi=(\setup,\prove,\verify)$ with the following syntax:
\begin{description}
\item[$\setup(1^\secpar)$:] This is a QPT algorithm that takes the security parameter $1^\secpar$ as input and outputs a quantum proving key $\pkey$ and a classical verification key $\vkey$. 
\item[$\prove(\pkey,\statement,\witness^{\otimes k})$:] This is a QPT algorithm that takes the proving key $\pkey$, a statement $\statement$, and $k=\poly(\secpar)$ copies $\witness^{\otimes k}$ of a witness $\witness\in R_L(\statement)$ as input and outputs a classical proof $\pi$.
\item[$\verify(\vkey,\statement,\pi)$:]
This is a PPT algorithm that takes the verification key $\vkey$, a statement $\statement$, and a proof $\pi$ as input and outputs $\top$ indicating acceptance or $\bot$ indicating rejection. 
\end{description}
We require $\Pi$ to satisfy the following properties for some $0<s<c<1$ such that $c-s>1/\poly(\secpar)$.
Especially, when we do not specify $c$ and $s$, they are set as $c=1-\negl(\secpar)$ and $s=\negl(\secpar)$. 

\medskip
\noindent \underline{
\textbf{$c$-Completeness.}
}
For all $\statement\in L_\yes\cap \bit^\secpar$, and $\witness\in R_L(\statement)$, we have 
\begin{align*}
    \Pr\left[
    \verify(\vkey,\statement,\pi)=\top 
    :(\pkey,\vkey)\sample \setup(1^\secpar),\pi \sample \prove(\pkey,\statement,\witness^{\otimes k})\right]\geq c.
\end{align*}
\medskip
\noindent \underline{
\textbf{(Adaptive Statistical) $s$-Soundness.}
}
For all unbounded-time adversary $\A$, we have 
\begin{align*}
    \Pr\left[
    \statement \in L_\no \land \verify(\vkey,\statement,\pi)=\top 
    :(\pkey,\vkey)\sample \setup(1^\secpar),(\statement,\pi) \sample \A(\pkey)\right]\leq s.
\end{align*}

\noindent \underline{
\textbf{(Adaptive Statistical Single-Theorem) Zero-Knowledge.}
}
There exists a PPT simulator $\siml$ such that for any  unbounded-time distinguisher $\dist$, we have 
\begin{align*}
    \left|\Pr\left[\dist^{\ora_P(\pkey,\cdot,\cdot)}(\vkey)=1\right]-   \Pr\left[\dist^{\ora_S(\vkey,\cdot,\cdot)}(\vkey)=1\right]\right|=\negl(\secpar)
\end{align*}
where $(\pkey,\vkey)\sample \setup(1^\secpar)$, 
$\dist$ can make at most one query, which should be of the form $(\statement,\witness^{\otimes k})$ where $\witness\in R_L(\statement)$ and $\witness^{\otimes k}$ is unentangled with $\dist$'s internal registers,\footnote{
Though our protocols are likely to remain secure even if they can be entangled, we assume that they are unentangled for simplicity.  
To the best of our knowledge, none of existing works on interactive or non-interactive zero-knowledge for $\QMA$ \cite{BJSW20,C:ColVidZha20,STOC:BitShm20,FOCS:BroGri20,C:Shmueli21,C:BCKM21a} considered entanglement between a witness and distinguisher's internal register. \label{footnote:unentangled}}  $\ora_P(\pkey,\statement,\witness^{\otimes k})$ returns $\prove(\pkey,\statement,\witness^{\otimes k})$, and
  $\ora_S(\vkey,\statement,\witness^{\otimes k})$ returns $\siml(\vkey,\statement)$.


\end{definition}

It is easy to see that we can amplify the gap between completeness and soundness thresholds by a simple parallel repetition.
Moreover, we can see that this does not lose the zero-knowledge property.
Therefore, we have the following lemma. 
\begin{lemma}[Gap Amplification for CV-NIZK]\label{lem:CV-NIZK_amplification}
If there exists a CV-NIZK for $L$ in the QSP model that satisfies $c$-completeness and $s$-soundness, for some $0<s<c<1$ such that $c-s>1/\poly(\secpar)$, then there exists a CV-NIZK for $L$ in the QSP model (with $(1-\negl(\secpar))$-completeness and $\negl(\secpar)$-soundness).
\end{lemma}
\ifnum\submission=1
A proof of \cref{lem:CV-NIZK_amplification} is given in \cref{sec:proof_amplification}.
\else
\begin{proof}
Let $\Pi=(\setup,\prove,\verify)$ be a CV-NIZK for $L$ in the SP model that satisfies $c$-completeness, $s$-soundness, and the zero-knowledge property for  some $0<s<c<1$ such that $c-s>1/\poly(\secpar)$.
Let $k$ be the number of copies of a witness $\prove$ takes as input.  
For any polynomial $N=\poly(\secpar)$, $\Pi^{N}=(\setup^{N},\prove^{N},\verify^{N})$ be the $N$-parallel version of $\Pi$. That is, $\setup^{N}$ and $\prove^{N}$ run $\setup$ and $\prove$ $N$ times parallelly and outputs tuples consisting of outputs of each execution, respectively where $\prove^{N}$ takes $Nk$ copies of the witness as input.
$\verify^{N}$ takes $N$-tuple of the verification key and proof, runs $\verify$ to verify each of them separately, and outputs $\top$ if the number of executions of $\verify$ that outputs $\top$ is larger than $\frac{N(\alpha+\beta)}{2}$.
By Hoeffding's inequality, it is easy to see that we can take $N=O\left(\frac{\log^{2} \secpar}{(\alpha-\beta)^2}\right)$ so that $\Pi^{N}$ satisfies $(1-\negl(\secpar))$-completeness and $\negl(\secpar)$-soundness. 

What is left is to prove that $\Pi^{N}$ satisfies the zero-knowledge property.
This can be reduced to the zero-knowledge property of $\Pi$ by a standard hybrid argument.
More precisely, for each $i\in \{0,...,N\}$, let $\ora_i$ be the oracle that works as follows where $\pkey'$ and $\vkey'$ denote the proving and verification keys of $\Pi^{N}$, respectively. 
\begin{description}
\item[$\ora_i(\pkey'=(\pkey^1,...,\pkey^N),\vkey'=(\vkey^1,...,\vkey^N),\statement,\witness^{\otimes Nk})$:]
It works as follows:
\begin{itemize}
    \item For $1\leq j\leq i$, it computes  $\pi_j\sample \siml(\vkey^j,\statement)$.
    \item For $i<j\leq N$,  it computes $\pi_j\sample \prove(\pkey^j,\statement,\witness^{\otimes k})$ where it uses the $(k(j-1)+1)$-th to $kj$-th copies of $\witness$.
    \item Output $\pi:=(\pi_1,...,\pi_N)$.
\end{itemize}
\end{description}
Clearly, we have $\ora_0(\pkey',\vkey',\cdot,\cdot)=\ora_P(\pkey',\cdot,\cdot)$  
and $\ora_N(\pkey',\vkey',\cdot,\cdot)=\ora_S(\vkey',\cdot,\cdot)$.\footnote{$\ora_P(\pkey',\cdot,\cdot)$ and $\ora_S(\vkey',\cdot,\cdot)$ mean the corresponding oracles for $\Pi^{N}$.
} 
Therefore, it suffices to prove that no  distinguisher can distinguish $\ora_i(\pkey',\vkey',\cdot,\cdot)$ 
and 
$\ora_{i+1}(\pkey',\vkey',\cdot,\cdot)$
for any $i\in \{0,1,...,N-1\}$. 
For the sake of contradiction, suppose that there exists a distinguisher $\dist'$ that 
distinguishes $\ora_i(\pkey',\vkey',\cdot,\cdot)$ and $\ora_{i+1}(\pkey',\vkey',\cdot,\cdot)$ with a non-negligible advantage by making one query of the form $(\statement,\witness^{\otimes Nk})$.
Then we construct a distinguisher $\dist$ that breaks the zero-knowledge property of $\Pi$ as follows:
\begin{itemize}
    \item[$\dist^{\ora}(\vkey)$:]
    $\dist$ takes $\vkey$ as input and is given a single oracle access to $\ora$, which is either 
    $\ora_P(\pkey,\cdot,\cdot)$
    or  $\ora_S(\vkey,\cdot,\cdot)$ where $\pkey$ is the proving key corresponding to $\vkey$.\footnote{ $\ora_P(\pkey,\cdot,\cdot)$ and $\ora_S(\vkey,\cdot,\cdot)$ mean the corresponding oracles for $\Pi$ by abuse of notation.
}  (Remark that $\dist$ is not given $\pkey$.)
    It sets $\vkey^{i+1}:=\vkey$ (which implicitly defines $\pkey^{i+1}:=\pkey$) and generates $(\pkey^j,\vkey^j)\sample \setup(1^\secpar)$ for all $j\in [N]\setminus \{i+1\}$.
    It sets $\vkey':=(\vkey^1,...,\vkey^N)$ and runs $\dist'^{\ora'}(\vkey')$ where when $\dist'$ makes a query $(\statement,\witness^{\otimes Nk})$ to $\ora'$, $\dist$ simulates the oracle $\ora'$ for $\dist'$ as follows:
    \begin{itemize}
          \item For $1\leq j\leq i$, $\dist$ computes  $\pi_j\sample \siml(\vkey^j,\statement)$.
    \item For $j=i+1$, $\dist$ queries $(\statement, \witness^{\otimes k})$ to the external oracle $\ora$ where it uses the $(ki+1)$-th to $k(i+1)$-th copies of $\witness$ as part of its query, and lets $\pi_{i+1}$ be the oracle's response.
    \item For $i+1<j\leq N$,  it computes $\pi_j\sample \prove(\pkey^j,\statement,\witness^{\otimes k})$ where it uses the $(k(j-1)+1)$-th to $kj$-th copies of $\witness$.
    We note that this can be simulated by $\dist$ since it knows $\pkey^j$  for $j\neq i+1$. 
    \item $\dist$ returns $\pi':=(\pi_1,...,\pi_N)$ to $\dist'$ as a response from the oracle $\ora'$.
    \end{itemize}
    Finally, when $\dist'$ outputs $b$, $\dist$ also outputs $b$.
\end{itemize}
We can see that the oracle $\ora'$ simualted by $\dist$ works similarly to $\ora_i(\pkey',\vkey',\cdot,\cdot)$ 
when $\ora$ is $\ora_P(\pkey,\cdot,\cdot)$
and works similarly to $\ora_{i+1}(\pkey',\vkey',\cdot,\cdot)$ when $\ora$ is $\ora_S(\vkey,\cdot,\cdot)$ where $\pkey'=(\pkey^1,...,\pkey^N)$.
Therefore, by the assumption that $\dist'$ distinguishes $\ora_i(\pkey',\vkey',\cdot,\cdot)$ and $\ora_{i+1}(\pkey',\vkey',\cdot,\cdot)$ with a non-negligible advantage, $\dist$ distinguishes $\ora_P(\pkey,\cdot,\cdot)$ and $\ora_S(\vkey,\cdot,\cdot)$ with a non-negligible advantage.
However, this contradicts the zero-knowledge property of $\Pi$.
Therefore, such $\dist'$ does not exist, which completes the proof of \cref{lem:CV-NIZK_amplification}. 
\end{proof}
\fi

%% file: NIZK_information_theoretical.tex
\section{CV-NIZK in the QSP model}\label{sec:CV-NIZK}
In this section, we construct a CV-NIZK in the QSP model (\cref{def:CV-NIZK}). 
Specifically, we prove the following theorem.
\begin{theorem}\label{thm:CV-NIZK}
There exists a CV-NIZK for $\QMA$ in the QSP model (without any computational assumption). 
\end{theorem}

Our construction of a CV-NIZK for a $\QMA$ promise problem $L$ is given in \cref{fig:CV-NIZK}
where $\ham_\statement$, $N$, $M$, $p_{i}$, $s_{i}$, $P_i$, $\alpha$, $\beta$, and $\rho_\hist$ are as in \cref{lem:five_local_Hamiltonian} for $L$ and $V_j(W_j)$ is as defined in \cref{sec:quantum_preliminaries}. 

We note that there is a slightly simpler construction of CV-NIZK as shown in \cref{fig:CV-NIZK_prime} in \cref{sec:alternative_NIZK}.
However, 
we consider the construction given in \cref{fig:CV-NIZK} as our main construction 
since this is more convenient to extend to the computationally secure construction given in \cref{sec:Dual-mode}.  

Moreover, if we require only the completeness and the soundness, there is a much simpler construction.
For details, see \cref{sec:CV-NIP}.

\begin{figure}[t]
\rule[1ex]{\textwidth}{0.5pt}
\begin{description}
\item[$\setup(1^\secpar)$:]
The setup algorithm chooses
$(W_1,...,W_N)\sample \{X,Y,Z\}^{N}$, $(m_1,...,m_N)\sample \{0,1\}^{N}$,
$(\whx,\whz)\sample \bit^{N}\times \bit^{N}$,  
and a uniformly random subset $S_V\subseteq [N]$ such that $1\le|S_V|\le5$, 
and outputs a proving key $\pkey:=\left(\rho_P:=\bigotimes_{j=1}^N(U(W_j)|m_j\rangle),\whx,\whz\right)$ 
and a verification key $\vkey:=(
W_1,...,W_N,m_1,...,m_N,S_V, 
\{\whx_j,\whz_j\}_{j\in S_V})$.

\item[$\prove(\pkey,\statement,\witness)$:]
The proving algorithm 
parses $\left(\rho_P,\whx,\whz\right)\la \pkey$, 
generates the history state $\rho_\hist$ for $\ham_\statement$ from $\witness$, 
and computes $\rho'_\hist:=X^{\whx} Z^{\whz} \rho_\hist Z^{\whz} X^{\whx}$.
It measures $j$-th qubits of $\rho'_\hist$ and $\rho_P$ in the Bell basis for $j\in [N]$.    
Let $x:=x_1\concat x_2\concat...\concat x_N$, and $z:=z_1\concat z_2\concat...\concat z_N$ where 
$(x_j,z_j)\in\{0,1\}^2$ denotes the outcome of $j$-th measurement.
It outputs a proof $\pi:=(x,z)$.

\item[$\verify(\vkey,\statement,\pi)$:]
The verification algorithm 
parses 
$(W_1,...,W_N,m_1,...,m_N,S_V, 
\{\whx_j,\whz_j\}_{j\in S_V})\la \vkey$ and 
$(x,z)\la \pi$,   
chooses $i\in [M]$ according to the probability distribution defined by $\{p_{i}\}_{i\in[M]}$ (i.e., chooses $i$ with probability $p_{i}$).
Let
\begin{eqnarray*}
S_i:=\{j\in[N]~|~\mbox{$j$th Pauli operator of $P_i$ is not $I$}\}.
\end{eqnarray*}
We note that we have $1\leq |S_i|\leq 5$ by the $5$-locality of $\ham_\statement$. 
We say that $P_i$ is consistent to $(S_V,\{W_j\}_{j\in S_V})$  if and only if 
$S_i = S_V$ and
the $j$th Pauli operator of
$P_i$ is $W_j$ for all
$j\in S_i$.    
If $P_i$ is not consistent to $(S_V,\{W_j\}_{j\in S_V})$,  it outputs $\top$.
If $P_i$ is consistent to 
$(S_V,\{W_j\}_{j\in S_V})$,  
it flips a biased coin that heads with probability $1-3^{|S_i|-5}$.
If heads, it outputs $\top$.
If tails,
it defines
\begin{eqnarray*}
m_{j}':= 
\left\{
\begin{array}{cc}
m_{j}\oplus x_{j}\oplus \hat{x}_j&(W_j=Z),\\
m_{j}\oplus z_{j}\oplus \hat{z}_j&(W_j=X),\\
m_{j}\oplus x_{j}\oplus \hat{x}_j\oplus z_j\oplus \hat{z}_j&(W_j=Y)
\end{array}
\right.
\end{eqnarray*}
for $j\in S_i$, and outputs $\top$ if  $(-1)^{\bigoplus_{j\in S_i}m'_{j}}=-s_{i}$ and $\bot$ otherwise. 
\end{description} 
\rule[1ex]{\textwidth}{0.5pt}
\hspace{-10mm}
\caption{CV-NIZK $\Pi_\NIZK$ in the QSP model.}
\label{fig:CV-NIZK}
\end{figure}

\begin{figure}[t]
\rule[1ex]{\textwidth}{0.5pt}
\begin{description}
\item[$\setup_{\virone}(1^\secpar)$:]
The setup algorithm generates $N$ Bell-pairs between registers $\regP$ and $\regV$ and lets $\rho_P$ and $\rho_V$ be quantum states in registers $\regP$ and $\regV$, respectively. 
It chooses  $(\whx,\whz)\sample \bit^{N}\times \bit^{N}$.  
It chooses a uniformly random subset $S_V\subseteq [N]$ such that $1\le|S_V|\leq 5$, 
and outputs a proving key $\pkey:=\left(\rho_P, \whx,\whz\right)$ 
and a verification key $\vkey:=(\rho_V,S_V,\whx,\whz)$.
\item[$\prove_{\virone}(\pkey,\statement,\witness)$:]
This is the same  as $\prove(\pkey,\statement,\witness)$ in \cref{fig:CV-NIZK}.
\item[$\verify_{\virone}(\vkey,\statement,\pi)$:]
The verification algorithm chooses $(W_1,...,W_N)\sample \{X,Y,Z\}^N$, and measures $j$-th qubit of $\rho_V$ in the $W_j$ basis for all $j\in [N]$, and lets $(m_1,...,m_N)$ be the measurement outcomes. 
The rest of this algorithm is the same as $\verify(\vkey,\statement,\pi)$ given in \cref{fig:CV-NIZK}.
\end{description} 
\rule[1ex]{\textwidth}{0.5pt}
\hspace{-10mm}
\caption{The virtual protocol 1 for $\Pi_{\NIZK}$}
\label{fig:virtual1_CV-NIZK}
\end{figure}

\begin{figure}[t]
\rule[1ex]{\textwidth}{0.5pt}
\begin{description}
\item[$\setup_{\virtwo}(1^\secpar)$:]
This is the same as $\setup_{\virone}(1^\secpar)$ in \cref{fig:virtual1_CV-NIZK}.
\item[$\prove_{\virtwo}(\pkey,\statement,\witness)$:]
This is the same  as $\prove(\pkey,\statement,\witness)$ in \cref{fig:CV-NIZK}.
\item[$\verify_{\virtwo}(\vkey,\statement,\pi)$:]
The verification algorithm 
parses $(\rho_V,S_V, 
\whx,\whz)\la \vkey$ and
$(x,z)\la \pi$, computes $\rho'_V:=X^{x\oplus \whx}Z^{z\oplus \whz}\rho_V Z^{z\oplus \whz}X^{x\oplus \whx}$, 
chooses $(W_1,...,W_N)\sample \{X,Y,Z\}^N$, measures $j$-th qubit of $\rho'_V$ in the $W_j$ basis  for all $j\in[N]$, and lets $(m'_1,...,m'_N)$ be the measurement outcomes. 

It chooses $i\in[M]$ and defines $S_i\subseteq [N]$ similarly to $\verify(\vkey,\statement,\pi)$ in \cref{fig:CV-NIZK}.
If $P_i$ is not consistent to $(S_V,\{W_j\}_{j\in S_V})$,  it outputs $\top$.
If $P_i$ is consistent to 
$(S_V,\{W_j\}_{j\in S_V})$,  
it flips a biased coin that heads with probability $1-3^{|S_i|-5}$.
If heads, it outputs $\top$.
If tails,
it  outputs $\top$ if  $(-1)^{\bigoplus_{j\in S_i}m'_{j}}=-s_{i}$ and $\bot$ otherwise. 
\end{description} 
\rule[1ex]{\textwidth}{0.5pt}
\hspace{-10mm}
\caption{The virtual protocol 2 for $\Pi_{\NIZK}$}
\label{fig:virtual2_CV-NIZK}
\end{figure}

To show \cref{thm:CV-NIZK}, we prove the following lemmas.
\begin{lemma}[Completeness and Soundness]\label{lem:NIZK_completeness_soundness}
$\Pi_\NIZK$ satisfies $\left(1-\frac{\alpha}{N'}\right)$-completeness and 
 $\left(1-\frac{\beta}{N'}\right)$-soundness
 where $N':=3^5\sum_{i=1}^5{N \choose i}$. 
\end{lemma}
\begin{lemma}[Zero-Knowledge]\label{lem:NIZK_zero-knowledge}
$\Pi_\NIZK$ satisfies the zero-knowledge property. 
\end{lemma}

Since $\left(1-\frac{\alpha}{N'}\right)-\left(1-\frac{\beta}{N'}\right)=\frac{\beta-\alpha}{N'}\geq 1/\poly(\secpar)$, 
by combining \cref{lem:CV-NIZK_amplification,lem:NIZK_completeness_soundness,lem:NIZK_zero-knowledge}, \cref{thm:CV-NIZK} follows.

In the following, we give proofs of \cref{lem:NIZK_completeness_soundness,lem:NIZK_zero-knowledge}.
\begin{proof}[Proof of \cref{lem:NIZK_completeness_soundness}]
We prove this lemma by considering virtual protocols that do not change completeness and soundness.
For more details, see \cref{sec:alternative_proof_NIZK}. 
First, we consider  the virtual protocol 1 described in \cref{fig:virtual1_CV-NIZK}.
There are two differences from the original protocol.
The first is that $\vkey$ includes the whole $(\whx,\whz)$ instead of $\{\whx_j,\whz_j\}_{j\in S_V}$.
This difference does not change the (possibly malicious) prover's view since $\vkey$ is not given to the prover. 
The second is that the setup algorithm generates $N$ Bell pairs and gives each halves to the prover and verifier, and the verifier obtains $(m_1,...,m_N)$ by measuring his halves in Pauli basis. 
Because the verifier's measurement and the prover's measurement
commute with each other, 
in the virtual protocol 1, 
the verifier's acceptance probability does not change even if the verifier chooses $(W_1,...,W_N)$ and measures $\rho_V$ in the corresponding basis to obtain outcomes $(m_1,...,m_N)$ before $\rho_P$ is given to the prover.
Moreover, conditioned on the above measurement outcomes, the state in $\regP$ collapses to $\bigotimes_{j=1}^N(U(W_j)|m_j\rangle)$ (See ~\cref{lem:statecollapsing}). 
Therefore, the virtual protocol 1 is exactly the same as the original protocol from the prover's view, and the verifier's acceptance probability of the virtual protocol 1
is the same as that of the original protocol $\Pi_\NIZK$ for any possibly malicious prover.

Next, we further modify the protocol to define the virtual protocol 2 described in \cref{fig:virtual2_CV-NIZK}.
The difference from the virtual protocol 1 is that instead of setting $m_{j}'$, the verification algorithm applies a corresponding Pauli $X^{x\oplus \whx}Z^{z\oplus \whz}$ on $\rho_V$, and then measures it to obtain $m_{j}'$. 
By \cref{lem:XZ_before_measurement}, 
this does not change the distribution of $(m_1',...,m_N')$.
Therefore, the verifier's acceptance probability of the virtual protocol 2
is the same as that of the virtual protocol 1 for any possibly malicious prover.

Therefore, it suffices to prove $(1-\frac{\alpha}{N'})$-completeness and 
$(1-\frac{\beta}{N'})$-soundness for the virtual protocol $2$. 
When $\statement\in L_\yes$ and $\pi$ is honestly generated, then $\rho'_V$  is the history state $\rho_\hist$, which satisfies $\Tr(\rho_\hist \ham_\statement)\leq \alpha$, by the correctness of quantum teleportation (Lemma \ref{lem:teleportation}). 
For any fixed $P_i$, the probability that $P_i$ is consistent to  $(S_V,\{W_j\}_{j\in S_V})$ and
the coin tails is $\frac{1}{N'}$.
Therefore, 
by \cref{lem:prob_and_energy} and \cref{lem:five_local_Hamiltonian}, 
the verifier's acceptance probability is $1-\frac{1}{N'}\Tr(\rho_\hist \ham_\statement)\geq 1-\frac{\alpha}{N'}$. 

Let $\A$ be an adaptive adversary against soundness of virtual protocol $2$.
That is, $\A$ is given $\pkey$ and outputs $(\statement,\pi)$.
We say that $\A$ wins if $\statement \in L_\no$ and $\verify(\vkey,\statement,\pi)=\top$.
For any $\statement$, let 
$\event_\statement$ be the event that the statement output by $\A$ is $\statement$, and
$\rho'_{V,\statement}$ be the state in $\regV$ right before the measurement by $\verify$ conditioned on $\event_\statement$.
Similarly to the analysis for the completeness, 
by \cref{lem:prob_and_energy} and \cref{lem:five_local_Hamiltonian}, we have 
\begin{align*}
    \Pr[\A\text{~wins}]=\sum_{\statement\in L_\no}\Pr[\event_\statement]\left(1-\frac{1}{N'}\Tr(\rho'_{V,\statement} \ham_\statement)\right)\leq 
    \sum_{\statement\in L_\no}\Pr[\event_\statement]\left(1-\frac{\beta}{N'}\right)\leq 1-\frac{\beta}{N'}.
\end{align*}
\end{proof}

\begin{proof}[Proof of \cref{lem:NIZK_zero-knowledge}]
We describe the simulator $\siml$ below.
\begin{description}
\item[$\siml(\vkey,\statement)$:]
The simulator parses $(W_1,...,W_N,m_1,...,m_N,S_V, 
\{\whx_j,\whz_j\}_{j\in S_V})\la \vkey$ and does the following.
\begin{enumerate}
    \item Generate the classical description of the density matrix $\rho_{S_V}:=\siml_{\hist}(\statement,S_V)$
    where $\siml_{\hist}$ is as in \cref{lem:five_local_Hamiltonian}. 
    \item Sample $\{x_j,z_j\}_{j\in S_V}$ according to the probability distribution of outcomes of the Bell-basis measurements of the corresponding pairs of qubits of $\left(\prod_{j\in S_V}X_j^{\whx_j}Z_j^{\whz_j}\right)\rho_{S_V}\left(\prod_{j\in S_V}Z_j^{\whz_j}X_j^{\whx_j}\right)$ and $\bigotimes_{j\in S_V}(U(W_j)\ket{m_j})$.  
We emphasize that this measurement can be simulated in a  classical probabilistic polynomial time since $|S_V|\leq 5$. 
\item Choose $(x_j,z_j)\sample \bit^2$ for all $j\in [N]\setminus S_V$. 
\item Output $\pi:=(x,z)$ where $x:=x_1\concat x_2\concat...\concat x_N$ and $z:=z_1\concat z_2\concat...\concat z_N$.
\end{enumerate}
\end{description} 
We prove that the output of this simulator is indistinguishable from the real proof.  
For proving this, we consider the following sequences of modified simulators.
We note that these simulators may perform quantum computations  unlike the real simulator.
\begin{description}
\item[$\siml_1(\vkey,\statement)$:]
The simulator parses $(W_1,...,W_N,m_1,...,m_N,S_V, 
\{\whx_j,\whz_j\}_{j\in S_V})\la \vkey$ and does the following.
\begin{enumerate}
    \item Generate the classical description of the density matrix $\rho_{S_V}:=\siml_{\hist}(\statement,S_V)$
    where $\siml_{\hist}$ is as in \cref{lem:five_local_Hamiltonian}. (This step is the same as the step 1 of $\siml(\vkey,\statement)$.)
    \item 
    Generate $\widetilde{\rho'}_{\hist}:=\left(\prod_{j\in S_V}X_j^{\whx_j}Z_j^{\whz_j}\right)\rho_{S_V}\left(\prod_{j\in S_V}Z_j^{\whz_j}X_j^{\whx_j}\right)\otimes \frac{I_{[N]\setminus S_V}}{2^{|[N]\setminus S_V|}}$.
    \item Measure $j$-th qubits of $\widetilde{\rho'}_{\hist}$ and $\rho_P:=\bigotimes_{j=1}^N(U(W_j)|m_j\rangle)$ in the Bell basis for $j\in[N]$, and let $(x_j,z_j)$ be the $j$-th measurement result. 
\item Output $\pi:=(x,z)$ where $x:=x_1\concat x_2\concat...\concat x_N$ and $z:=z_1\concat z_2\concat...\concat z_N$.
\end{enumerate}
\end{description} 

Clearly, the distributions of 
$\{x_j,z_j\}_{j\in S_V}$ 
output by $\siml(\vkey,\statement)$ and $\siml_1(\vkey,\statement)$ are the same.
Moreover, the distributions of $\{x_j,z_j\}_{j\in [N]\setminus S_V}$ 
output by $\siml(\vkey,\statement)$ and $\siml_1(\vkey,\statement)$ are both uniformly and independently random. 
Therefore, output distributions of $\siml(\vkey,\statement)$ and $\siml_1(\vkey,\statement)$  are exactly the same.

Next, we consider the following modified simulator that takes a witness $\witness\in R_L(\statement)$ as input.

\begin{description}
\item[$\siml_2(\vkey,\statement,\witness)$:]
The simulator parses $(W_1,...,W_N,m_1,...,m_N,S_V, 
\{\whx_j,\whz_j\}_{j\in S_V})\la \vkey$ and does the following.
\begin{enumerate}
    \item
    Generate the history state $\rho_{\hist}$ for $\ham_\statement$ from $\witness$.
    \item Generate $(\whx_j,\whz_j)\sample \bit^2$ for $j\in [N]\setminus S_V$ and let $\whx:=\whx_1\concat...\concat\whx_N$ and $\whz:=\whz_1\concat...\concat\whz_N$. 
    \item Compute $\rho'_{\hist}:=X^{\whx} Z^{\whz} \rho_{\hist} Z^{\whz}X^{\whx}$.
    \item Measure $j$-th qubits of $\rho'_{\hist}$ and $\rho_P:=\bigotimes_{j=1}^N(U(W_j)|m_j\rangle)$ in the Bell basis for $j\in[N]$, and let $(x_j,z_j)$ be the $j$-th measurement result. 
\item Output $\pi:=(x,z)$ where $x:=x_1\concat x_2\concat...\concat x_N$ and $z:=z_1\concat z_2\concat...\concat z_N$.
\end{enumerate}

By \cref{lem:Pauli_mixing}, we have $\rho'_{\hist}=\left(\prod_{j\in S_V}X_j^{\whx_j}Z_j^{\whz_j}\right)\Tr_{N\setminus S_V}[\rho_{\hist}]\left(\prod_{j\in S_V}Z_j^{\whz_j}X_j^{\whx_j}\right)\otimes \frac{I_{[N]\setminus S_V}}{2^{|[N]\setminus S_V|}}$
from the view of a distinguisher that has no information on $\{\whx_j,\whz_j\}_{j\in[N]\setminus S_V}$. 
By \cref{lem:five_local_Hamiltonian}, we have
$\|\rho_{S_V}-\Tr_{[N]\setminus S_V}\rho_\hist\|_{tr}=\negl(\secpar)$. 
Therefore, we have $\|\widetilde{\rho'}_\hist-\rho'_\hist\|_{tr}=\negl(\secpar)$. 
This means that $\siml_1(\vkey,\statement)$ and $\siml_2(\vkey,\statement,\witness)$ are statistically indistinguishable from the view of a distinguisher that makes at most one query. 

Finally, noting that the output distribution of $\siml_2(\vkey,\statement,\witness)$ is exactly the same as that of $\prove(\pkey,\statement,\witness)$, the proof of \cref{lem:NIZK_zero-knowledge} is completed. 
\end{description}
\end{proof}

%% file: NIZK_computational.tex
\section{Dual-Mode CV-NIZK with Preprocessing}\label{sec:Dual-mode}
In this section, we extend the CV-NIZK given in \cref{sec:CV-NIZK} to reduce the amount of trust in the setup at the cost of 
introducing a quantum preprocessing and 
relying on a computational assumption. 
In the construction in \cref{sec:CV-NIZK}, we assume that the trusted setup algorithm honestly generates proving and verification keys, which are correlated with each other, and sends them to the prover and verifier, respectively, without revealing them to the other party. 
Here, we give a construction of CV-NIZK with preprocessing that consists of the generation of common reference string by a trusted party and a single instance-independent quantum message from the verifier to the prover. 
We call such a model the  \CRSVP  model.
We note this is the same model as is considered in \cite{C:ColVidZha20}. 
Moreover, our construction has a nice feature called the dual-mode property, which has been considered for NIZKs for $\NP$~\cite{GrothSahai,GOS,C:PeiShi19}.
\ifnum\submission=1
\else
The dual-mode property requires that there are two computationally indistinguishable modes of generating a common reference string, one of which ensures statistical soundness (and computational zero-knowledge) while the other ensures statistical zero-knowledge (and computational soundness). 
To the best of our knowledge, ours is the first construction of a dual-mode NIZK for $\QMA$ in any kind of model. 
\fi

\subsection{Definition}
We give a formal definition of a dual-mode CV-NIZK in the \CRSVP model. 

\begin{definition}[Dual-Mode CV-NIZK in the \CRSVP Model]\label{def:dual-mode}
A dual-mode CV-NIZK for a $\QMA$ promise problem $L=(L_\yes,L_\no)$ in the \CRSVP model consists of algorithms $\Pi=(\crsgen,\preprocess,\allowbreak \prove,\verify)$ with the following syntax:
\begin{description}
\item[$\crsgen(1^\secpar,\mode)$:] This is a PPT algorithm that takes the security parameter $1^\secpar$ and a mode $\mode\in \{\binding,\hiding\}$ as input and outputs a classical common reference string $\crs$.
We note that $\crs$ can be reused  and thus this algorithm is only needed to run once by a trusted third party. 
\item[$\preprocess(\crs)$:]
This is a QPT algorithm that takes the common reference string $\crs$ as input and outputs a quantum proving key $\pkey$ and a classical  verification key $\vkey$.
We note that this algorithm is supposed to be run by the verifier as preprocessing, and $\pkey$ is supposed to be sent to the prover while $\vkey$ is supposed to be kept on verifier's side in secret. 
We also note that they can be used only once  and cannot be reused unlike $\crs$.
\item[$\prove(\crs,\pkey,\statement,\witness^{\otimes k})$:] This is a QPT algorithm that takes 
the common reference string $\crs$, 
the proving key $\pkey$, a statement $\statement$, and $k=\poly(\secpar)$ copies $\witness^{\otimes k}$ of a witness $\witness\in R_L(\statement)$ as input and outputs a classical proof $\pi$.
\item[$\verify(\crs,\vkey,\statement,\pi)$:]
This is a PPT algorithm that takes 
the common reference string $\crs$,  
the verification key $\vkey$, a statement $\statement$, and a proof $\pi$ as input and outputs $\top$ indicating acceptance or $\bot$ indicating rejection. 
\end{description}
We require $\Pi$ to satisfy the following properties for some $0<s<c<1$ such that $c-s>1/\poly(\secpar)$.
Especially, when we do not specify $c$ and $s$, they are set as $c=1-\negl(\secpar)$ and $s=\negl(\secpar)$. 

\medskip
\noindent \underline{
\textbf{$c$-Completeness.}
}
For all
$\mode\in \{\binding,\hiding\}$, $\statement\in L_\yes\cap \bit^\secpar$, and $\witness\in R_L(\statement)$, we have 
\begin{align*}
    \Pr\left[
    \verify(\crs,\vkey,\statement,\pi)=\top 
    :
    \begin{array}{c}
          \crs\sample \crsgen(1^\secpar,\mode) \\
         (\pkey,\vkey) \sample \preprocess(\crs)\\
         \pi \sample \prove(\crs,\pkey,\statement,\witness^{\otimes k})
    \end{array}
    \right]
    \geq c.
\end{align*}

\medskip
\noindent \underline{
\textbf{(Adaptive) Statistical $s$-Soundness in the Binding Mode}
}
For all unbounded-time adversary $\A$, we have 
\begin{align*}
    \Pr\left[
    \statement\in L_\no \land \verify(\crs,\vkey,\statement,\pi)=\top
    :
    \begin{array}{c}
          \crs\sample \crsgen(1^\secpar,\binding) \\
         (\pkey,\vkey) \sample \preprocess(\crs)\\
       (\statement,\pi) \sample \A(\crs,\pkey)
    \end{array}
    \right]
    \leq s.
\end{align*}

\medskip
\noindent \underline{
\textbf{(Adaptive Multi-Theorem) Statistical Zero-Knowledge in the Hiding Mode.}
}
There exists a PPT simulator $\siml_0$ and a QPT simulator $\siml_1$ such that for any unbounded-time  distinguisher $\dist$, we have 
\begin{align*}
    &\left|\Pr\left[\dist^{\ora_P(\crs,\cdot,\cdot,\cdot)}(\crs)=1:
    \begin{array}{c}
          \crs\sample \crsgen(1^\secpar,\hiding) 
    \end{array}
    \right]\right.\\     &-\left.\Pr\left[\dist^{\ora_S(\td,\cdot,\cdot,\cdot)}(\crs)=1:
    \begin{array}{c}
          (\crs,\td)\sample \siml_0(1^\secpar)
    \end{array}
    \right]\right|\leq \negl(\secpar)
\end{align*}
where 
$\dist$ can make $\poly(\secpar)$ queries, which should be of the form
$(\pkey,\statement,\witness^{\otimes k})$ 
where $\witness\in R_L(\statement)$ and $\witness^{\otimes k}$ is unentangled with $\dist$'s internal registers,\footnote{
We remark that $\pkey$ is allowed to be entangled  with $\dist$'s internal registers unlike $\witness^{\otimes k}$.
See also footnote \ref{footnote:unentangled}.}  $\ora_P(\crs,\pkey,\statement,\witness^{\otimes k})$ returns $\prove(\crs,\pkey,\statement,\witness^{\otimes k})$, and
  $\ora_S(\td,\pkey,\statement,\witness^{\otimes k})$ returns 
  $\siml_1(\td,\pkey,\statement)$.
 
 \medskip
  \noindent \underline{
\textbf{Computational Mode Indistinguishability.}
}
For any non-uniform QPT distinguisher $\dist$, we have
\ifnum\submission=1
\begin{align*}
    \left|\Pr\left[\dist(\crs_{\binding})=1\right]
    - \Pr\left[\dist(\crs_{\hiding})=1\right]
    \right|\leq \negl(\secpar)
\end{align*} 
where $\crs_{\binding}\sample \crsgen(1^\secpar,\binding)$ and $\crs_{\hiding}\sample \crsgen(1^\secpar,\hiding)$.
\else
\begin{align*}
    \left|\Pr\left[\dist(\crs)=1:\crs\sample \crsgen(1^\secpar,\binding)\right]
    - \Pr\left[\dist(\crs)=1:\crs\sample \crsgen(1^\secpar,\hiding)\right]
    \right|\leq \negl(\secpar).
\end{align*} 
\fi
\end{definition}
\ifnum\submission=1
Some remarks on the above definition are given in \cref{sec:remark_definition_dual}.
\else
\begin{remark}[On definition of  zero-knowledge property]
By considering a combination of $\crsgen$ (for a fixed $\mode$) and $\preprocess$ as a setup algorithm, (dual-mode) CV-NIZK in the \CRSVP~model can be seen as a CV-NIZK in the QSP model in a syntactical sense.    
However, it seems difficult to prove that this  satisfies (even a computational variant of) the zero-knowledge property  defined in \cref{def:CV-NIZK} due to the following reasons:
\begin{enumerate}
    \item In \cref{def:dual-mode}, $\siml_1$ is quantum, whereas a simulator is required to be classical in \cref{def:CV-NIZK}. 
    We observe that this seems unavoidable in the above model: If $\pkey$ is quantum, then a classical simulator cannot even take $\pkey$ as input. On the other hand, if $\pkey$ is classical, then that implies $L\in\mathbf{AM}$ similarly to the final paragraph of \cref{sec:CV-NIP}. 
    \item A simulator in \cref{def:dual-mode} can embed a trapdoor $\td$ behind the common reference string $\crs$  whereas a simulator in  \cref{def:CV-NIZK} just takes an honestly generated verification key $\vkey$ as input. 
    We remark that this also seems unavoidable since $\vkey$ may be maliciously generated when the verifier is malicious, in which case just taking $\vkey$ as input would be useless for the simulation. 
\end{enumerate}
On the other hand, the definition in \cref{def:dual-mode} allows a distinguisher (that plays the role of a malicious verifier)  to maliciously generate $\pkey$, which is  a stronger capability than that of a distinguisher in  \cref{def:CV-NIZK}. 
Therefore, the zero-knowledge properties in \cref{def:dual-mode} and  \cref{def:CV-NIZK} are incomparable. 
We believe that the definition of the zero-knowledge property in \cref{def:dual-mode} ensures meaningful security.
It 
roughly means that any malicious verifier cannot learn anything beyond what could be computed in quantum polynomial time by itself even if it is allowed to interact with many sessions of honest provers under maliciously generated proving keys and the reused honestly generated common reference string.
While this does not seem very meaningful when $L\in \BQP$, we can ensure a meaningful privacy of the witness when $L\in \QMA$.  
Finally we remark that our definition is essentially the same as that in \cite{C:ColVidZha20}  (except for the dual-mode property). 
\end{remark}
\begin{remark}[Comparison to NIZK in the malicious designated verifier model]
A CV-NIZK for $\QMA$ in the \CRSVP model as defined above is syntactically very similar to the NIZK for $\QMA$ in the malicious designated verifier model as introduced in \cite{C:Shmueli21}.
However, a crucial difference is that the proving key $\pkey$ is a quantum state in our case and cannot be reused whereas that is classical and can be reused for proving multiple statements in \cite{C:Shmueli21}.
On the other hand, a CV-NIZK in the \CRSVP model has two nice features  that the NIZK of \cite{C:Shmueli21} does not have: one is that verification can be done classically in the online phase and the other is the dual-mode property.
\end{remark}
\fi

Though \cref{def:dual-mode} does not explicitly require anything on soundness in the hiding mode or the zero-knowledge property in the binding mode, we can easily prove that they are satisfied in a computational sense. 
\ifnum\submission=1
See \cref{sec:dual_mode_transfer} for details.
\else
Specifically, we have the following lemma.
\begin{lemma}
If a dual-mode CV-NIZK $\Pi=(\crsgen,\preprocess,\prove,\verify)$ for a $\QMA$ promise problem $L$ satisfies statistical $s$-soundness in the binding mode, statistical zero-knowledge property in the hiding mode, and computational mode indistinguishability, then it also satisfies the following properties.

\medskip
\noindent \underline{
\textbf{(Exclusive-Adaptive) Computational $(s+\negl(\secpar))$-Soundness in the Hiding Mode}
}
For all non-uniform QPT adversaries $\A$, we have 
\begin{align*}
    \Pr\left[
  \verify(\crs,\vkey,\statement,\pi)=\top
    :
    \begin{array}{c}
          \crs\sample \crsgen(1^\secpar,\hiding) \\
         (\pkey,\vkey) \sample \preprocess(\crs)\\
       (\statement,\pi) \sample \A(\crs,\pkey)
    \end{array}
    \right]
    \leq s+\negl(\secpar).
\end{align*}
where $\A$'s output must always satisfy $\statement \in L_\no$.  


\medskip
\noindent \underline{
\textbf{(Adaptive Multi-Theorem) Computational Zero-Knowledge in the Binding Mode.}
}
There exists a PPT simulator $\siml_0$ and QPT simulator $\siml_1$ such that for any  non-uniform QPT distinguisher $\dist$, we have 
\begin{align*}
    &\left|\Pr\left[\dist^{\ora_P(\crs,\cdot,\cdot,\cdot)}(\crs)=1:
    \begin{array}{c}
          \crs\sample \crsgen(1^\secpar,\binding) 
    \end{array}
    \right]\right.\\     &-\left.\Pr\left[\dist^{\ora_S(\td,\cdot,\cdot,\cdot)}(\crs)=1:
    \begin{array}{c}
          (\crs,\td)\sample \siml_0(1^\secpar)
    \end{array}
    \right]\right|\leq \negl(\secpar)
\end{align*}
where 
$\dist$ can make $\poly(\secpar)$ queries, which should be of the form $(\pkey,\statement,\witness^{\otimes k})$ where $\witness\in R_L(\statement)$ and $\witness^{\otimes k}$ is unentangled with $\dist$'s internal registers, $\ora_P(\crs,\pkey,\statement,\witness^{\otimes k})$ returns $\prove(\crs,\pkey,\statement,\witness^{\otimes k})$, and
  $\ora_S(\td,\pkey,\statement,\witness^{\otimes k})$ returns $\siml_1(\td,\pkey,\statement)$.
\end{lemma}  
Intuitively, the above lemma holds because soundness and zero-knowledge should transfer from one mode to the other by the mode indistinguishability since otherwise we can  distinguish the two modes. 
Here, security degrades to  computational ones as the mode indistinguishability only holds against QPT distinguishers. 
We omit a formal proof since this is easy and can be proven similarly to a similar statement for dual-mode NIZKs for $\NP$, which has been folklore and formally proven recently \cite{INDOCRYPT:ArtBel20}. 

\begin{remark}
Remark that soundness in the hiding mode is 
defined in the ``exclusive style" where $\A$ should always output $\statement \in L_\no$. 
This is weaker than  soundness in the ``penalizing style" as in \cref{def:dual-mode} where $\A$ is allowed to also output 
$\statement \in L_\yes$ and we add $\statement \in L_\no$ as part of the adversary's winning condition.  
This is because the adaptive soundness in the penalizing style does not transfer well through the mode change while the adaptive soundness in the exclusive style does.
This was formally proven for NIZK for $\NP$ in the common reference string model  in \cite{INDOCRYPT:ArtBel20}, and easily extends to CV-NIZK for $\QMA$ in the \CRSVP model.
This is justified by the impossibility of  penalizing-adaptively (computational) sound and statistically zero-knowledge NIZK for $\NP$ in the classical setting (under falsifiable assumptions)  \cite{TCC:Pass13}.
We leave it open to study if a similar impossibility holds for dual-mode CV-NIZK for $\QMA$ in the \CRSVP model.
\end{remark}
\fi

Finally, we note that we can amplify the gap between the thresholds for completeness and soundness by parallel repetitions similarly to CV-NIZK in the QSP model as discussed in \cref{sec:def_CV-NIZK}.
As a result, we obtain the following lemma. 
\begin{lemma}[Gap amplification for dual-mode CV-NIZK in the \CRSVP model]\label{lem:gap_amplification_dual}
If there exists a dual-mode CV-NIZK for $L$ in the \CRSVP model that satisfies $c$-completeness and $s$-soundness, for some $0<s<c<1$ such that $c-s>1/\poly(\secpar)$, then there exists a dual-mode CV-NIZK for $L$ in the \CRSVP model (with $(1-\negl(\secpar))$-completeness and $\negl(\secpar)$-soundness).
\end{lemma}
Since this can be proven similarly to \cref{lem:CV-NIZK_amplification}, we omit a proof. 

\subsection{Building Blocks}
We introduce two cryptographic bulding blocks for our dual-mode CV-NIZK in the \CRSVP model.

\paragraph{Lossy Encryption}
\ifnum\submission=1
Intuitively, a lossy encryption scheme is  a public key encryption scheme with a special property that we can generate a \emph{lossy key} that is computationally indistinguishable from an honestly generated public key, for which there is no corresponding decryption key. 
Since this is a standard primitive, we give a definition in \cref{sec:definition_lossy}.
\else
The first building block is lossy encryption \cite{EC:BelHofYil09}.
Intuitively, a lossy encryption scheme is  a public key encryption scheme with a special property that we can generate a \emph{lossy key} that is computationally indistinguishable from an honestly generated public key, for which there is no corresponding decryption key. 
\begin{definition}[Lossy Encryption]\label{def:lossy}
A lossy encryption scheme 
over the message space $\mathcal{M}$ and the randomness space $\mathcal{R}$
consists of PPT algorithms $\Pi_\LE=(\injgen,\lossygen,\enc,\dec)$ with the following syntax.
\begin{description}
\item[$\injgen(1^\secpar)$:]
The injective key generation algorithm takes the security parameter $1^\secpar$ as input and ouputs an \emph{injective} public key $\pk$ and a secret key $\sk$.
\item[$\lossygen(1^\secpar)$:]
The lossy key generation algorithm takes the security parameter $1^\secpar$ as input and ouputs a \emph{lossy} public key $\pk$.
\item[$\enc(\pk,\mu)$:]
The encryption algorithm takes the public key $\pk$ and a message $\mu \in \mathcal{M}$ as input and outputs a ciphertext $\ct$.
This algorithm uses a randomness $R\in \mathcal{R}$.
We denote by  $\enc(\pk,\mu;R)$ to mean that we run $\enc$ on input $\pk$ and $\mu$ and randomness $R$ when we need to clarify the randomness.  
\item[$\dec(\sk,\ct)$:]
The decryption algorithm takes the secret key $\sk$ and a ciphertext $\ct$ as input and outputs a message $\mu$.
\end{description}
We require $\Pi_\LE$ to satisfy the following properties.
\end{definition}

\medskip
\noindent \underline{
\textbf{Correctness on Injective Keys}
}
For all $\mu\in \mathcal{M}$, we have 
\begin{align*}
    \Pr\left[
    \dec(\sk,\ct)=\mu:
    \begin{array}{ll}
    (\pk,\sk)\sample \injgen(1^\secpar)\\
    \ct\sample \enc(\pk,\mu)
    \end{array}
    \right]= 1.
\end{align*}

\medskip
\noindent \underline{
\textbf{Lossiness on Lossy Keys}
}
With overwhelming probability over $\pk\sample \lossygen(1^\secpar)$, for all $\mu_0,\mu_1 \in \mathcal{M}$ and all unbounded-time distinguisher $\dist$, we have  
\begin{align*}
    \left|\Pr\left[
    \dist(\ct)=1:
    \begin{array}{ll}
        \ct\sample \enc(\pk,\mu_0) 
    \end{array}
    \right]-
        \Pr\left[
    \dist(\ct)=1:
    \begin{array}{ll}
        \ct\sample \enc(\pk,\mu_1) 
    \end{array}
    \right]\right|
    \leq \negl(\secpar).
\end{align*}

\medskip
\noindent \underline{
\textbf{Computational Mode Indistinguishability}
}
For any non-uniform QPT distinguisher $\dist$, we have
\begin{align*}
    \left|\Pr\left[\dist(\pk)=1:(\pk,\sk)\sample \injgen(1^\secpar)\right]
    - \Pr\left[\dist(\pk)=1:\pk\sample \lossygen(1^\secpar)\right]
    \right|\leq \negl(\secpar).
\end{align*}

It is well-known that Regev's encryption \cite{JACM:Regev09} is lossy encryption under the LWE assumption with a negligible correctness error.
We can modify the scheme to achieve perfect correctness by a standard technique. 
Then we have the following lemma.
 
\begin{lemma}\label{lem:lossy_encryption}
If the LWE assumption holds, then there exists a lossy encryption scheme.
\end{lemma}
\fi

\paragraph{Dual-Mode Oblivious Transfer}
The second building block is a $k$-out-of-$n$ \emph{dual-mode oblivious transfer}. 
Though this is a newly introduced definition in this paper, $1$-out-of-$2$ case is already implicit in existing works on universally composable (UC-secure) \cite{Canetti20}  oblivious transfers \cite{C:PeiVaiWat08,SCN:Quach20}.
\begin{definition}[Dual-mode oblivious transfer]\label{def:OT}
A (2-round) $k$-out-of-$n$ dual-mode oblivious transfer  with a message space $\mathcal{M}$ consists of PPT algorithms $\Pi_\OT=(\crsgen, \receiver, \sender,\derive)$.
\begin{description}
\item[$\crsgen(1^\secpar,\mode)$:] This is an algorithm supposed to be run by a trusted third party that takes the security parameter $1^\secpar$ and a mode $\mode\in \{\binding,\hiding\}$ as input and outputs a common reference string $\crs$.
\item[$\receiver(\crs,J)$:] This is an algorithm supposed to be run by a receiver that takes the common reference string $\crs$ and an ordered set of $k$ indices $J\in [n]^k$ as input and outputs a first message $\ot_1$ and a receiver's state $\otst$.
\item[$\sender(\crs,\ot_1,\vecmu)$:] This is an algorithm supposed to be run by a sender that takes  the common reference string $\crs$, a first message $\ot_1$ sent from a receiver and a tuple of messages $\vecmu\in \mathcal{M}^n$ as input and outputs a second  message $\ot_2$.
\item[$\derive(\crs,\otst,\ot_2)$:] This is an algorithm supposed to be run by a receiver that takes a receiver's state $\otst$ and a second message $\ot_2$ as input and outputs a tuple of messages $\vecmu'\in \mathcal{M}^k$.
\end{description}

We require the following properties.

\medskip
\noindent \underline{
\textbf{Correctness}
}
For all $\mode\in \{\binding,\hiding\}$, $J=(j_1,...,j_k) \in [n]^k$, and $\vecmu=(\mu_1,...,\mu_n)\in \mathcal{M}^n$, we have 
\begin{align*}
\Pr\left[\derive(\crs,\otst,\ot_2)=(\mu_{j_1},...,\mu_{j_k}):
\begin{array}{ll}
    \crs\sample \crsgen(1^\secpar,\mode)   \\
  (\ot_1,\otst) \sample \receiver(\crs,J)   \\
 \ot_2 \sample \sender(\crs,\ot_1,\vecmu)
\end{array}
\right]\geq 1-\negl(\secpar).
\end{align*}

\medskip
\noindent \underline{
\textbf{Statistical Receiver's Security in the Binding Mode}
}
 Intuitively, this security requires that 
 the indices chosen by a receiver are hidden from a sender. 
 Formally, we require that 
  there is a
 PPT algorithm $\siml_\rec$  such that for any unbounded-time distinguisher $\dist$ and $J\in [n]^k$, we have
\begin{align*}
    &\left|\Pr\left[\dist(\crs,\ot_1)=1:
    \begin{array}{ll}
         \crs \sample \crsgen(1^\secpar,\binding)  \\
         (\ot_1,\otst)\sample \receiver(\crs,J)
    \end{array}
    \right]\right.\\
    &\left.- \Pr\left[\dist(\crs,\ot_1)=1:
    \begin{array}{ll}
         \crs\sample \crsgen(1^\secpar,\binding)  \\
         \ot_1\sample \siml_\rec(\crs)
    \end{array}
    \right]
    \right|\leq \negl(\secpar).
\end{align*}

 \medskip
\noindent \underline{
\textbf{Statistical Sender's Security in the Hiding Mode}
}
Intuitively, this security requires that we can extract the indices of messages which a (possibly malicious) receiver tries to learn by using a trapdoor  in the hiding mode. 
Formally, there are PPT algorithms $\siml_\CRS$ and
$\siml_\sen$ and a
deterministic classical polynomial-time algorithm $\Open_\rec$ such that the following two properties are satisfied.
\begin{itemize}
\item For any unbounded-time distinguisher $\dist$, we have 
\ifnum\submission=1
\begin{align*}
    \left|\Pr\left[\dist(\crs_{\mathsf{real}})=1
    \right]
    - \Pr\left[\dist(\crs_{\mathsf{sim}})=1
    \right]
    \right|\leq \negl(\secpar)
\end{align*}
where $\crs_{\mathsf{real}} \sample \crsgen(1^\secpar,\hiding)$ and $(\crs_{\mathsf{sim}},\td) \sample \siml_\CRS(1^\secpar)$.
\else
\begin{align*}
    \left|\Pr\left[\dist(\crs)=1:
         \crs \sample \crsgen(1^\secpar,\hiding)
    \right]
    - \Pr\left[\dist(\crs)=1:
         (\crs,\td) \sample \siml_\CRS(1^\secpar)
    \right]
    \right|\leq \negl(\secpar).
\end{align*}
\fi
\item For any unbounded-time adversary $\A=(\A_0,\A_1)$ (that plays the role of a malicious receiver) and  $\vecmu=(\mu_1,...,\mu_n)$, we have
 \begin{align*}
 &\left|\Pr\left[\A_1(\st_\A,\ot_2)=1: 
 \begin{array}{ll}
 (\crs,\td)\sample \siml_\CRS(1^\secpar)\\
 (\ot_1,\st_\A)\sample\A_0(\crs,\td)  \\
 \ot_2\sample \sender(\crs,\ot_1,\vecmu)\\
 \end{array}
 \right]\right.\\
 &\left.-
 \Pr\left[\A_1(\st_\A,\ot_2)=1: 
 \begin{array}{ll}
 (\crs,\td)\sample \siml_\CRS(1^\secpar)\\
 (\ot_1,\st_\A)\sample\A_0(\crs,\td)  \\
J := \Open_\rec(\td,\ot_1)\\
 \ot_2\sample \siml_\sen(\crs,\ot_1,J,\vecmu_J)\\
 \end{array}
 \right]
 \right|\leq \negl(\secpar)
 \end{align*}
where the output of $\Open_\rec$ always satisfies $J\in [n]^k$ and
$\vecmu_J:=(\mu_{j_1},...,\mu_{j_k})$ 
for $J=(j_1,...,j_k)$.
\end{itemize}

 \medskip
  \noindent \underline{
\textbf{Computational Mode Indistinguishability.}
}
For any non-uniform QPT distinguisher $\dist$, we have
\ifnum\submission=1
\begin{align*}
    \left|\Pr\left[\dist(\crs_{\binding})=1\right]
    - \Pr\left[\dist(\crs_{\hiding})=1\right]
    \right|\leq \negl(\secpar)
\end{align*} 
where $\crs_{\binding}\sample \crsgen(1^\secpar,\binding)$ and $\crs_{\hiding}\sample \crsgen(1^\secpar,\hiding)$.
\else
\begin{align*}
    \left|\Pr\left[\dist(\crs)=1:\crs\sample \crsgen(1^\secpar,\binding)\right]
    - \Pr\left[\dist(\crs)=1:\crs\sample \crsgen(1^\secpar,\hiding)\right]
    \right|\leq \negl(\secpar).
\end{align*}
\fi
\end{definition}

\ifnum\submission=1
A remark on the definition of dual-mode oblivious transfer is given in \cref{sec:remark_definition_OT}.
\else
\begin{remark}[On security definition of dual-mode oblivious transfer]
We remark that security of a $k$-out-of-$n$  dual-mode  oblivious transfer as defined in \cref{def:OT}
does not imply UC-security \cite{Canetti20,C:PeiVaiWat08,SCN:Quach20} or even full-simulation security  in the standard stand-alone simulation-based definition \cite{RSA:Lindell08a}.
This is because the receiver's security in \cref{def:OT} only ensures privacy of $J$ and does not prevent a malicious sender from generating $\ot_2$ so that he can manipulate the message derived on the receiver's side depending on $J$.  
The security with such a weaker receiver's security is often referred to as half-simulation security \cite{EC:CamNevshe07}.  
We define the security in this way due to the following reasons:
\begin{enumerate}
    \item This definition is sufficient for constructing a dual-mode CV-NIZK in the \CRSVP model given in \cref{sec:construction_dual-mode}
    by additionally relying on lossy encryption. 
    \item We are not aware of an efficient construction of a $k$-out-of-$n$ oblivious transfer  that satisfies full-simulation security under a post-quantum assumption (even if we ignore the dual-mode property).  
    We note that Quach \cite{SCN:Quach20} gave a construction of a $1$-out-of-$2$ oblivious transfer with full-simulation security based on LWE and we can extend it to $1$-out-of-$n$ one.\footnote{His construction further satisfies UC-security, which is stronger than full-simulation security.}
    However, we are not aware of an efficient way to convert this into $k$-out-of-$n$ one without losing the full-simulation security. 
    We note that a conversion from $1$-out-of-$n$ to $k$-out-of-$n$ oblivious transfer by a simple $k$-parallel repetition  does not work if we require the  full-simulation security since a malicious sender can send different inconsistent messages in different sessions, which should be considered as an attack against full-simulation security.
One possible way to prevent such an inconsistent message attack is to let the sender prove that the messages in all sessions are consistent by using (post-quantum) NIZK for $\NP$ in the common reference string model \cite{C:PeiShi19}.
However, such a construction is very inefficient since it uses the underlying $1$-out-of-$n$ oblivious transfer in a non-black-box manner.
On the other hand, the half-simulation security is preserved under parallel repetitions as shown in \cref{sec:OT}, and thus we can achieve this much more efficiently.  
\end{enumerate}


\end{remark}
\fi

\begin{lemma}\label{lem:OT}
If the LWE assumption holds, then there exists $k$-out-of-$n$ dual-mode oblivious transfer for arbitrary $0<k<n$ that are polynomial in $\secpar$. 
\end{lemma}
\begin{proof}[Proof (sketch)]
First, we can see that the LWE-based UC-secure OT by Quach \cite{SCN:Quach20} can be seen as a $1$-out-of-$2$ dual-mode oblivious transfer.
This construction can be converted into $1$-out-of-$n$ dual-mode oblivious transfer by using the generic conversion for an ordinary oblivious transfer given in \cite{FOCS:BraCreRob86} observing that the conversion preserves the dual-mode property.\footnote{Alternatively, it may be possible to directly construct $1$-out-of-$n$ dual-mode oblivious transfer by appropriately modifying the construction by Quach \cite{SCN:Quach20}.} 
By $k$-parallel repetition of the $1$-out-of-$n$ dual-mode oblivious transfer, we obtain $k$-out-of-$n$ dual-mode oblivious transfer.
The full proof can be found in \cref{sec:OT}.
\end{proof}

\subsection{Construction}\label{sec:construction_dual-mode}
In this section, we construct a dual-mode CV-NIZK in the \CRSVP  model. 
As a result, we obtain the following theorem.
\begin{theorem}\label{thm:dual-mode}
If the LWE assumption holds, then there exists a dual-mode CV-NIZK in the \CRSVP  model. 
\end{theorem}

Let $L$ be a $\QMA$ promise problem, and $\ham_\statement$, $N$, $M$, $p_{i}$, $s_{i}$, $P_i$, $\alpha$, $\beta$, and $\rho_\hist$ be as in \cref{lem:five_local_Hamiltonian} for the language $L$.  
We let $N':=3^5\sum_{i=1}^{5}{N \choose i}$ similarly to \cref{lem:NIZK_completeness_soundness}.
Let $\Pi_\LE=(\injgen_\LE,\lossygen_\LE,\enc_\LE,\dec_\LE)$ be a lossy encryption scheme over the message space $\mathcal{M}_\LE=\bit^2$ and the randomness space $\mathcal{R}_\LE$ as defined in \cref{def:lossy}.
Let $\Pi_\OT=(\crsgen_\OT,\receiver_\OT,\allowbreak\sender_\OT,\derive_\OT)$ be a $5$-out-of-$N$ dual-mode oblivious transfer over the message space $\mathcal{M}_\OT=\mathcal{M}_\LE \times \mathcal{R}_\LE$  
as defined in \cref{def:OT}.
Then our dual-mode CV-NIZK 
$\Pi_\DM=(\crsgen_\DM,\allowbreak\preprocess_\DM,\prove_\DM,\verify_\DM)$ for $L$ is described in \cref{fig:dual}.

\begin{figure}[p]
\rule[1ex]{\textwidth}{0.5pt}
\begin{description}
\item[$\crsgen_\DM(1^\secpar,\mode)$:]
The CRS generation algorithm generates $\crs_\OT\sample \crsgen_\OT(1^\secpar,\mode)$.
\begin{itemize}
    \item If $\mode=\binding$, then it generates $(\pk,\sk)\sample \injgen_\LE(1^\secpar)$.
    \item If $\mode=\hiding$, then it generates $\pk\sample \lossygen_\LE(1^\secpar)$.
\end{itemize}
Then it outputs $\crs_\DM:=(\crs_\OT,\pk)$. 
\item[$\preprocess_\DM(\crs_\DM)$:]
The preprocessing algorithm 
parses $(\crs_\OT,\pk)\la \crs_\DM$ and 
chooses
$(W_1,...,W_N)\sample \{X,Y,Z\}^{N}$, $(m_1,...,m_N)\sample \{0,1\}^{N}$,
and a uniformly random subset $S_V\subseteq [N]$ such that $1\le|S_V|\le5$. 
Let  $J=(j_1,...,j_5)\in [N]^5$ be the elements of $S_{V}$ in the ascending order where we append arbitrary indices when $|S_{V}|<5$. 
It generates $(\ot_{1},\otst)\sample \receiver_\OT(\crs_\OT,J)$ and outputs  a proving key $\pkey:=\left(\rho_{P}:=\bigotimes_{j=1}^N(U(W_j)|m_{j}\rangle),\ot_{1}\right)$ 
and a verification key $\vkey:=\left(W_{1},...,W_{N},m_{1},...,m_{N},S_V,\otst\right)$.
\item[$\prove_\DM(\crs_\DM,\pkey,\statement,\witness)$:]
The proving algorithm 
parses 
$(\crs_\OT,\pk)\la\crs_\DM$ and
$\left(\rho_{P},\ot_{1}\right)\la \pkey$,
generates $(\whx,\whz)\sample \bit^{N}\times \bit^{N}$,  
generates the history state $\rho_\hist$ for $\ham_\statement$ from $\witness$, 
and computes $\rho'_\hist:=X^{\whx} Z^{\whz} \rho_\hist Z^{\whz} X^{\whx}$.
It measures $j$-th qubits of $\rho'_\hist$ and $\rho_P$ in the Bell basis for $j\in [N]$.    
Let $x:=x_1\concat x_2\concat...\concat x_N$, and $z:=z_1\concat z_2\concat...\concat z_N$ where $(x_j,z_j)$ denotes the outcome of $j$-th measurement.
For $j\in[N]$, 
it generates 
$\ct_j:=\enc_\LE(\pk,(\whx_j,\whz_j);R_j)$ where $R_j\sample \mathcal{R}_\LE$ and $\whx_{j}$ and $\whz_{j}$ denote the $j$-th bits of $\whx$ and $\whz$, respectively. 
It 
sets $\mu_j:=((\whx_j,\whz_j),R_j)$ for $j\in [N]$ and
generates $\ot_2\sample \sender_\OT(\crs_\OT,\ot_{1},(\mu_1,...,\mu_N))$. 
It outputs a proof $\pi:=(x,z,\{\ct_j\}_{j\in[N]},\ot_{2})$.

\item[$\verify_\DM(\crs_\DM,\vkey,\statement,\pi)$:]
The verification algorithm 
parses $(\crs_\OT,\pk)\la\crs_\DM$,
$\left(W_{1},...,W_{N},m_{1},...,m_{N},S_V,\otst\right)\la \vkey$, and 
$(x,z,\{\ct_j\}_{j\in[N]},\ot_{2})\la \pi$.
It runs $\vecmu'\sample \derive_\OT(\crs_\OT,\otst,\ot_{2})$ and  parses $(((\whx'_{1},\whz'_{1}),R'_1),...,((\whx'_{5},\whz'_{5}),R'_5))\la \vecmu'$.
If $\enc_\LE(\pk,(\whx'_i,\whz'_i);R'_i)\neq \ct_{j_i}$ for some $i\in [5]$, it outputs $\bot$. Otherwise,  it recovers  $\{\whx_{j},\whz_{j}\}_{j\in S_{V}}$ by setting $(\whx_{j_i},\whz_{j_i}):=(\whx'_{i},\whz'_{i})$ for $i\in [|S_V|]$.
It chooses $i\in [M]$ according to the probability distribution defined by $\{p_{i}\}_{i\in[M]}$ (i.e., chooses $i$ with probability $p_{i}$).
Let
\begin{eqnarray*}
S_i:=\{j\in[N]~|~\mbox{$j$th Pauli operator of $P_i$ is not $I$}\}.
\end{eqnarray*}
We note that we have $1\leq |S_i|\leq 5$ by the $5$-locality of $\ham_\statement$. 
We say that $P_i$ is consistent to $(S_V,\{W_j\}_{j\in S_V})$  if and only if 
$S_i = S_V$ and
the $j$th Pauli operator of
$P_i$ is $W_j$ for all
$j\in S_i$.    
If $P_i$ is not consistent to $(S_V,\{W_j\}_{j\in S_V})$,  it outputs $\top$.
If $P_i$ is consistent to 
$(S_V,\{W_j\}_{j\in S_V})$,  
it flips a biased coin that heads with probability $1-3^{|S_i|-5}$.
If heads, it outputs $\top$.
If tails,
it defines
\begin{eqnarray*}
m_{j}':= 
\left\{
\begin{array}{cc}
m_{j}\oplus x_{j}\oplus \hat{x}_j&(W_j=Z),\\
m_{j}\oplus z_{j}\oplus \hat{z}_j&(W_j=X),\\
m_{j}\oplus x_{j}\oplus \hat{x}_j\oplus z_j\oplus \hat{z}_j&(W_j=Y)
\end{array}
\right.
\end{eqnarray*}
for $j\in S_i$, and outputs $\top$ if  $(-1)^{\bigoplus_{j\in S_i}m'_{j}}=-s_{i}$ and $\bot$ otherwise. 
\end{description} 
\rule[1ex]{\textwidth}{0.5pt}
\hspace{-10mm}
\caption{Dual-Mode CV-NIZK $\Pi_\DM$.}
\label{fig:dual}
\end{figure}

Then we prove the following lemmas.

\begin{lemma}\label{lem:DM_completeness}
$\Pi_\DM$ satisfies  $\left(1-\frac{\alpha}{N'}-\negl(\secpar)\right)$-completeness.
\end{lemma}
\begin{proof}
By the correctness of $\Pi_\OT$, it is easy to see that  the probability that an honestly generated proof passes the verification differs from that in $\Pi_\NIZK$ in \cref{fig:CV-NIZK} only by $\negl(\secpar)$.
Since $\Pi_\NIZK$ satisfies
$\left(1-\frac{\alpha}{N'}\right)$-completeness as shown in \cref{lem:NIZK_completeness_soundness}, $\Pi_\DM$ satisfies  $\left(1-\frac{\alpha}{N'}-\negl(\secpar)\right)$-completeness.
\end{proof}
\begin{lemma}\label{lem:DM_mode_ind}
$\Pi_\DM$ satisfies  the computational mode indistinguishability.
\end{lemma}
\begin{proof}
This can be reduced to the computational mode indistinguishability of $\Pi_\OT$ and $\Pi_\LE$ in a straightforward manner.
\end{proof}

\begin{lemma}\label{lem:DM_soundness}
$\Pi_\DM$ satisfies  statistical $\left(1-\frac{\beta}{N'}+\negl(\secpar)\right)$-soundness in the binding mode.
\end{lemma}

\begin{lemma}\label{lem:DM_ZK}
$\Pi_\DM$ satisfies  the statistical zero-knowledge property  in the hiding mode.
\end{lemma}

By combining
\cref{lem:lossy_encryption,lem:OT,lem:gap_amplification_dual,lem:DM_completeness,lem:DM_mode_ind,lem:DM_soundness,lem:DM_ZK} and
\begin{align*}
\left(1-\frac{\alpha}{N'}-\negl(\secpar)\right)-\left(1-\frac{\beta}{N'}+\negl(\secpar)\right)=\frac{\beta-\alpha}{N'}-\negl(\secpar)=\frac{1}{\poly(\secpar)}, 
\end{align*}
we obtain \cref{thm:dual-mode}.

\ifnum\submission=1
In the following, we give proof sketches of \cref{lem:DM_soundness,lem:DM_ZK}.
See \cref{sec:omitted_sec_proof_dual_mode_NIZK} for full proofs.
\paragraph{Soundness in the binding mode.}
For a cheating prover, we consider a modified soundness game where the challenger extracts $\{\whx_j,\whz_j\}_{j\in S_V}$ from $\{\ct_{j}\}_{j\in S_V}$ by just decrypting them instead of deriving $\{(\whx_j,\whz_j),R_j\}_{j\in S_V}$ from $\ot_2$ and then checking the consistency to $\{\ct_{j}\}_{j\in S_V}$ as in the actual verification algorithm. 
This does not decrease adversary's winning probability since $\{\whx_j,\whz_j\}_{j\in S_V}$ derived from $\ot_2$ should be equal to decryption of $\{\ct_{j}\}_{j\in S_V}$ or otherwise the verification algorithm immediately rejects.
In this game, the challenger does not use $\st$ of $\Pi_\OT$.
Therefore, by the receiver's security of $\Pi_\OT$, adversary's winning probability changes negligibly even if we generate $\ot_1$ by the simulator $\siml_\rec$.
At this point, the challenger obtain no information about $S_V$. 
Then soundness in this game can be reduced to the soundness of $\Pi_{\NIZK}$ in \cref{fig:CV-NIZK} against augmented cheating provers with an additional capability to choose $\{\whx_j,\whz_j\}_{j\in [N]}$. 
By carefully examining the proof of the soundness of $\Pi_{\NIZK}$, one can see that the proof works against such augmented cheating provers as well.
(Note that what is essential for the soundness of $\Pi_{\NIZK}$ is that $S_V$ is hidden from the cheating prover.)
\paragraph{Zero-knowledge in the hiding mode.}
In the hiding mode, $\pk$ of $\Pi_{\LE}$ is in the lossy mode, and thus $\{\ct_{j}\}_{j\in [N]}$ can be simulated only from $\pk$ by encrypting all $0$ message. 
Moreover, by sender's security in the hiding mode of $\Pi_\OT$, $\ot_2$ can be simulated from $\{\whx_j,\whz_j\}_{j\in S_V}$ where $S_V$ is a subset such that $|S_V|=5$ extracted from $\ot_1$.
Therefore, the zero-knowledge property of $\Pi_{\DM}$ can be reduced to the zero-knowledge property of $\Pi_{\NIZK}$ in \cref{fig:CV-NIZK} against augmented malicious verifiers with an additional capability to choose $S_V$ and $\rho_P$.
By carefully examining the proof of the zero-knowledge property of $\Pi_{\NIZK}$, one can see that the proof works against such augmented malicious verifiers as well.
(Note that what is essential for the zero-knowledge property of $\Pi_{\NIZK}$ is that $\{\whx_j,\whz_j\}_{j\notin S_V}$ is hidden from the malicious verifier.)
\else
In the following, we prove \cref{lem:DM_soundness,lem:DM_ZK}.
\input{security_proof_dual-mode}
\fi

%% file: security_proof_dual-mode.tex
\begin{proof}[Proof of \cref{lem:DM_soundness} (Soundness)]
For any adversary $\A$, we consider the following sequence of games between $\A$ and the challenger where we denote by $\Win_i$ the event that the challenger returns $\top$ in $\game_i$. 
\begin{description}
\item[$\game_1$:]
This game is the original soundness game in the binding game. That is, it works as follows:
\begin{enumerate}
    \item The challenger  generates $\crs_\OT\sample \crsgen_\OT(1^\secpar,\binding)$ and $(\pk,\sk)\sample \injgen_\LE(1^\secpar)$.
\item The challenger generates  $(W_1,...,W_N)\sample \{X,Y,Z\}^{N}$, $(m_1,...,m_N)\sample \{0,1\}^{N}$, and $\rho_{P}:=\bigotimes_{j=1}^N(U(W_j)|m_{j}\rangle$. 
    \item The challenger generates $S_V$ and $J=(j_1,...,j_5)$ similarly to $\preprocess_\DM$. 
    \item \label{step:receiver}
    The challenger generates  $(\ot_{1},\otst)\sample \receiver_\OT(\crs_\OT,J)$.
    \item The challenger gives $\crs_\DM$ and a proving key $\pkey:=\left(\rho_P,\ot_{1}\right)$ to $\A$, and $\A$ outputs $(\statement,\pi=(x,z,\{\ct_j\}_{j\in [N]}, \ot_2))$.
    If $\statement \in L_\yes$, the challenger outputs $\bot$ and immediately halts.
    \item \label{step:derive}
    The challenger runs $\vecmu'\sample \derive_\OT(\crs_\OT,\otst,\ot_{2})$ and parses $(((\whx'_{1},\whz'_{1}),R'_1),\allowbreak ...,((\whx'_{5},\whz'_{5}),R'_5))\la \vecmu'$.
If $\enc_\LE(\pk,(\whx'_i,\whz'_i);R'_i)\neq \ct_{j_i}$ for some $i\in [5]$, it outputs $\bot$ and immediately halts. Otherwise, it recovers $\{\whx_{j},\whz_{j}\}_{j\in S_{V}}$ by setting $(\whx_{j_i},\whz_{j_i}):=(\whx'_{i},\whz'_{i})$ for $i\in [|S_V|]$.
    \item The challenger samples $i$ and defines $S_i$ and $P_i$ similarly to $\verify_\DM$. 
    If $P_i$ is not consistent to $(S_V,\{W_j\}_{j\in S_V})$, it outputs $\top$.
If $P_i$ is consistent to 
$(S_V,\{W_j\}_{j\in S_V})$,  
it flips a biased coin that heads with probability $1-3^{|S_i|-5}$.
If heads, it outputs $\top$.
If tails,
it defines $m_{j}'$ for $j\in S_i$ similarly to $\verify_\DM$ and outputs $\top$ if  $(-1)^{\bigoplus_{j\in S_i}m'_{j}}=-s_{i}$ and $\bot$ otherwise.
\end{enumerate}
Our goal is to prove $\Pr[\Win_1]\leq 1-\frac{\beta}{N'}+\negl(\secpar)$.   
\item[$\game_2$:]
This game is identical to the previous game except that 
Step $6$ is replaced with Step $6'$ described as follows.
\begin{itemize}
    \item[$6'$.] 
      The challenger computes $(\whx_j,\whz_j)\sample  \dec_\LE(\sk,\ct_j)$ for $j\in[N]$. 
\end{itemize}
If the challenger does not output $\bot$ in Step $6$, then we have $\enc_\LE(\pk,(\whx'_i,\whz'_i);R'_i)= \ct_{j_i}$ for all $i\in [5]$. 
In this case, we have $\dec_\LE(\sk,\ct_{j_i})=(\whx'_i,\whz'_i)$ by correctness of $\Pi_\LE$.
Therefore, the values of $\{\whx_{j},\whz_{j}\}_{j\in S_{V}}$ computed in Step $6$ and $6'$ are identical conditioned on that the challenger does not output $\bot$ in Step $6$.
Noting that Step $7$ only uses the values of $(\whx_{j},\whz_{j})$ for $j\in S_{V}$, 
we have $\Pr[\Win_1]\leq \Pr[\Win_2]$.
\item[$\game_3$:]
This game is identical to the previous game except that 
Step $4$ is replaced with Step $4'$ described as follows.
\begin{itemize}
\item[$4'$] The challenger generates $\ot_1\sample \siml_\rec(\crs_\OT)$.
\end{itemize}
By statistical receiver's security in the binding mode of $\Pi_\OT$, it is clear that we have $|\Pr[\Win_3]-\Pr[\Win_2]|\leq \negl(\secpar)$. 
\item[$\game_4$:]
This game is identical to the previous game except that Step $2$ is replaced with Step $2'$  described below.
\begin{itemize}
    \item[$2'$.]The challenger generates $N$ Bell-pairs between registers $\regP$ and $\regV$ and lets $\rho_P$ and $\rho_V$ be quantum states in registers $\regP$ and $\regV$, respectively. 
    Then it chooses  $(W_1,...,W_N)\sample \{X,Y,Z\}^{N}$, and measures $j$-th qubit of $\rho_V$ in the $W_j$ basis for all $j\in [N]$, and lets $(m_1,...,m_N)$ be the measurement outcomes. 
\end{itemize}
By \cref{lem:statecollapsing}, 
the joint distributions of $(\rho_P,(W_1,...,W_N,m_1,...m_N))$ in $\game_3$ and $\game_4$ are identical, and thus 
we have $\Pr[\Win_4]=\Pr[\Win_3]$.
\item[$\game_5$:]
This game is identical to the previous game except that the measurement of $\rho_V$ in Step $2'$ is omitted and the way of generating $\{m'_j\}_{j\in S_i}$ in Step $7$ is modified as follows.
\begin{itemize}
    \item  The challenger computes $\rho'_V:=X^{x\oplus \whx}Z^{z\oplus \whz}\rho_V Z^{z\oplus \whz}X^{x\oplus \whx}$.
    For all $j\in S_i$, it 
 measures $j$-th qubit of $\rho'_V$ in $W_j$ basis, and lets $m'_j$ be the measurement outcome.
\end{itemize}

By \cref{lem:XZ_before_measurement},  this does not change the distribution of $\{m'_j\}_{j\in S_i}$.
Therefore, we have $\Pr[\Win_5]=\Pr[\Win_4]$. 

Let 
$\event_\statement$ be the event that the statement output by $\A$ is $\statement$, and
$\rho'_{V,\statement}$ be the state in $\regV$ right before the measurement in the modified Step $7$ conditioned on $\event_\statement$.
For any fixed $P_i$, the probability that $P_i$ is consistent to  $(S_V,\{W_j\}_{j\in S_V})$ and the coin tails is $\frac{1}{N'}$.
Therefore, 
by \cref{lem:prob_and_energy}, we have 
\begin{align*}
    \Pr[\Win_5|\event_\statement]=1-\frac{1}{N'}\Tr(\rho'_{V,\statement} \ham_\statement).
\end{align*}
Then  we have 
\begin{align*}
    \Pr[\Win_5]=\sum_{\statement\notin L}\Pr[\event_\statement]\left(1-\frac{1}{N'}\Tr(\rho'_{V,\statement} \ham_\statement)\right)\leq 
    \sum_{\statement\notin L}\Pr[\event_\statement]\left(1-\frac{\beta}{N'}\right)\leq 1-\frac{\beta}{N'}
\end{align*}
where
the first inequality follows from  \cref{lem:five_local_Hamiltonian}. 
\end{description}

By combining the above, we obtain $\Pr[\Win_1]\leq 1-\frac{\beta}{N'}+\negl(\secpar)$.

This completes the proof of \cref{lem:DM_soundness}.
\end{proof}
\begin{proof}[Proof of \cref{lem:DM_ZK} (Zero-Knowledge)]
Let $\siml_\CRS$, $\siml_\sen$, and $\Open_\rec$ be the corresponding algorithms for statistical sender's security in the hiding mode of $\Pi_\OT$.
The simulator $\siml=(\siml_0,\siml_1)$ for $\Pi_\DM$ is described below.
\begin{description}
\item[$\siml_0(1^\secpar)$:]
It generates $(\crs_\OT,\td_\OT)\sample \siml_\CRS(1^\secpar)$ and $\pk \sample \lossygen_\LE(1^\secpar)$  
and outputs $\crs_\DM:=(\crs_\OT,\pk)$ and $\td_\DM:=(\crs_\OT,\td_\OT,\pk)$. 
\item[$\siml_1(\td_\DM,\pkey,\statement)$:]
The simulator parses   
$(\crs_\OT,\td_\OT,\pk)\la\td_\DM$ and
$(\rho_P,\ot_1)\la \pkey$ and
does the following.
\begin{enumerate}
    \item Compute $J:=\Open_\rec(\td_\OT,\ot_1)$. 
    Let $S_V:=\{j_1,...,j_5\}\subseteq [N]$ where $J=(j_1,...,j_5)$.
    \item Generate $(\whx,\whz)\sample \bit^{N}\times \bit^{N}$, $R_j\sample \mathcal{R}_\LE$ for $j\in[N]$, 
    $\ct_j:=\enc_\LE(\pk,(\whx_j,\whz_j);R_j)$ for all $j\in[N]$,
    and 
     $\ot_2\sample \siml_\sen(\crs_\OT,\ot_{1},J,\mu_J)$ where 
    $\mu_J:=(\mu_{j_1},...,\mu_{j_5})$ and
    $\mu_{j_i}:=((\whx_{j_i},\whz_{j_i}),R_{j_i})$ for $i\in [5]$.
    
    \item Generate the classical description of the density matrix $\rho_{S_V}:=\siml_{\hist}(\statement,S_V)$
    where $\siml_{\hist}$ is as in \cref{lem:five_local_Hamiltonian}. 
    \item 
    Generate $\widetilde{\rho'}_{\hist}:=\left(\prod_{j\in S_V}X_j^{\whx_j}Z_j^{\whz_j}\right)\rho_{S_V}\left(\prod_{j\in S_V}Z_j^{\whz_j}X_j^{\whx_j}\right)\otimes \frac{I_{[N]\setminus S_V}}{2^{|[N]\setminus S_V|}}$.
    \item Measure $j$-th qubits of $\widetilde{\rho'}_{\hist}$ and $\rho_P$ in the Bell basis for $j\in[N]$, and let $(x_j,z_j)$ be the $j$-th measurement result. 
\item Output $\pi:=(x,z,\{\ct_j\}_{j\in[N]},\ot_2)$ where $x:=x_1\concat x_2\concat...\concat x_N$ and $z:=z_1\concat z_2\concat...\concat z_N$.
\end{enumerate}
\end{description}
We consider the following sequence of modified versions of $\siml_1$,  which take $\witness\in R_L(\statement)$ as an additional input.
\begin{description}
\item[$\siml^{(1)}_1(\td_\DM,\pkey,\statement,\witness)$:]
This simulator works similarly to $\siml_1$ except that it generates the history state $\rho_{\hist}$ for $\ham_\statement$ from $\witness$ instead of $\rho_{S_V}$
in Step $3$, defines  $\rho'_{\hist}:=X^{\whx} Z^{\whz} \rho_{\hist} Z^{\whz}X^{\whx}$ in Step $4$, and uses $\rho'_{\hist}$ instead of  $\widetilde{\rho}'_{\hist}$ in Step $5$. 
 \item[$\siml^{(2)}_1(\td_\DM,\pkey,\statement,\witness)$:]
This simulator works similarly to $\siml^{(1)}_1$ except that in Step $2$, it generates $\ot_2\sample \sender_\OT(\crs_\OT,\ot_{1},(\mu_1,...,\mu_N))$
instead of 
$\ot_2\sample \siml_\sen(\crs_\OT,\ot_{1},J,\vecmu_J)$ where 
    $\mu_j:=((\whx_j,\whz_j),R_j)$ for $j\in [N]$. 
 We note that $\siml^{(2)}_1$ needs not run Step 1 since it does not use $J$ in later steps and thus it does not use  $\td_\OT$.
\end{description}
Let $\ora_P(\crs_\DM,\cdot,\cdot,\cdot)$ and $\ora_S(\td_\DM,\cdot,\cdot,\cdot)$ be as in \cref{def:dual-mode} and $\ora^{(i)}_S(\td_\DM,\cdot,\cdot,\cdot)$ be the oracle that works similarly to $\ora_S(\td_\DM,\cdot,\cdot,\cdot)$ except that it uses $\siml^{(i)}_1$ instead of $\siml_1$ for $i=1,2$.

Then we prove the following claims.

\begin{myclaim}\label[myclaim]{cla:simulation_zero_to_one}
If $\Pi_\LE$ satisfies  lossiness on lossy keys, we have 
\begin{align*}
    \left|\Pr\left[\dist^{\ora_S(\td_\DM,\cdot,\cdot,\cdot)}(\crs_\DM)=1
    \right] -\Pr\left[\dist^{\ora^{(1)}_S(\td_\DM,\cdot,\cdot,\cdot)}(\crs_\DM)=1
    \right]\right|\leq \negl(\secpar)
\end{align*}
where $(\crs_\DM,\td_\DM)\sample \siml_0(1^\secpar)$
for any distinguisher $\dist$ that makes $\poly(\secpar)$ queries of the form $(\pkey=(\rho_P,\ot_1),\statement,\witness)$ for some $\witness\in R_L(\statement)$.
\end{myclaim}
\begin{proof}[Proof of \cref{cla:simulation_zero_to_one}]
Let $\widetilde{\ora}_S(\td_\DM,\cdot,\cdot,\cdot)$ and    $\widetilde{\ora}^{(1)}_S(\td_\DM,\cdot,\cdot,\cdot)$ be oracles that work similarly to $\ora_S(\td_\DM,\cdot,\cdot,\cdot)$ and $\ora^{(1)}_S(\td_\DM,\cdot,\cdot,\cdot)$ except that they generate $\ct_j:=\enc_\LE(\pk,(0,0);R_j)$ instead of $\ct_j:=\enc_\LE(\pk,(\whx_j,\whz_j);R_j)$ for $j\notin S_V$, respectively.
By lossiness on lossy keys of $\Pi_{\LE}$, $\dist$ cannot distinguish $\widetilde{\ora}_S(\td_\DM,\cdot,\cdot,\cdot)$ and    $\widetilde{\ora}^{(1)}_S(\td_\DM,\cdot,\cdot,\cdot)$ from $\ora_S(\td_\DM,\cdot,\cdot,\cdot)$ and $\ora^{(1)}_S(\td_\DM,\cdot,\cdot,\cdot)$ with non-negligible advantage, respectively, noting that no information of $\{R_j\}_{j\notin S_V}$ is given to $\dist$.
When $\dist$ is given either of $\widetilde{\ora}_S(\td_\DM,\cdot,\cdot,\cdot)$ or    $\widetilde{\ora}^{(1)}_S(\td_\DM,\cdot,\cdot,\cdot)$, it has no information on $\{\whx_j,\whz_j\}_{j\notin S_V}$.
Therefore, by \cref{lem:Pauli_mixing}, we have 
\[
\rho'_{\hist}=\left(\prod_{j\in S_V}X_j^{\whx_j}Z_j^{\whz_j}\right)\Tr_{N\setminus S_V}[\rho_{\hist}] \left(\prod_{j\in S_V}Z_j^{\whz_j}X_j^{\whx_j}\right)\otimes \frac{I_{[N]\setminus S_V}}{2^{|[N]\setminus S_V|}}
\]
from the view of $\dist$. 
By \cref{lem:five_local_Hamiltonian}, we have
$\|\rho_{S_V}-\Tr_{[N]\setminus S_V}\rho_\hist\|_{tr}\leq\negl(\secpar)$. 
Therefore, we have $\|\widetilde{\rho'}_\hist-\rho'_\hist\|_{tr}\leq \negl(\secpar)$. 
This means that it cannot distinguish $\widetilde{\ora}_S(\td_\DM,\cdot,\cdot,\cdot)$ and    $\widetilde{\ora}^{(1)}_S(\td_\DM,\cdot,\cdot,\cdot)$ with non-negligible advantage. 
By combining the above,  \cref{cla:simulation_zero_to_one} follows.
\end{proof}


\begin{myclaim}\label[myclaim]{cla:simulation_one_to_two}
If $\Pi_\OT$ satisfies  the second item of statistical sender's security in the hiding mode, we have
\begin{align*}
    \left|\Pr\left[\dist^{\ora^{(1)}_S(\td_\DM,\cdot,\cdot,\cdot)}(\crs_\DM)=1
    \right] -\Pr\left[\dist^{\ora^{(2)}_S(\td_\DM,\cdot,\cdot,\cdot)}(\crs_\DM)=1
    \right]\right|\leq \negl(\secpar)
\end{align*}
where $(\crs_\DM,\td_\DM)\sample \siml_0(1^\secpar)$
for any distinguisher $\dist$ that makes $\poly(\secpar)$ queries. 
\end{myclaim}
\begin{proof}[Proof of \cref{cla:simulation_one_to_two}]
Let $Q=\poly(\secpar)$ be the maximum number of $\dist$'s queries. 
For $i=0,...,Q$, let $\ora^{(1.i)}_S(\td_\DM,\cdot,\cdot,\cdot)$ be the hybrid oracle that works similarly to $\ora^{(2)}_S(\td_\DM,\cdot,\cdot,\cdot)$ for the first $i$ queries and works similarly to $\ora^{(1)}_S(\td_\DM,\cdot,\cdot,\cdot)$ for the rest. 
By a standard hybrid argument, it suffices to prove 
\begin{align}
\begin{split}
    \left|\Pr\left[\dist^{\ora^{(1.i)}_S(\td_\DM,\cdot,\cdot,\cdot)}(\crs_\DM)=1
    \right]
    -\Pr\left[\dist^{\ora^{(1.(i+1))}_S(\td_\DM,\cdot,\cdot,\cdot)}(\crs_\DM)=1
    \right]\right|\leq \negl(\secpar)
\end{split} \label{eq:i_and_iplusone}
\end{align}
where $ (\crs_\DM,\td_\DM)\sample \siml_0(1^\secpar)$ 
for all $i=0,...,Q-1$.
For proving this, for any fixed $(\whx,\whz)\in \bit^N\times \bit^N$ and $\{R_j\}_{j\in[N]}\in \mathcal{R}_\LE^N$, 
we consider the following adversary $\A=(\A_0,\A_1)$ against the second item of statistical sender's security in the hiding mode of $\Pi_\OT$.
\begin{description}
\item[$\A_0(\crs_\OT,\td_\OT)$:]
It generates $\pk\sample \lossygen(1^\secpar)$, 
gives $\crs_\DM:=(\crs_\OT,\pk)$ to $\dist$ as input and runs it until it makes $(i+1)$-th query where $\A_0$ simulates responses to
the first $i$ queries similarly to $\ora^{(2)}_S(\td_\DM,\cdot,\cdot,\cdot)$ where $\td_\DM=(\crs_\OT,\td_\OT,\pk)$.
Let $(\pkey,\statement,\witness)$ be $\dist$'s $(i+1)$-th query. 
$\A_0$ parses $(\rho_P,\ot_1)\la\pkey$ and computes the history state $\rho_\hist$ for $\ham_\statement$ from $\witness$. 
It outputs $\ot_1$ and $\st_\A:=(\rho_P,\rho_\hist)$. 
\item[$\A_1(\st_\A=(\rho_P,\rho_\hist),\ot_2)$:]
It generates $\ct_j:=\enc_\LE(\pk,(\whx_j,\whz_j);R_j)$ for all $j\in [N]$  and
 $\rho'_{\hist}:=X^{\whx} Z^{\whz} \rho_{\hist} Z^{\whz}X^{\whx}$, 
    measures $j$-th qubits of $\rho'_{\hist}$ and $\rho_P$ in the Bell basis for $j\in[N]$,
    lets $(x_j,z_j)$ be the $j$-th measurement result, 
and returns $\pi:=(x,z,\{\ct_j\}_{j\in[N]},\ot_2)$ to $\dist$
as the response of the oracle to the $(i+1)$-th query  
where $x:=x_1\concat x_2\concat...\concat x_N$ and $z:=z_1\concat z_2\concat...\concat z_N$.
$\A_1$ runs the rest of the execution of $\dist$ by simulating the oracle similarly to $\ora_S^{(1)}(\td_\DM,\cdot,\cdot,\cdot)$. 
Finally, $\A_1$ outputs whatever $\dist$ outputs. 
\end{description}
Let $\vecmu:=(((\whx_1,\whz_1),R_1),...,((\whx_N,\whz_N),R_N))$.
If $\ot_2$ is generated as $\ot_2\sample \sender(\crs_\OT,\allowbreak \ot_1,\vecmu)$, then $\A$ perfectly simulates the execution of $\dist^{\ora^{(1.i)}_S(\td_\DM,\cdot,\cdot,\cdot)}(\crs_\DM)$ conditioned on the fixed $(\whx,\whz)$ and $\{R_j\}_{j\in[N]}$.
On the other hand, if $\ot_2$ is generated as 
$J:=\Open_\rec(\td_\OT,\ot_1)$ and
$\ot_2\sample \siml_\sen(\crs_\OT,\ot_1,J,\vecmu_J)$,  then $\A$ perfectly simulates the execution of $\dist^{\ora^{(1.(i+1))}_S(\td_\DM,\cdot,\cdot,\cdot)}(\crs_\DM)$ conditioned on the fixed $(\whx,\whz)$ and $\{R_j\}_{j\in[N]}$.
Therefore, averaging over the random choice of $(\whx,\whz)$ and $\{R_j\}_{j\in[N]}$, the l.h.s. of \cref{eq:i_and_iplusone} can be upper bounded by the average of the advantage of $\A$ to distinguish the two cases, which is negligible by the assumption.
This completes the proof of \cref{cla:simulation_one_to_two}.
\end{proof}

\begin{myclaim}\label[myclaim]{cla:simulation_two_to_real}
If $\Pi_\OT$ satisfies the first item of statistical sender's security in the hiding mode, 
We have 
\begin{align*}
    &\left|\Pr\left[\dist^{\ora^{(2)}_S(\td_\DM,\cdot,\cdot,\cdot)}(\crs_\DM)=1:
    \begin{array}{c}
           (\crs_\DM,\td_\DM)\sample \siml_0(1^\secpar) 
    \end{array}
    \right]\right.\\     &-\left.\Pr\left[\dist^{\ora_P(\crs_\DM,\cdot,\cdot,\cdot)}(\crs_\DM)=1:
    \begin{array}{c}
          \crs_\DM\sample \crsgen_\DM(1^\secpar,\hiding)
    \end{array}
    \right]\right|\leq \negl(\secpar)
\end{align*}
\end{myclaim}
\begin{proof}[Proof of \cref{cla:simulation_two_to_real}]
For any $(\crs_\DM,\td_\DM)\sample \siml_0(1^\secpar)$, $\pkey$, $\statement$, and $\witness$,  we have   
\[\ora^{(2)}_S(\td_\DM,\pkey,\statement,\witness)=\ora_P(\crs_\DM,\pkey,\statement,\witness)\]
observing that $\siml^{(2)}_1$ works in the exactly the same way as the honest proving algorithm.     
Moreover, we can see that the distributions of $\crs_\DM$ generated by $\siml_0(1^\secpar)$ and $\crsgen_\DM(1^\secpar,\hiding)$ are statistically indistinguishable by the first item of statistical sender's security in the hiding mode of $\Pi_\OT$. 
Therefore \cref{cla:simulation_two_to_real} follows.
\end{proof}

By combining 
\cref{cla:simulation_zero_to_one,cla:simulation_one_to_two,cla:simulation_two_to_real}, We can complete the proof of \cref{lem:DM_ZK}. 
\end{proof}

%% file: Fiat-Shamir.tex
\section{CV-NIZK via Fiat-Shamir Transformation}\label{sec:Fiat-Shamir}
In this section, we construct CV-NIZK in the quantum random oracle model via the
Fiat-Shamir transformation.

\subsection{Definition}\label{sec:def_for_FS}
We give a formal definition of CV-NIZK in the \QROVP model.

\begin{definition}[CV-NIZK in the \QROVP Model]\label{def:qrovp_nizk}
A CV-NIZK for a $\QMA$ promise problem $L=(L_\yes,L_\no)$ in the \QROVP model 
w.r.t. a random oracle distribution $\ROdist$ 
consists of algorithms $\Pi=(\preprocess,\allowbreak\prove,\verify)$ with the following syntax:
\begin{description}
\item[$\preprocess(1^\secpar)$:]
This is a QPT algorithm that takes the security parameter $1^\secpar$ as input, and outputs a quantum proving key $\pkey$ and a classical  verification key $\vkey$.
We note that this algorithm is supposed to be run by the verifier as preprocessing, and $\pkey$ is supposed to be sent to the prover while $\vkey$ is supposed to be kept on verifier's side in secret. 
We also note that they can be used only once  and cannot be reused. 
\item[$\prove^{H}(\pkey,\statement,\witness^{\otimes k})$:] This is a QPT algorithm that 
is given quantum oracle access to the random oracle $H$. It
takes  
the proving key $\pkey$, a statement $\statement$, and $k=\poly(\secpar)$ copies $\witness^{\otimes k}$ of a witness $\witness\in R_L(\statement)$ as input, and outputs a classical proof $\pi$.
\item[$\verify^{H}(\vkey,\statement,\pi)$:]
This is a PPT algorithm that 
is given classical oracle access to the random oracle $H$. It  
takes 
the verification key $\vkey$, a statement $\statement$, and a proof $\pi$ as input, and outputs $\top$ indicating acceptance or $\bot$ indicating rejection. 
\end{description}
We require $\Pi$ to satisfy the following properties. 

\medskip
\noindent \underline{
\textbf{Completeness.}
}
For all
$\statement\in L_\yes\cap \bit^\secpar$, and $\witness\in R_L(\statement)$, we have 
\begin{align*}
    \Pr\left[
    \verify^H(\vkey,\statement,\pi)=\top 
    :
    \begin{array}{c}
          H\sample \ROdist \\
         (\pkey,\vkey) \sample \preprocess(1^\secpar)\\
         \pi \sample \prove^H(\pkey,\statement,\witness^{\otimes k})
    \end{array}
    \right]
    \geq 1-\negl(\secpar).
\end{align*}

\medskip
\noindent \underline{
\textbf{Adaptive Statistical Soundness.}
}
For all adversaries $\A$ that make at most $\poly(\secpar)$ quantum random oracle queries, we have 
\begin{align*}
    \Pr\left[
    \statement\in L_\no \land \verify^H(\vkey,\statement,\pi)=\top
    :
    \begin{array}{c}
          H\sample \ROdist \\
         (\pkey,\vkey) \sample \preprocess(1^\secpar)\\
       (\statement,\pi) \sample \A^H(\pkey)
    \end{array}
    \right]
    \leq \negl(\secpar).
\end{align*}

\medskip
\noindent \underline{
\textbf{Adaptive Multi-Theorem  Zero-Knowledge.}
}
For defining the zero-knowledge property in the QROM, we define the syntax of a simulator in the QROM following \cite{EC:Unruh15}. A simulator is given quantum access to the random oracle $H$ and classical access to reprogramming oracle $\reprogram$.
When the simulator queries $(x,y)$ to $\reprogram$, the random oracle $H$ is reprogrammed so that $H(x):=y$ while keeping the values on other inputs unchanged. 
Then the adaptive multi-theorem zero-knowledge property is defined as follows:

There exists  a QPT simulator $\siml$ with the above syntax such that for any QPT distinguisher $\dist$, we have 
\begin{align*}
    &\left|\Pr\left[\dist^{H,\ora_P^H(\cdot,\cdot,\cdot)}(1^\secpar)=1:
    \begin{array}{c}
          H \sample \ROdist
    \end{array}
    \right]\right.\\     &-\left.\Pr\left[\dist^{H,
    \ora_S^{H,\reprogram}(\cdot,\cdot,\cdot)}(1^\secpar)=1:
    \begin{array}{c}
          H \sample \ROdist
    \end{array}
    \right]\right|\leq \negl(\secpar)
\end{align*}
where  
$\dist$'s queries to the second oracle should be of the form
$(\pkey,\statement,\witness^{\otimes k})$ 
where $\witness\in R_L(\statement)$ and $\witness^{\otimes k}$ is unentangled with $\dist$'s internal registers, \footnote{
We remark that $\pkey$ is allowed to be entangled  with $\dist$'s internal registers unlike $\witness^{\otimes k}$.
See also footnote \ref{footnote:unentangled}.}
$\ora_P^H(\pkey,\statement,\witness^{\otimes k})$  
returns $\prove^H(\pkey,\statement,\witness^{\otimes k})$, and
$\ora_S^{H,\reprogram}(\pkey,\statement,\witness^{\otimes k})$ returns 
 $\siml^{H,\reprogram}(\pkey,\statement)$.
\end{definition}

\begin{remark}
Remark that the ``multi-theorem" zero-knowledge does not mean that a preprocessing can be reused many times. It rather means that a single random oracle can be reused as long as a fresh preprocessing is run every time. This is consistent to the definition in the \CRSVP model (\cref{def:dual-mode}) if we think of the random oracle as replacement of CRS.
\end{remark}

\subsection{Building Blocks}
We use the two cryptographic primitives, a non-interactive commitment scheme and
a $\Sigma$-protocol with quantum preprocessing, for our construction.

\ifnum\submission=1
The definition of  a non-interactive commitment scheme is given 
in \cref{sec:def_commitment}.
\else
\input{non-interactive_commitment}
\fi

\begin{definition}[$\Sigma$-protocol with Quantum Preprocessing]\label{def:sigma_q_prepro}
A $\Sigma$-protocol with quantum preprocessing for a $\QMA$ promise problem $L=(L_\yes,L_\no)$ 
consists of algorithms $\Pi=(\preprocess,\prove_1,\allowbreak\verify_1,\prove_2,\verify_2)$ with the following syntax:
\begin{description}
\item[$\preprocess(1^\secpar)$:]
This is a QPT algorithm that takes the security parameter $1^\secpar$ as input, and outputs a quantum proving key $\pkey$ and a classical verification key $\vkey$.
We note that this algorithm is supposed to be run by the verifier as preprocessing, and $\pkey$ is supposed to be sent to the prover while $\vkey$ is supposed to be kept on verifier's side in secret. 
We also note that they can be used only once  and cannot be reused. 
\item[$\prove_1(\pkey,\statement,\witness^{\otimes k})$:] This is a QPT algorithm that 
takes  
the proving key $\pkey$, a statement $\statement$, and $k=\poly(\secpar)$ copies $\witness^{\otimes k}$ of a witness $\witness\in R_L(\statement)$ as input, and outputs 
a classical message $\msg_1$ and a state $\st$.
\item[$\verify_1(1^\secpar)$:]
This is a PPT algorithm that  
takes 
the security parameter $1^\secpar$, and outputs a classical message $\msg_2$, which is uniformly sampled from a certain set.
\item[$\prove_2(\st,\msg_2)$:] This is a QPT algorithm that 
takes  
the state $\st$ and the message $\msg_2$ as input,
and outputs a classical message $\msg_3$.
\item[$\verify_2(\vkey,\statement,\msg_1,\msg_2,\msg_3)$:]
This is a PPT algorithm that  
takes 
the verification key $\vkey$, the statement $\statement$, and classical messages $\msg_1,\msg_2,\msg_3$ as input,
and outputs $\top$ indicating acceptance or $\bot$ indicating rejection.
\end{description}
We require $\Pi$ to satisfy the following properties. 

\medskip
\noindent \underline{
\textbf{$c$-Completeness.}
}
For all
$\statement\in L_\yes\cap \bit^\secpar$, and $\witness\in R_L(\statement)$, we have 
\begin{align*}
    \Pr\left[
    \verify_2(\vkey,\statement,\msg_1,\msg_2,\msg_3)=\top 
    :
    \begin{array}{c}
         (\pkey,\vkey) \sample \preprocess(1^\secpar)\\
         (\msg_1,\st) \sample \prove_1(\pkey,\statement,\witness^{\otimes k})\\
         \msg_2 \sample \verify_1(1^\secpar)\\
         \msg_3 \sample \prove_2(\st,\msg_2)\\
    \end{array}
    \right]
    \geq c.
\end{align*}

\medskip
\noindent \underline{
\textbf{(Adaptive Statistical) $s$-soundness.}
}
For all adversary $(\A_1,\A_2)$, we have 
\begin{align*}
    \Pr\left[
    \statement\in L_\no \land \Sigma.\verify_2(\vkey,\statement,\msg_1,\msg_2,\msg_3)=\top
    :
    \begin{array}{c}
         (\pkey,\vkey) \sample \preprocess(1^\secpar)\\
       (\statement,\st,\msg_1) \sample \A_1(\pkey)\\
       \msg_2 \sample \verify_1(1^\secpar)\\
       \msg_3 \sample \A_2(\st,\msg_2)\\
    \end{array}
    \right]
    \leq s.
\end{align*}

\medskip
\noindent \underline{
\textbf{Special  Zero-Knowledge.}}
There exists a QPT algorithm $\siml$ such that 
for any $\statement\in L_\yes$, $\witness\in R_L(\statement)$, $\msg_2$, and QPT adversary $(\A_1,\A_2)$, we have 
\begin{align*}
\left|
\begin{array}{cc}
    &\Pr\left[
    \A_2(\st_\A,\statement,\msg_1,\msg_2,\msg_3)=1
    :
    \begin{array}{c}
         (\pkey,\st_\A) \sample \A_1(1^\secpar)\\
       (\msg_1,\st) \sample \prove_1(\pkey,\statement,\witness^{\otimes k})\\
       \msg_3 \sample \prove_2(\st,\msg_2)\\ 
    \end{array}
    \right]\\
    - &\Pr\left[
    \A_2(\st_\A,\statement,\msg_1,\msg_2,\msg_3)=1
    :
    \begin{array}{c}
         (\pkey,\st_\A) \sample \A_1(1^\secpar)\\
       (\msg_1,\msg_3) \sample \siml(\pkey,\statement,\msg_2)\\
    \end{array}
    \right]
    \end{array}
    \right|
    \leq \negl(\secpar).
\end{align*}

\medskip
\noindent \underline{
\textbf{High Min-Entropy.}}
$\prove_1$ can be divided into the ``quantum part" and ``classical part" as follows: 
\begin{description}
\item[$\prove_1^Q(\pkey,\statement,\witness^{\otimes k})$:]
This is a QPT algorithm that outputs a classical string $\st'$. 
\item[$\prove_1^C(\st')$:]
This is a PPT algorithm that outputs $\msg_1$ and $\st$. 
\end{description}
Moreover, for any $\st'$ generated by $\prove_1^Q$, 
we have 
\begin{align*}
    \max_{\msg_1^*}\Pr[\prove_1^C(\st')=\msg_1^*]=\negl(\secpar).
\end{align*}  
\end{definition}

\ifnum\submission=1
There are some remarks on the definition. See \cref{sec:remark_definition_sigma}.
\else
\input{remark_definition_sigma}
\fi

\begin{lemma}[Gap Amplification for $\Sigma$-protocol with quantum preprocessing]\label{lem:sigma_amplification}
If there exists a $\Sigma$-protocol with quantum preprocessing for a promise problem $L$ that satisfies $c$-completeness, $s$-soundness, special zero-knowledge, and high min-entropy for some $0<s<c<1$ such that $c-s>1/\poly(\secpar)$, 
then there exists a $\Sigma$-protocol with quantum preprocessing for $L$ with $(1-\negl(\secpar))$-completeness, $\negl(\secpar)$-soundness, special zero-knowledge, and high min-entropy.
\end{lemma}
\begin{proof}
It is clear that the parallel repetition can amplify the completeness-soundness gap, 
and that
the high min-entropy is preserved under the parallel repetition. 
We can also show that parallel repetition preserves the special zero-knowledge property by a standard hybrid argument. 
\end{proof}

\begin{theorem}\label{thm:sigma}
If a non-interactive commitment scheme exists, then there exists a $\Sigma$-protocol with quantum preprocessing for $\QMA$.
\end{theorem}
As mentioned in \cref{sec:def_for_FS}, a non-interactive commitment scheme unconditionally exists in the QROM. Therefore, the above theorem implies the following corollary.
\begin{corollary}\label{cor:sigma}
There exists a $\Sigma$-protocol  with quantum preprocessing for $\QMA$ in the QROM.
\end{corollary}

\ifnum\submission=1
Proof of \cref{thm:sigma} is given in \cref{sec:proof_sigma}.
\else
\input{proof_sigma}
\fi

\subsection{Construction}\label{sec:construction_QRO}
In this section, we construct a CV-NIZK in the \QROVP  model. 
As a result, we obtain the following theorem.
\begin{theorem}\label{thm:fiat-shamir}
There exists a CV-NIZK for $\QMA$ in the \QROVP  model. 
\end{theorem}

Let $L=(L_\yes,L_\no)$ be a $\QMA$ promise problem, $H$ be a random oracle, and 
$\Pi_{\Sigma}=(\Sigma.\preprocess,\allowbreak\Sigma.\prove_1,\Sigma.\verify_1,\Sigma.\prove_2,\Sigma.\verify_2)$ be a $\Sigma$-protocol with quantum preprocessing
(with $(1-\negl(\secpar))$-completeness and $\negl(\secpar)$-soundness).
Then our CV-NIZK in the \QROVP model
$\Pi_{\mathsf{QRO}}=(\preprocess_{\mathsf{QRO}},\prove_{\mathsf{QRO}},\verify_{\mathsf{QRO}})$ for $L$ is described in \cref{fig:CV-NIZK-FiatShamir}.

\begin{figure}[t]
\rule[1ex]{\textwidth}{0.5pt}
\begin{description}
\item[$\preprocess_{\mathsf{QRO}}(1^\secpar)$:]
It runs $\Sigma.\preprocess(1^\secpar)\ra (\Sigma.\vkey,\Sigma.\pkey)$, and outputs
$\vkey:=\Sigma.\vkey$
and $\pkey:=\Sigma.\pkey$.

\item[$\prove^H_{\mathsf{QRO}}(\pkey,\statement,\witness^{\otimes k})$:]
It parses $\Sigma.\pkey\la \pkey$, 
and runs $\Sigma.\prove_1(\pkey,\statement,\witness^{\otimes k})\ra (\msg_1,\st)$.
It computes $\msg_2:=H(\statement,\msg_1)$.
It runs $\Sigma.\prove_2(\st,\msg_2)\ra \msg_3$.
It outputs $\pi:=(\msg_1,\msg_3)$.

\item[$\verify^H_{\mathsf{QRO}}(\vkey,\statement,\pi)$:]
It
parses 
$\Sigma.\vkey\la \vkey$ and 
$(\msg_1,\msg_3)\la \pi$.   
It computes $\Sigma.\verify_2(\vkey,\statement,\msg_1,H(\statement,\msg_1),\msg_3)$.
If the output is $\bot$, it outputs $\bot$.
If the output is $\top$, it outputs $\top$.
\end{description} 
\rule[1ex]{\textwidth}{0.5pt}
\hspace{-10mm}
\caption{CV-NIZK 
in the \QROVP model
$\Pi_{\mathsf{QRO}}$.}
\label{fig:CV-NIZK-FiatShamir}
\end{figure}

\begin{lemma}\label{lem:QRO_completeness_soundness}
$\Pi_{\mathsf{QRO}}$ satisfies $(1-\negl(\secpar))$-completeness and
adaptive $\negl(\secpar)$-soundness.
\end{lemma}

\ifnum\submission=1
Its proof is given in \cref{sec:proof_QRO_completeness_soundness}.
\else
\input{FS_completeness_soundness}
\fi

\begin{lemma}\label{lem:QRO_ZK}
$\Pi_{\mathsf{QRO}}$ satisfies adaptive multi-theorem zero-knowledge property.
\end{lemma}

\ifnum\submission=1
Its proof is given in \cref{sec:proof_QRO_ZK}.
\else
\input{FS_ZK}
\fi

\paragraph{Shared Bell-pair model.}
Remark that the verifier of  $\Pi_{\mathsf{QRO}}$ just sends a state $\rho_P:=\bigotimes_{j=1}^N(U(W_j)|m_j\rangle)$ for
$(W_1,...,W_N)\sample \{X,Y,Z\}^{N}$ and $(m_1,...,m_N)\sample \{0,1\}^{N}$ while keeping $(
W_1,...,W_N,m_1,...,m_N)$ as a verification key. 
This step can be done in a non-interactive way if $N$ Bell-pairs are a priori shared between the prover and verifier.
That is, the verifier can measure his halves of Bell pairs in a randomly chosen bases $(W_1,...,W_N)$ to get measurement outcomes $(m_1,...,m_N)$. 
Apparently, this does not harm either of soundness or zero-knowledge since 
the protocol is the same as $\Pi_{\mathsf{QRO}}$ from the view of the prover and 
the malicious verifier's power is just weaker than that in $\Pi_{\mathsf{QRO}}$ in the sense that it cannot control the quantum state to be sent to the prover.
Thus, we obtain the following theorem.
\begin{theorem}\label{thm:fiat-shamir_Bell_pair}
There exists a CV-NIZK for $\QMA$ in the QRO + shared Bell pair model. 
\end{theorem}
See \cref{app:bell} for the formal definition of CV-NIZK for $\QMA$ in the QRO + shared Bell pair model.

%% file: non-interactive_commitment.tex
\begin{definition}[Non-interactive commitment scheme]\label{def:commitment_QROM}
A non-interactive commitment scheme with the message space $\mathcal{M}$   
is a tuple of PPT algorithms $(\commit,\verify)$ with the following syntax:  
\begin{description}
\item[$\commit(1^\secpar, m):$]
It takes the security parameter $1^\secpar$ and a message $m\in \mathcal{M}$ as input, and outputs a commitment $\com$ and a decommitment $\decom$. 
\item[$\verify(1^\secpar, m,\com, \decom):$]
It takes the security parameter $1^\secpar$,  a message $m\in \mathcal{M}$, commitment $\com$, and decommitment $\decom$ as input, and outputs $\top$ indicating acceptance or $\bot$ indicating rejection.  
\end{description}
We require a non-interactive commitment scheme to satisfy the following properties: 

\medskip
\noindent \underline{
\textbf{Perfect Correctness.}
}
For any $\secpar\in \mathbb{N}$ and $m\in \mathcal{M}$,  we have 
\begin{align*}
    \Pr[\verify(1^\secpar,m,\com,\decom)=\top:(\com,\decom)\sample \commit(1^\secpar,m)]=1.
\end{align*}

\medskip
\noindent \underline{
\textbf{Perfect Binding.}
}
For all $\secpar \in \mathbb{N}$, there do not exist $m,m',\com,\decom,\decom'$ such that $m\neq m'$ and $\verify(1^\secpar,m,\com,\decom)=\verify(1^\secpar,m',\com,\decom')=\top$.  

\noindent \underline{
\textbf{Computational Hiding.}
}
For any QPT adversary $\A$  and messages $m_0,m_1$, we have 
\begin{align*}
    \left|
    \begin{array}{cc}
    \Pr[\A(\com)=1:
    (\com,\decom)\sample \commit(1^\secpar,m_0)]\\
    -\Pr[\A(\com)=1: 
    (\com,\decom)\sample \commit(1^\secpar,m_1)]
        \end{array}
    \right|=\negl(\secpar).
\end{align*}
\end{definition}

It is known that a non-interactive commitment scheme exists assuming the existence of injective one-way functions or perfectly correct public key encryption (or more generally key exchange protocols) \cite{EPRINT:LomSch19}. 
In the QROM, a non-interactive commitment scheme exists without any assumption since a random oracle with a sufficiently large range is injective with overwhelming probability over the choice of the random oracle and hard to invert even with quantum access to the oracle \cite{BBBV}. In our constructions and security proofs, we use a non-interactive commitment scheme in the standard model. This is for notational simplicity and also for  clarifying that the full power of random oracles is not needed for this component.  We stress that this does not mean that we assume an additional assumption for our construction of NIZK since a non-interactive commitment scheme unconditionally exists in the QROM as mentioned above and all security proofs work similarly with a non-interactive commitment scheme in the QROM.

%% file: remark_definition_sigma.tex
\begin{remark}[On Soundness]
Some existing works require a $\Sigma$-protocol to satisfy special soundness, which means that one can extract a witness from two accepting transcripts whose first messages are idential and the second messages are different. This property is often useful for achieving proof of knowledge. 
We do not require special soundness since we do not consider proof of knowledge in this paper and our construction does not seem to satisfy special soundness. 
\end{remark}
\begin{remark}[On Zero-Knowledge]
Our definition of the zero-knowledge property is based on the special honest-verifier zero-knowledge often required for classical $\Sigma$-protocol without preprocessing. 
However, our definition considers a partially malicious verifier that maliciously runs the preprocessing, which is a crucial difference from the classical case. 
This is why we call this property as special zero-knowledge rather than special honest-verifier zero-knowledge.
Note that special zero-knowledge property is weaker than the standard zero-knowledge property for general interactive protocols since the standard zero-knowledge considers malicious verifiers that adaptively choose $\msg_2$ rather than fixing it. 
\end{remark}
\begin{remark}[On High Min-Entropy]
We require the high min-entropy property because this property is needed in the proof of adaptive multi-theorem zero-knowledge property of the NIZK obtained by the Fiat-Shamir transform in \cref{sec:construction_QRO}. 
The property requires two requirements: the first is about the structure of $\prove_1$ and the second is that $\msg_1$ has a high min-entropy.
The latter is needed even for Fiat-Shamir transform for $\Sigma$-protocols for $\NP$ (e.g., see \cite{EC:Unruh15}). 
On the other hand, the former is unique to our work, and we do not know if this is inherent. However, since this requirement makes the security proof of our NIZK easier and our construction of $\Sigma$-protocol with quantum preprocessing satisfies this property, we include this as a default requirement.
\end{remark}

%% file: proof_sigma.tex
\begin{proof}[Proof of \cref{thm:sigma}]
Let $L=(L_\yes,L_\no)$ be a $\QMA$ promise problem, and $\ham_\statement$, $N$, $M$, $p_{i}$, $s_{i}$, $P_i$, $\alpha$, $\beta$, and $\rho_\hist$ be as in \cref{lem:five_local_Hamiltonian} for the promise problem $L$.  
We let $N':=3^5\sum_{i=1}^{5}{N \choose i}$ similarly to \cref{lem:NIZK_completeness_soundness}.
Let $\Pi_{\mathsf{comm}}=(\commit_{\mathsf{comm}},\verify_{\mathsf{comm}})$ be a non-interactive commitment scheme as defined in \cref{def:commitment_QROM}.
Then our $\Sigma$-protocol with quantum preprocessing
$\Pi_{\Sigma}=(\Sigma.\preprocess,\Sigma.\prove_1,\Sigma.\verify_1,\allowbreak\Sigma.\prove_2,\Sigma.\verify_2)$ for $L$ is described in \cref{fig:sigma}.

\begin{figure}[hpt]
\rule[1ex]{\textwidth}{0.5pt}
\begin{description}
\item[$\Sigma.\preprocess(1^\secpar)$:]
It chooses
$(W_1,...,W_N)\sample \{X,Y,Z\}^{N}$ and $(m_1,...,m_N)\sample \{0,1\}^{N}$,
and outputs a proving key $\pkey:=\rho_P:=\bigotimes_{j=1}^N(U(W_j)|m_j\rangle)$ 
and a verification key $\vkey:=(
W_1,...,W_N,m_1,...,m_N)$.

\item[$\Sigma.\prove_1(\pkey,\statement,\witness)$:]
It parses $\rho_P\la \pkey$, 
and generates the history state $\rho_\hist$ for $\ham_\statement$ from $\witness$. 
It measures $j$-th qubits of $\rho_\hist$ and $\rho_P$ in the Bell basis for $j\in [N]$.    
Let $x:=x_1\concat x_2\concat...\concat x_N$, and $z:=z_1\concat z_2\concat...\concat z_N$ where 
$(x_j,z_j)\in\{0,1\}^2$ denotes the outcome of $j$-th measurement.
It computes $\commit_{\mathsf{comm}}(1^\secpar,(x_j,z_j))\ra (\com_j,d_j)$ for each $j\in[N]$.
It outputs a classical message $\msg_1:=\{\com_j\}_{j\in[N]}$ and the state $\st$, which is its entire final state.

\item[$\Sigma.\verify_1(1^\secpar)$:]
It chooses a subset $S\subset[N]$ such that $1\le|S|\le5$ uniformly at random, and outputs $\msg_2:=S$.

\item[$\Sigma.\prove_2(\st,\msg_2)$:]
It parses $\st$ as the final entire state of $\Sigma.\prove_1$ and
$\msg_2\la S$.
It outputs
$\msg_3:=(\{d_j\}_{j\in S},\{x_j,z_j\}_{j\in S})$.   

\item[$\Sigma.\verify_2(\vkey,\statement,\msg_1,\msg_2,\msg_3)$:]
It parses
$(W_1,...,W_N,m_1,...,m_N)\la \vkey$,  
$\{\com_j\}_{j\in[N]}\la\msg_1$,
$S\la\msg_2$, and $(\{d_j\}_{j\in S},\{x_j,z_j\}_{j\in S})\la \msg_3$.
It computes $\verify_{\mathsf{comm}}(1^\secpar,(x_j,z_j),\com_j,d_j)$ for all $j\in S$. If not all outputs are $\top$, it outputs $\bot$ and aborts.
It
chooses $i\in [M]$ according to the probability distribution defined by $\{p_{i}\}_{i\in[M]}$ (i.e., chooses $i$ with probability $p_{i}$).
Let
\begin{eqnarray*}
S_i:=\{j\in[N]~|~\mbox{$j$th Pauli operator of $P_i$ is not $I$}\}.
\end{eqnarray*}
We note that we have $1\leq |S_i|\leq 5$ by the $5$-locality of $\ham_\statement$. 
We say that $P_i$ is consistent to $(S,\{W_j\}_{j\in S})$  if and only if 
$S_i = S$ and
the $j$th Pauli operator of
$P_i$ is $W_j$ for all
$j\in S_i$.    
If $P_i$ is not consistent to $(S,\{W_j\}_{j\in S})$,  it outputs $\top$.
If $P_i$ is consistent to 
$(S,\{W_j\}_{j\in S})$,  
it flips a biased coin that heads with probability $1-3^{|S_i|-5}$.
If heads, it outputs $\top$.
If tails,
it defines
\begin{eqnarray*}
m_{j}':= 
\left\{
\begin{array}{cc}
m_{j}\oplus x_{j}&(W_j=Z),\\
m_{j}\oplus z_{j}&(W_j=X),\\
m_{j}\oplus x_{j}\oplus z_j&(W_j=Y)
\end{array}
\right.
\end{eqnarray*}
for $j\in S_i$, and outputs $\top$ if  $(-1)^{\bigoplus_{j\in S_i}m'_{j}}=-s_{i}$ and $\bot$ otherwise. 
\end{description} 
\rule[1ex]{\textwidth}{0.5pt}
\hspace{-10mm}
\caption{$\Sigma$-protocol with quantum preprocessing $\Pi_\Sigma$.}
\label{fig:sigma}
\end{figure}

\begin{figure}[t]
\rule[1ex]{\textwidth}{0.5pt}
\begin{description}
\item[$\Sigma.\preprocess(1^\secpar)$:]
It generates $N$ Bell-pairs between registers $P$ and $V$. Let $\rho_P$
and $\rho_V$ be quantum states in registers $P$ and $V$, respectively. 
It outputs
a proving key $\pkey:=\rho_P$ and a verification key $\vkey:=\rho_V$.

\item[$\Sigma.\prove_1(\pkey,\statement,\witness)$:]
The same as that of $\Pi_\Sigma$.

\item[$\Sigma.\verify_1(1^\secpar)$:]
The same as that of $\Pi_\Sigma$.

\item[$\Sigma.\prove_2(\st,\msg_2)$:]
The same as that of $\Pi_\Sigma$.

\item[$\Sigma.\verify_2(\vkey,\statement,\msg_1,\msg_2,\msg_3)$:]
It parses
$\rho_V\la \vkey$,  
$\{\com_j\}_{j\in[N]}\la\msg_1$,
$S\la\msg_2$, and $(\{d_j\}_{j\in S},\{x_j,z_j\}_{j\in S})\la \msg_3$.
It computes $\verify_{\mathsf{comm}}(1^\secpar,(x_j,z_j),\com_j,d_j)$ for all $j\in S$. If not all outputs are $\top$, it outputs $\bot$ and aborts.
It chooses $(W_1,...,W_N)\sample\{X,Y,Z\}^N$.
It generates $(\prod_{j\in S}X_j^{x_j}Z_j^{z_j})\rho_V(\prod_{j\in S}X_j^{x_j}Z_j^{z_j})$, and measures its $j$th qubit in the $W_j$-basis for every $j\in [N]$.
Let $m_j\in\bit$ be the measurement result for the $j$th qubit.
It chooses $i\in [M]$ according to the probability distribution defined by $\{p_{i}\}_{i\in[M]}$ (i.e., chooses $i$ with probability $p_{i}$).
Let
\begin{eqnarray*}
S_i:=\{j\in[N]~|~\mbox{$j$th Pauli operator of $P_i$ is not $I$}\}.
\end{eqnarray*}
We note that we have $1\leq |S_i|\leq 5$ by the $5$-locality of $\ham_\statement$. 
We say that $P_i$ is consistent to $(S,\{W_j\}_{j\in S})$  if and only if 
$S_i = S$ and
the $j$th Pauli operator of
$P_i$ is $W_j$ for all
$j\in S_i$.    
If $P_i$ is not consistent to $(S,\{W_j\}_{j\in S})$,  it outputs $\top$.
If $P_i$ is consistent to 
$(S,\{W_j\}_{j\in S})$,  
it flips a biased coin that heads with probability $1-3^{|S_i|-5}$.
If heads, it outputs $\top$ and aborts.
If tails, it outputs $\top$ if $(-1)^{\bigoplus_{j\in S_i}m_{j}}=-s_{i}$ and $\bot$ otherwise. 
\end{description} 
\rule[1ex]{\textwidth}{0.5pt}
\hspace{-10mm}
\caption{The virtual protocol $\Pi_{\Sigma}'$ for $\Sigma$-protocol with quantum preprocessing $\Pi_\Sigma$.}
\label{fig:sigma-v}
\end{figure}

\begin{lemma}\label{lem:sigma_completeness}
$\Pi_{\Sigma}$ satisfies $\left(1-\frac{\alpha}{N'}\right)$-completeness and
$\left(1-\frac{\beta}{N'}+\negl(\secpar)\right)$-soundness.
\end{lemma}
\begin{proof}
Let us consider the virtual protocol $\Pi_{\Sigma}'$ given in Fig.~\ref{fig:sigma-v}.
Due to \cref{lem:XZ_before_measurement} and the fact that the measurements by the prover and the verifier are commute with each other,
the acceptance probability in $\Pi_\Sigma'$ is equal to that in $\Pi_\Sigma$.
We therefore have only to show
the $\left(1-\frac{\alpha}{N'}\right)$-completeness and
and $\left(1-\frac{\beta}{N'}+\negl(\secpar)\right)$-soundness for the virtual protocol $\Pi_\Sigma'$.

First let us show the completeness.
If the prover is honest, it is clear that the history state with Pauli byproducts,
$(\prod_{j\in[N]}X_j^{x_j}Z_j^{z_j})
\rho_\hist
(\prod_{j\in[N]}X_j^{x_j}Z_j^{z_j})$,
is teleported to the verifier, and the verifier
can correct the byproducts on $S$ (with probability one
from the prefect completeness of the commitment scheme).
From \cref{lem:prob_and_energy} and \cref{lem:five_local_Hamiltonian},
and the fact that the probability that $P_i$ is consistent to $(S,\{W_j\}_{j\in S})$ and the coin tails is $1/N'$, we obtain
the acceptance probability in $\Pi_\Sigma'$ when $\statement\in L_\yes$ to be
\begin{align*}
\Big(1-\frac{1}{N'}\Big)+\frac{1}{N'}\Big[1-\Tr(\ham_\statement \rho_\hist)\Big]\ge1-\frac{\alpha}{N'}.
\end{align*}
We have therefore shown the $\left(1-\frac{\alpha}{N'}\right)$-completeness.

Next let us show the soundness.
The malicious prover first does any POVM measurement on $\rho_P$ to get $\msg_1=\{\com_j\}_{j\in [N]}$,
and sends it to the verifier.
After receiving $S$ from the verifier, the prover does another POVM measurement on the remaining state $\st$ to get $\msg_3$, and sends it to the verifier.
The verifier therefore measures all qubits of the $N$-qubit
state $(\prod_{j\in S}X_j^{x_j}Z_j^{z_j})\rho(\prod_{j\in S}X_j^{x_j}Z_j^{z_j})$, 
where $\rho$ is the state of the register $V$ after the prover does the first POVM measurement, and
$\{x_j,z_j\}_{j\in S}$ is that in $\msg_3$.
Note that $\rho$ is independent of $S$, because the first POVM measurement is done before $S$ is given to the prover.
Due to the binding of the commitment scheme, each $\com_j$ can be opened to a unique value $(\hat{x}_j,\hat{z}_j)$ or rejected by $\verify_{\mathsf{comm}}$.
We can assume that the prover always sends correct $\msg_3$ so that all $\{\com_j\}_{j\in [S]}$ are accepted by $\verify_{\mathsf{comm}}$, because otherwise the prover is rejected.
Therefore, it is equivalent that the verifier measures the energy of $\ham_\statement$ on the $N$-qubit state 
$\hat{\rho}:=(\prod_{j\in [N]}X_j^{\hat{x}_j}Z_j^{\hat{z}_j})\rho(\prod_{j\in [N]}X_j^{\hat{x}_j}Z_j^{\hat{z}_j})$.
Because $\{\hat{x}_j,\hat{z}_j\}_{j\in [N]}$ is fixed before $S$ is chosen, $\hat{\rho}$ is independent of $S$.
Then due to \cref{lem:prob_and_energy} and \cref{lem:five_local_Hamiltonian}, the acceptance probability
in $\Pi_\Sigma'$ when $\statement\in L_\no$ is at most
\begin{align*}
\Big(1-\frac{1}{N'}\Big)+\frac{1}{N'}\Big[1-\Tr(\ham_\statement \hat{\rho})\Big]+\negl(\secpar)\le
1-\frac{\beta}{N'}+\negl(\secpar).
\end{align*}

For any $\statement$, let 
$\event_\statement$ be the event that the statement output by $\A_1$ is $\statement$. 
Then,
\begin{align*}
    \Pr[\statement\in L_\no \land \text{verifier outputs $\top$}]
    \le\sum_{\statement\in L_\no}\Pr[\event_\statement]\left(1-\frac{\beta}{N'}+\negl(\secpar)\right)\leq 
    \left(1-\frac{\beta}{N'}+\negl(\secpar)\right).
\end{align*}

We have therefore shown the $\left(1-\frac{\beta}{N'}+\negl(\secpar)\right)$-soundness.
\end{proof}

\begin{lemma}\label{lem:sigma_ZK}
$\Pi_{\Sigma}$ satisfies special zero-knowledge property.
\end{lemma}
\begin{proof}
We construct the simulator $\siml_\Sigma$ as follows.
\begin{description}
\item[$\siml_\Sigma(\pkey,\statement,\msg_2)$:]
It parses $\msg_2=S$ and 
generates a quantum state $\rho_{S}:=\siml_{\hist}(\statement,S)$ using $\siml_{\hist}$ in \cref{lem:five_local_Hamiltonian}.  
Then it measures the corresponding qubits of $\rho_{S}$ and $\rho_P$ in the Bell basis.   
Let $\{x_j,z_j\}_{j\in S}$ be the measurement outcomes. 
It computes $\commit_{\mathsf{comm}}(1^\secpar,(x_j,z_j))\ra (\com_j,d_j)$ for each $j\in S$ and $\commit_{\mathsf{comm}}(1^\secpar,(0,0))\ra (\com_j,d_j)$ for each $j\in [N]\setminus S$. 
It outputs  $\msg_1:=\{\com_j\}_{j\in[N]}$ and $\msg_3=(\{d_j\}_{j\in S},\{x_j,z_j\}_{j\in S})$. 
\end{description}
In the following, we prove that the above simulator satisfies the requirement of the special zero-knowledge. 
For proving this, we consider the following sequence of modified versions of $\siml_\Sigma$,  which take $\witness\in R_L(\statement)$ as an additional input.
\begin{description}
\item[$\siml_\Sigma^{(1)}(\pkey,\statement,\witness,\msg_2)$:]
This simulator works similarly to $\siml_\Sigma$ except that it first generates the history state $\rho_{\hist}$ and then uses the corresponding part of $\rho_{\hist}$ instead of $\rho_S$. 
Note that this simulator can generate the history state since it takes $\witness$ as input.
\item[$\siml_\Sigma^{(2)}(\pkey,\statement,\witness,\msg_2)$:]
This simulator works similarly to $\siml_\Sigma^{(1)}$ except that it measures $j$-th qubits of $\rho_{\hist}$ and $\rho_P$ for all $j\in [N]$ (rather than only for $j\in S$) and gets the measurement outcomes $\{x_j,z_j\}_{j\in [N]}$.
Note that this simulator generates the commitments in the same way as $\siml_\Sigma^{(1)}$. 
\item[$\siml_\Sigma^{(3)}(\pkey,\statement,\witness,\msg_2)$:]
This simulator works similarly to $\siml_\Sigma^{(2)}$ except that it generates\\ 
$\commit_{\mathsf{comm}}(1^\secpar,(x_j,z_j))\ra (\com_j,d_j)$ for all $j\in [N]$.
\end{description}
Let $\A=(\A_1,\A_2)$ be a QPT adversary. 
For notational simplicity, we let $\siml_\Sigma^{(0)}:=\siml_\Sigma$,  
\begin{align*}
    p_0:=\Pr\left[
    \A_2(\st_\A,\statement,\msg_1,\msg_2,\msg_3)=1
    :
    \begin{array}{c}
         (\pkey,\st_\A) \sample \A_1(1^\secpar)\\
       (\msg_1,\msg_3) \sample \siml_\Sigma(\pkey,\statement,\msg_2)\\
    \end{array}
    \right],
\end{align*}
\begin{align*}
    p_i:=\Pr\left[
    \A_2(\st_\A,\statement,\msg_1,\msg_2,\msg_3)=1
    :
    \begin{array}{c}
         (\pkey,\st_\A) \sample \A_1(1^\secpar)\\
       (\msg_1,\msg_3) \sample \siml_\Sigma^{(i)}(\pkey,\statement,\witness,\msg_2)\\
    \end{array}
    \right]
\end{align*}
for $i=1,2,3$, and 
\begin{align*}
p_{\mathsf{real}}:=\Pr\left[
    \A_2(\st_\A,\statement,\msg_1,\msg_2,\msg_3)=1
    :
    \begin{array}{c}
         (\pkey,\st_\A) \sample \A_1(1^\secpar)\\
       (\msg_1,\st) \sample \prove_1(\pkey,\statement,\witness)\\
       \msg_3 \sample \prove_2(\st,\msg_2)\\ 
    \end{array}
    \right].
\end{align*}
What we have to prove is $|p_\mathsf{real}-p_0|=\negl(\secpar)$. 
We prove this by the following claims.
\begin{myclaim}
$|p_0-p_1|\leq \negl(\secpar)$.
\end{myclaim}
\begin{proof}
By \cref{lem:five_local_Hamiltonian}, $\|\rho_S-\Tr_{[N]\setminus S}\rho_\hist\|_{tr}=\negl(\secpar)$. The claim  immediately follows from this. 
\end{proof}
\begin{myclaim}
$p_1=p_2$.
\end{myclaim}
\begin{proof}
This immediately follows from the fact that
the measurement results corresponding to $j\in [N]\setminus S$ are not used, which is equivalent to
tracing out all qubits of $\rho_\hist$ in $[N]\setminus S$.
\end{proof}
\begin{myclaim}
$|p_2-p_3|\leq \negl(\secpar)$.
\end{myclaim}
\begin{proof}
This follows from a straightforward reduction to the computational hiding property of the commitment scheme.  
\end{proof}
\begin{myclaim}
$p_3=p_{\mathsf{real}}$.
\end{myclaim}
\begin{proof}
This claim clearly holds since $\siml_\Sigma^{(3)}$ generates $\msg_1$ and $\msg_3$ in exactly the same way as by the real proving algorithm.   
\end{proof}
By combining the above claims, we have $|p_{\mathsf{real}}-p_0|\leq \negl(\secpar)$. 
This completes the proof of \cref{lem:sigma_ZK}. 
\end{proof}

\begin{lemma}\label{lem:sigma_entropy}
$\Pi_{\Sigma}$ satisfies high min-entropy property.
\end{lemma}
\begin{proof}
We define $\Sigma.\prove_1^Q$ to be the part of $\Sigma.\prove_1$ that generates $\{x_j,z_j\}_{j\in [N]}$ by the Bell basis measurements and $\Sigma.\prove_1^C$ to be the rest of $\Sigma.\prove_1$. 
By the computational hiding property of the commitment,  a commitment does not take a fixed value with non-negligible probability. 
Then it is clear that $\msg_1=\{\com_j\}_{j\in[N]}$ does not take a fixed value with non-negligible probability. 
\end{proof}
\end{proof}

%% file: FS_completeness_soundness.tex
\begin{proof}[Proof of \cref{lem:QRO_completeness_soundness}]

The completeness is clear.
For proving soundness, we rely on the following lemma shown in \cite{C:DonFehMaj20}.
\begin{lemma}[{\cite[Theorem 2]{C:DonFehMaj20}}]\label{lem:measure_and_reprogram}
Let $X$ and $Y$ be non-empty sets and $\A$ be an
arbitrary oracle quantum algorithm that takes as input a quantum state $\rho$, makes $q$ queries to a uniformly random $H:X\rightarrow Y$, and outputs some $x\in X$ and a (possibly quantum) output $z$. 
There exist black-box quantum algorithms $\mathcal{S}_1^\A$ and $\mathcal{S}_2^\A$ such that for any quantum input $\rho$, $x^*\in X$, and any predicate $V$:
\begin{align*}
&\Pr_{H}\left[x=x^*\land V(x,H(x),z):(x,z)\la\A^{H}(\rho)\right]\\
&\leq 
(2q+1)^2\Pr_{y}\left[x=x^*\land V(x,y,z):
\begin{array}{c}
(x,\st)\la \mathcal{S}_1^\A(\rho)\\
z \la \mathcal{S}_2^\A(\st,y)
\end{array}
\right]
\end{align*}
Furthermore, $\mathcal{S}_1^\A$ and $\mathcal{S}_2^\A$ run in time polynomial in $q$, $\log|X|$, and $\log|Y|$. 
\end{lemma}
Based on the above lemma, we prove the soundness of $\Pi_\Sigma$ as follows:
\begin{align*}
&\Pr_{H,(\pkey,\vkey)}\left[\statement\in L_\no \land \verify_{\mathsf{QRO}}^H(\vkey,\statement,\pi)=\top:(\statement,\pi)\la\A^H(\pkey)\right]\\
&=\Pr_{H,(\pkey,\vkey)}\left[
\begin{array}{c}
\statement\in L_\no\\
\land\\
\Sigma.\verify_2(\vkey,\statement,\msg_1,H(\msg_1),\msg_3)=\top
\end{array}
:
(\statement,(\msg_1,\msg_3)) 
\la\A^H(\pkey)
\right]\\
&=\mathbb{E}_{(\pkey^*,\vkey^*)}\Pr_{H}\left[
\begin{array}{c}
\statement\in L_\no\\
\land\\
\Sigma.\verify_2(\vkey^*,\statement,\msg_1,H(\msg_1),\msg_3)=\top
\end{array}
:
(\statement,(\msg_1,\msg_3))
\la\A^H(\pkey^*)
\right]\\
&=\mathbb{E}_{(\pkey^*,\vkey^*)}\sum_{\statement^*\in L_\no, \msg_1^*}\\
&\Pr_{H}\left[
\begin{array}{c}
(\statement,\msg_1)=(\statement^*,\msg_1^*)\\
\land\\
\Sigma.\verify_2(\vkey^*,\statement,\msg_1,H(\msg_1),\msg_3)=\top
\end{array}
:
(\statement,(\msg_1,\msg_3)) 
\la\A^H(\pkey^*)
\right]\\
&\le (2q+1)^2\mathbb{E}_{(\pkey^*,\vkey^*)} 
\sum_{\statement^*\in L_\no,\msg_1^*}\\
& \Pr_{\msg_2}\left[
 \begin{array}{c}
  (\statement,\msg_1)=(\statement^*,\msg_1^*)\\
    \land\\
    \Sigma.\verify_2(\vkey^*,\statement,\msg_1,\msg_2,\msg_3)=\top
    \end{array}
    :
    \begin{array}{c}
       (\statement,\msg_1,\st) \sample \mathcal{S}_1^{\A}(\pkey^*)\\
       \msg_3 \sample \mathcal{S}_2^{\A}(\st,\msg_2)\\
    \end{array}
    \right] \\
&= (2q+1)^2
\Pr_{\msg_2,(\pkey,\vkey)}\left[
 \begin{array}{c}
  \statement\in L_\no\\
    \land\\
    \Sigma.\verify_2(\vkey,\statement,\msg_1,\msg_2,\msg_3)=\top
    \end{array}
    :
    \begin{array}{c}
       (\statement,\msg_1,\st) \sample \mathcal{S}_1^{\A}(\pkey)\\
       \msg_3 \sample \mathcal{S}_2^{\A}(\st,\msg_2)\\
    \end{array}
    \right] \\
&\le (2q+1)^2\negl(\secpar)\\
&=\negl(\secpar)
\end{align*}
where the first inequality is obtained by applying \cref{lem:measure_and_reprogram} for each fixed $(\pkey^*,\vkey^*)$ with 
$\rho:=\pkey^*$, $x:=(\statement,\msg_1)$, $y:=\msg_2$,  $z:=\msg_3$, and $V((\statement,\msg_1),\msg_2, \msg_3):=(\Sigma.\verify_2(\vkey^*,\statement,\msg_1,\msg_2,\allowbreak\msg_3)\overset{?}{=}\top)$ and the second inequality follows from the soundness of $\Pi_{\Sigma}$.

\end{proof}

%% file: FS_ZK.tex
\begin{proof}[Proof of \cref{lem:QRO_ZK}]
For proving the zero-knowledge property, we use the following lemma.

\begin{lemma}[Adaptive Reprogramming  \cite{GHHM20}]\label{lem:adaptive_reprogramming} 
Let $X_1,X_2,X',Y$ be some finite sets. 
For an algorithm $\A$, we consider the following experiment for $b\in \bit$:
\begin{description}
\item[$\mathsf{Exp}_b^{\A}$:]
The experiment first uniformly chooses a function $H:X_1\times X_2\rightarrow Y$, which may be updated during the execution of the experiment. 
$\A$ can make the following two types of queries:
\begin{description}
\item[Random Oracle Query:] 
When $\A$ queries $(x_1,x_2)\in X_1\times X_2$, the oracle returns $H(x)$. 
This oracle can be accessed quantumly (i.e., upon a query $\sum_{x_1,x_2,y}\ket{x_1,x_2}\ket{y}$, the oracle returns $\sum_{x_1,x_2,y}\ket{x_1,x_2}\ket{y\oplus H(x_1,x_2)}$).  
\item[Reprogramming Query:]
A reprogramming query should consist of $x_1\in X_1$ and 
a description of a probabilistic distribution $D$ over $X_2\times X'$. On input $(x_1,D)$, the oracle works as follows. First, the oracle takes $(x_2,x')\sample D$ and $y\sample Y$. 
Then it does either of the following depending of the value of $b$. 
\begin{enumerate}
    \item If $b=0$, it does nothing. 
\item If $b=1$, it reprograms $H$ so that $H(x_1,x_2)=y$. Note that  the reprogrammed $H$ is used for answering random queries hereafter. 
\end{enumerate}
Finally, it returns $(x_2,x')$ to $\A$.
Note that this algorithm is only classically accessed.
\end{description}
After making an arbitrary number of queries to the above oracles, $\A$ finally outputs a bit $b'$, which is treated as the output of the experiment.
\end{description}
Suppose that $\A$ makes at most $q_H$ random oracle queries and at most $q_R$ reprogramming queries and let $p_{max}:=\max_{D,x_2^*}\Pr[x_2=x_2^*:(x_2,x')\sample D]$ where the maximum is taken over all $D$ queried by $\A$ as part of a reprogramming query and $x_2^*\in X_2$. 
Then we have 
\begin{align*}
    \left|\Pr[\mathsf{Exp}_0^{\A}=1]-\Pr[\mathsf{Exp}_1^{\A}=1]\right|\leq \frac{3 q_R}{2}\sqrt{q_{H}p_{max}}
\end{align*}
\end{lemma}
\begin{remark}
The above lemma is a special case of \cite[Theorem 1]{GHHM20}.
We note that the roles of $X_1$ and $X_2$ are swapped from the original one, but this is just for convenience in later use and does not make any essential difference. We also note that a similar special case is stated in  \cite[Proposition 2]{GHHM20}, but the above lemma is slightly more general than that since their proposition assumes that $\A$ uses the same $D$ for all reprogramming queries.  
\end{remark}

\begin{proof}
Let $\siml_\Sigma$ be the simulator for $\Pi_\Sigma$. 
We construct a simulator $\siml_{\mathsf{QRO}}$ for $\Pi_{\mathsf{QRO}}$ as follows where $\mathcal{C}$ is the set from which $\msg_2$ is uniformly chosen.
\begin{description}
\item[$\siml_{\mathsf{QRO}}^{H,\reprogram}(\pkey,\statement)$:]
It randomly chooses $\msg_2\sample \mathcal{C}$, 
generates $(\msg_1,\msg_3)\sample \siml_\Sigma(\pkey,\statement,\msg_2)$, 
queries $((\statement,\msg_1),\msg_2)$ to $\reprogram$, which reprograms $H$ so that $H(\statement,\msg_1)=\msg_2$,  
and outputs $(\msg_1,\msg_3)$.   
\end{description}
In the following, we prove that the above simulator satisfies the requirement for adaptive multi-theorem zero-knowledge. 
For proving this, 
we consider the following sequence of modified versions of $\siml_{\mathsf{QRO}}$,  which take $k$ copies of a witness $\witness\in R_L(\statement)$ as an additional input.
\begin{description}
\item[${\siml_{\mathsf{QRO}}^{(1)}}^{H,\reprogram}(\pkey,\statement,\witness^{\otimes k})$:]
This simulator uses the real proving algorithm instead of the simulator to generate $\msg_1$ and $\msg_3$.
That is, it generates  
$(\msg_1,\st)\sample \Sigma.\prove_1(\pkey,\statement,\witness^{\otimes k})$, randomly chooses $\msg_2\sample \mathcal{C}$, 
generates 
$\msg_3\sample \Sigma.\prove_2(\st,\msg_2)$, 
queries $((\statement,\msg_1),\msg_2)$ to $\reprogram$, which reprograms $H$ so that $H(x,\msg_1)=\msg_2$,  
and outputs $(\msg_1,\msg_3)$.   
\item[${\siml_{\mathsf{QRO}}^{(2)}}^{H,\reprogram}(\pkey,\statement,\witness^{\otimes k})$:]
This simulator derives $\msg_2$ by querying to the random oracle instead of randomly choosing $\msg_2$ and then reprogramming the random oracle to be consistent. 
That is,  it generates  
$(\msg_1,\st)\sample \Sigma.\prove_1(\pkey,\statement,\witness^{\otimes k})$, sets $\msg_2:= H(\statement,\msg_1)$, 
generates 
$\msg_3\sample \Sigma.\prove_2(\st,\msg_2)$,  
and outputs $(\msg_1,\msg_3)$. 
Note that this simulator no longer makes a query to $\reprogram$. 
\end{description}
Let $\dist$ be a QPT distinguisher. 
For notational simplicity, let $\ora_{S(0)}:=\ora_{S}$, 
$\ora_{S(i)}$ be the oracle that works similarly to $\ora_{S}$ except that $\siml_{\mathsf{QRO}}^{(i)}$ is used instead of $\siml_{\mathsf{QRO}}$ for $i=1,2$, 
\begin{align*}
    p_i:= \Pr\left[\dist^{H,
    \ora_{S(i)}^{H,\reprogram}(\cdot,\cdot,\cdot)}(1^\secpar)=1:
    \begin{array}{c}
          H \sample \ROdist
    \end{array}
    \right]
\end{align*}
for $i=0,1,2$, and 
\begin{align*}
    p_{\mathsf{real}}:=\Pr\left[\dist^{H,
    \ora_{P}^{H}(\cdot,\cdot,\cdot)}(1^\secpar)=1:
    \begin{array}{c}
          H \sample \ROdist
    \end{array}
    \right].
\end{align*}
What we have to prove is $|p_{\mathsf{real}}-p_0|=\negl(\secpar)$. 
We prove this by the following claims.
\begin{myclaim}\label{cla:FS_simulation_zero_to_one}
$|p_0-p_1|\leq \poly(\secpar)$.
\end{myclaim}
\begin{proof}
This claim can be proven by a straightforward reduction to the special zero-knowledge property of $\Pi_{\Sigma}$ and a standard hybrid argument. 
\end{proof}
\begin{myclaim}\label{cla:FS_simulation_one_to_two}
$|p_1-p_2|\leq \poly(\secpar)$.
\end{myclaim}
\begin{proof}
This claim can be proven by a straightforward reduction to \cref{lem:adaptive_reprogramming} where $\statement$, $\msg_1$, $\st$, and $\msg_2$ play the roles of $x_1$, $x_2$, $x'$, and $y$, respectively, 
and the output distribution of $\Sigma.\prove_1^C(\st')$ where $\st'\sample \Sigma.\prove_1^Q(\pkey,\statement,\witness^{\otimes k})$ plays 
the role of the distribution $D$. (See \cref{def:sigma_q_prepro} for the definitions of $\Sigma.\prove_1^C$ and $\Sigma.\prove_1^Q$).
Since the number of $\dist$'s queries is $\poly(\secpar)$ and $\msg_1$ sampled by  $\Sigma.\prove_1^C(\st')$ does not take any fixed value with non-negligible probability as required by the high min-entropy property of $\Pi_\Sigma$, $p_{max}$ in \cref{lem:adaptive_reprogramming} is negligible.
Then \cref{lem:adaptive_reprogramming} 
directly gives the above claim. 
\end{proof}
\begin{myclaim}\label{cla:FS_simulation_two_to_real}
$p_2=p_{\mathsf{real}}$.
\end{myclaim}
\begin{proof}
This is clear since $\siml_{\mathsf{QRO}}^{(2)}$ works similarly to the real proving algorithm $\prove_{\mathsf{QRO}}$. 
\end{proof}
By combining the above claims, we obtain $|p_{\mathsf{real}}-p_0|\leq \negl(\secpar)$. This completes the proof of \cref{lem:QRO_ZK}. 
\end{proof}
\end{proof}

%% file: Broadbent-Grilo.tex
\ifnum\submission=1
\subsection{More Explanation on \cref{lem:five_local_Hamiltonian}}\label{sec:BG20} 
\else
\section{More Explanation on \cref{lem:five_local_Hamiltonian}}\label{sec:BG20} 
\fi
Here, we explain how to obtain  \cref{lem:five_local_Hamiltonian} based on \cite{FOCS:BroGri20}. 
Let $L=(L_\yes,L_\no)$ be any $\QMA$ promise problem,
and $V=U_T...U_1$ be its verification circuit, where each $U_i$ is an elementary gate taken from
a universal gate set.
For $\statement \in L_\yes$, there exists a witness state
$|\psi\rangle$ such that $V$ accepts with probability
exponentially close to 1, whereas
for $\statement\in L_\no$, any state makes $V$
accept with probability exponentially small.

As is explained in \cite{FOCS:BroGri20}, we consider the
encoded version of the verification circuit
$V'$ with a certain quantum error
correcting code. 
The circuit $V'$ consists of
gates from the universal gate set 
$\{CNOT,T,H,X,Z\}$.
From the standard circuit-to-Hamiltonian
construction technique,
we can construct a local Hamiltonian
$H_\statement:=\sum_i H_i$ corresponding to
$V'$.
If there is a witness state $|\psi\rangle$
that makes $V'$ accept with
probability $1-\negl(|\statement|)$, then the history state
\begin{eqnarray*}
\frac{1}{\sqrt{T+1}}\sum_{t\in[T+1]}|0^{T-t}1^t\rangle\otimes U_t...U_1(Enc(|\psi\rangle)\otimes|0^A\rangle)
\end{eqnarray*}
has exponentially small energy.
Due to the local simulatability,
there is an efficient deterministic algorithm
that outputs the classical description of a state that is close to the reduced
density matrix of the history state on at most five
qubits~\cite{FOCS:BroGri20,FOCS:GriSloYue19}.
If every quantum state $|\psi\rangle$ makes $V'$ reject with probability at least $\epsilon$,
then the groundenergy of $H$ is at least $\Omega(\frac{\epsilon}{T^3})$.

Let $\ham_\statement=\sum_{i=1}^M c_i P_i$ be the local Hamiltonian,
where $M=\poly(|\statement|)$, $c_i$ is real, and $P_i$ is a tensor product of
Pauli operators $(I,X,Y,Z)$.
In the standard circuit-to-Hamiltonian construction,
each $P_i$
is a tensor product of at most five non-trivial Pauli operators
$(X,Y,Z)$.
As is shown in \cite{MNS}, this Hamiltonian can be changed to the form of
$\sum_{i=1}^M p_i \frac{I+s_iP_i}{2}$ with
$M=\poly(|\statement|)$, $s_i\in\{+1,-1\}$, $p_i>0$, $\sum_{i=1}^Mp_i=1$, and $P_i$ is a tensor
product of Pauli operators $(I,X,Y,Z)$ with at most five non-trivial Pauli operators $(X,Y,Z)$.
In fact, define the normalized Hamiltonian
\begin{eqnarray*}
\ham_\statement':=\frac{1}{2}\Big(I+\frac{\ham_\statement}{\sum_{i=1}^M |c_i|}\Big)=\sum_{i=1}^M \frac{|c_i|}{\sum_{i=1}^M|c_i|} \frac{I+sign(c_i) P_i}{2},
\end{eqnarray*}
and we have only to take $p_i:=\frac{|c_i|}{\sum_{i=1}^M|c_i|}$
and $s_i:=sign(c_i)$.

%% file: Alternative_Proofs.tex

\section{More details for the proof of \cref{lem:NIZK_completeness_soundness}}\label{sec:alternative_proof_NIZK}
Here we give more details of the completeness and the soundness of the virtual protocol 2.
In the virtual protocol 2,
$i\in[M]$ is chosen after $S_V$ and $(W_1,...,W_N)$ are chosen, but we can assume that $i$
is chosen before $S_V$ and $(W_1,...,W_N)$ are chosen, because they are independent.
When $P_i$ is not consistent to $(S_V,\{W_j\}_{j\in S_V})$ or
the coin heads, the measurement result on $\rho_V'$ is not used.
The probability that such cases happen is
\begin{eqnarray*}
&&\sum_{i=1}^Mp_i \Big({\rm Pr}[\mbox{not consistent}|i]+{\rm Pr}[\mbox{consistent}|i](1-3^{|S_i|-5})\Big)\\
&=&
\sum_{i=1}^M p_i \Big(
\frac{3^N\sum_{j=1}^5{N\choose j}-3^{N-|S_i|}}{3^N\sum_{j=1}^5{N\choose j}}
+\frac{3^{N-|S_i|}}{3^N\sum_{j=1}^5{N\choose j}}
(1-3^{|S_i|-5})\Big)\\
&=&\sum_{i=1}^M p_i \Big(
1-\frac{1}{3^5\sum_{j=1}^5{N\choose j}}\Big)\\
&=&
1-\frac{1}{3^5\sum_{j=1}^5{N\choose j}}\\
&=&1-\frac{1}{N'}.
\end{eqnarray*}
The probability that it is consistent and the coin tails is therefore
$\frac{1}{N'}$.
In this case, the measurement result on $\rho_V'$ is used.
The probability that the measurement result satisfies $(-1)^{\bigoplus_{j\in S_i}m_j'}=-s_i$
is from Lemma~\ref{lem:prob_and_energy},
\begin{eqnarray*}
\sum_{i=1}^Mp_i \Tr\Big[\Big(I-\frac{I+s_iP_i}{2}\Big)\rho_V'\Big]
=
1-\Tr(\ham_\statement \rho_V').
\end{eqnarray*}
The total acceptance probability is therefore
\begin{eqnarray*}
1-\frac{1}{N'}
+
\frac{1}{N'}\Big[
1-\Tr(\ham_\statement \rho_V')\Big]
=1-\frac{\Tr(\ham_\statement \rho_V')}{N'}.
\end{eqnarray*}

%% file: Alternative_NIZK.tex
\section{Alternative Simpler Construction of CV-NIZK in the QSP Model.}\label{sec:alternative_NIZK}
Here, we give an alternative construction of a CV-NIZK in the QSP model, which is slightly simpler than the construction given in \cref{sec:CV-NIZK}.

Our construction of a CV-NIZK for a $\QMA$ promise problem $L$ is given in \cref{fig:CV-NIZK_prime}
where $\ham_\statement$, $N$, $M$, $p_{i}$, $s_{i}$, $P_i$, $\alpha$, $\beta$, and $\rho_\hist$ are as in \cref{lem:five_local_Hamiltonian} for $L$.

\begin{figure}[h]
\rule[1ex]{\textwidth}{0.5pt}
\begin{description}
\item[$\setup(1^\secpar)$:]
The setup algorithm chooses
$(W_1,...,W_N,m_1,...,m_N)\sample \{X,Y,Z\}^N\times\{0,1\}^N$ and a uniformly random subset $S_V\subseteq [N]$ such that $1\leq |S_V|\leq 5$, 
and outputs a proving key $\pkey:=\bigotimes_{j=1}^N(U(W_j)|m_j\rangle)$ 
and a verification key $\vkey:=(S_V,\{W_j,m_j\}_{j\in S_V})$.
\item[$\prove(\pkey,\statement,\witness)$:]
The proving algorithm generates the history state $\rho_\hist$ for $\ham_\statement$ from $\witness$
and measures $j$-th qubits of $\rho_\hist$ and $\pkey$ in the Bell basis for $j\in [N]$.    
Let $x:=x_1\concat x_2\concat...\concat x_N$, and $z:=z_1\concat z_2\concat...\concat z_N$ where $(x_j,z_j)$ denotes the outcome of $j$-th measurement.
It outputs a proof $\pi:=(x,z)$.

\item[$\verify(\vkey,\statement,\pi)$:]
The verification algorithm 
parses 
$(S_V,\{W_j,m_j\}_{j\in S_V})\la \vkey$ and 
$(x,z)\la \pi$,   
chooses $i\in [M]$ according to the probability distribution defined by $\{p_{i}\}_{i\in[M]}$ (i.e., chooses $i$ with probability $p_{i}$).
Let 
\begin{eqnarray*}
S_i:=\{j\in[N]~|~\mbox{$j$th Pauli operator of $P_i$ is not $I$}\}.
\end{eqnarray*}
We note that we have $1\leq |S_i|\leq 5$ by the $5$-locality of $\ham_\statement$. 
We say that $P_i$ is consistent to $(S_V,\{W_j\}_{j\in S_V})$  if and only if 
$S_i = S_V$ and
the $j$th Pauli operator of
$P_i$ is $W_j$ for all
$j\in S_i$.
If $P_i$ is not consistent to $(S_V,\{W_j\}_{j\in S_V})$,  it outputs $\top$.
If $P_i$ is consistent to $(S_V,\{W_j\}_{j\in S_V})$,  it 
flips a biased coin that heads with probability $1-3^{|S_i|-5}$.
If heads, it outputs $\top$.
If tails, it defines
\begin{eqnarray*}
m_{j}':= 
\left\{
\begin{array}{cc}
m_{j}\oplus x_{j}&(W_j=Z),\\
m_{j}\oplus z_{j}&(W_j=X),\\
m_{j}\oplus x_{j}\oplus z_j&(W_j=Y)
\end{array}
\right.
\end{eqnarray*}
for $j\in S_i$, and outputs $\top$ if  
$(-1)^{\bigoplus_{j\in S_i}m'_{j}}=-s_{i}$ and $\bot$ otherwise. 
\end{description} 
\rule[1ex]{\textwidth}{0.5pt}
\hspace{-10mm}
\caption{CV-NIZK in the QSP model $\Pi'_\NIZK$.}
\label{fig:CV-NIZK_prime}
\end{figure}

We have the following lemmas.
\begin{lemma}[Completeness and Soundness]\label{lem:NIZK_completeness_soundness_prime}
$\Pi'_\NIZK$ satisfies $(1-\frac{\alpha}{N'})$-completeness and  $(1-\frac{\beta}{N'})$-soundness where $N':=3^5\sum_{i=1}^{5}{N \choose i}$.  
\end{lemma}
\begin{lemma}[Zero-Knowledge]\label{lem:NIZK_zero-knowledge_prime}
$\Pi'_\NIZK$ satisfies the zero-knowledge property. 
\end{lemma}
They can be proven similarly to \cref{lem:NIZK_completeness_soundness,lem:NIZK_zero-knowledge}, respectively. 

%% file: Verification.tex
\section{CV-NIP in the QSP model}\label{sec:CV-NIP}
We call a CV-NIZK in the QSP model a CV-NIP (classically-verifiable non-interactive proof) in the QSP model if the
zero-knowledge is not satisfied.
Here we give a construction of an information-theoretically sound
CV-NIP for $\QMA$ in the QSP model. 
Specifically, we prove the following theorem.
\begin{theorem}\label{thm:CV-NIP}
There exists a CV-NIP for $\QMA$ in the QSP model (without any computational assumption).
\end{theorem}
We note that this theorem is subsumed by \cref{thm:CV-NIZK}. 
Nonetheless, we give a proof of the theorem because the CV-NIP given here is much simpler.

Its proof is based on the fact that the \emph{2-local $\{ZZ,XX\}$-local Hamiltonian problem} is $\QMA$-complete.
That is, we have the following lemma.
\begin{lemma}[$\QMA$-completeness of 2-local $\{ZZ,XX\}$-Hamiltonian problem \cite{SIAM:CubMon16}]\label{lem:two_local_Hamiltonian}
For any $\QMA$ promise problem $L=(L_\yes,L_\no)$, there is a classical polynomial-time computable deterministic function that maps $\statement\in \bit^*$ to an $N$-qubit Hamiltonian $\ham_\statement$ of the form  
\begin{eqnarray*}
\ham_\statement=\sum_{j_1<j_2}
\frac{p_{j_1,j_2}}{2}\Big(\frac{I+s_{j_1,j_2} X_{j_1} X_{j_2}}{2} 
+\frac{I+s_{j_1,j_2} Z_{j_1} Z_{j_2}}{2} 
\Big)
\end{eqnarray*}
where $N=\poly(|\statement|)$, $p_{j_1,j_2}>0$, $\sum_{j_1<j_2} p_{j_1,j_2}=1$, and $s_{j_1,j_2}\in\{+1,-1\}$, 
and satisfies the following:
There are $0<\alpha<\beta<1$ such that $\beta-\alpha = 1/\poly(|\statement|)$ and
\begin{itemize}
    \item if $\statement \in L_\yes$, then there exists an $N$-qubit state $\rho$ such that $\Tr(\rho \ham_\statement)\leq \alpha$, and
    \item if $\statement \in L_\no$, then for any $N$-qubit state $\rho$, we have $\Tr(\rho \ham_\statement)\geq \beta$.
\end{itemize}
Moreover, for any $\statement \in L_\yes$, 
we can convert any witness $\witness\in R_L(\statement)$ into a state $\rho_{\hist}$, called the history state, such that $\Tr(\rho_\hist \ham_\statement)\leq \alpha$ in quantum polynomial time. 
\end{lemma}

\begin{remark}
It might be possible to prove $\QMA$-completeness of 2-local $\{ZZ,XX\}$-Hamiltonian problem with local simulatability by combining the techniques of \cite{FOCS:BroGri20,FOCS:GriSloYue19} and \cite{SIAM:CubMon16}. 
However, this is not clear, and indeed, this is mentioned as an open problem in \cite{FOCS:BroGri20}.    
Therefore we consider the  $5$-local Hamiltonian problem whenever we need local simulatability. 
\end{remark}

Our construction of a CV-NIP for a $\QMA$ promise problem $L$ is given in \cref{fig:CV-NIP}
where $\ham_\statement$, $N$, $p_{j_1,j_2}$, $s_{j_1,j_2}$, $\alpha$, $\beta$, and $\rho_\hist$ are as in \cref{lem:two_local_Hamiltonian} for $L$.
We remark that the proving algorithm uses only one witness, and thus we have $k=1$ in \cref{def:CV-NIZK} for this protocol. 
Multiple copies of the witness are needed only when we do the gap amplification (\cref{lem:CV-NIZK_amplification}).
A similar remark applies to all protocols proposed in this paper.

\begin{figure}[t]
\rule[1ex]{\textwidth}{0.5pt}
\begin{description}
\item[$\setup(1^\secpar)$:]
The setup algorithm chooses
$(h,m_1,...,m_N)\sample \{0,1\}^{N+1}$,
and outputs a proving key $\pkey:=\bigotimes_{j=1}^N(H^h|m_j\rangle)$ 
and a verification key $\vkey:=(h,m_1,...,m_N)$.
\item[$\prove(\pkey,\statement,\witness)$:]
The proving algorithm generates the history state $\rho_\hist$ for $\ham_\statement$ from $\witness$ and measures $j$-th qubits of $\rho_\hist$ and $\pkey$ in the Bell basis for $j\in [N]$.    
Let $x:=x_1\concat x_2\concat...\concat x_N$, and $z:=z_1\concat z_2\concat...\concat z_N$ 
where $(x_j,z_j)\in\{0,1\}^2$ denotes the outcome of $j$-th measurement.
It outputs a proof $\pi:=(x,z)$.

\item[$\verify(\vkey,\statement,\pi)$:]
The verification algorithm 
parses 
$(h,m_1,...,m_N)\la \vkey$ and 
$(x,z)\la \pi$, 
chooses $(j_1,j_2)\in [N]^2$ according to the probability distribution defined by $\{p_{j_1,j_2}\}_{j_1<j_2}$ (i.e., chooses $(j_1,j_2)$ with probability $p_{j_1,j_2}$), defines
$m_{j_b}':= m_{j_b}\oplus (hz_{j_b}\oplus (1-h)x_{j_b})$ for $b\in \{1,2\}$, and outputs $\top$ if  $(-1)^{m'_{j_1}\oplus m'_{j_2}}=-s_{j_1,j_2}$ and $\bot$ otherwise. 
\end{description} 
\rule[1ex]{\textwidth}{0.5pt}
\hspace{-10mm}
\caption{CV-NIP $\Pi_\NIP$.}
\label{fig:CV-NIP}
\end{figure}

\begin{figure}[t]
\rule[1ex]{\textwidth}{0.5pt}
\begin{description}
\item[$\setup_{\virone}(1^\secpar)$:]
The setup algorithm generates $N$ Bell-pairs between registers $\regP$ and $\regV$ and lets $\pkey$ and $\vkey$ be quantum states in registers $\regP$ and $\regV$, respectively. 
Then it outputs $(\pkey,\vkey)$.
\item[$\prove_{\virone}(\pkey,\statement,\witness)$:]
This is the same  as $\prove(\pkey,\statement,\witness)$ in \cref{fig:CV-NIP}.
\item[$\verify_{\virone}(\vkey,\statement,\pi)$:]
The verification algorithm chooses $h\sample \bit$, and measures each qubit of $\vkey$ in basis $\{H^h\ket{0},H^h\ket{1}\}$, and lets $(m_1,...,m_N)\in\{0,1\}^N$ be the measurement outcomes. 
The rest of this algorithm is the same as $\verify(\vkey,\statement,\pi)$ given in \cref{fig:CV-NIP}.
\end{description} 
\rule[1ex]{\textwidth}{0.5pt}
\hspace{-10mm}
\caption{The virtual protocol 1 for $\Pi_{\NIP}$}
\label{fig:virtual1_CV-NIP}
\end{figure}

\begin{figure}[t]
\rule[1ex]{\textwidth}{0.5pt}
\begin{description}
\item[$\setup_{\virtwo}(1^\secpar)$:]
This is the same as $\setup_{\virone}(1^\secpar)$ in \cref{fig:virtual1_CV-NIP}.
\item[$\prove_{\virtwo}(\pkey,\statement,\witness)$:]
This is the same  as $\prove(\pkey,\statement,\witness)$ in \cref{fig:CV-NIP}.
\item[$\verify_{\virtwo}(\vkey,\statement,\pi)$:]
The verification algorithm 
parses 
$(x,z)\la \pi$, computes $\vkey':=X^xZ^z\vkey Z^zX^x$, 
chooses $h\sample \bit$, measures each qubit of $\vkey'$ in basis $\{H^h\ket{0},H^h\ket{1}\}$, and lets $(m'_1,...,m'_N)$ be the measurement outcomes. 
It chooses $(j_1,j_2)\in [N]^2$ according to the probability distribution defined by $\{p_{j_1,j_2}\}_{j_1<j_2}$ (i.e., chooses $(j_1,j_2)$ with probability $p_{j_1,j_2}$) and outputs $\top$ if  $(-1)^{m'_{j_1}\oplus m'_{j_2}}=-s_{j_1,j_2}$ and $\bot$ otherwise. 
\end{description} 
\rule[1ex]{\textwidth}{0.5pt}
\hspace{-10mm}
\caption{The virtual protocol 2 for $\Pi_{\NIP}$}
\label{fig:virtual2_CV-NIP}
\end{figure}
We prove the following lemma.
\begin{lemma}[Completeness and Soundness]\label{lem:NIP_completeness_soundness}
$\Pi_\NIP$ satisfies $(1-\alpha)$-completeness and 
 $(1-\beta)$-soundness.
\end{lemma}

Since $(1-\alpha)-(1-\beta)=\beta-\alpha\geq 1/\poly(\secpar)$, 
by combining \cref{lem:CV-NIZK_amplification} and \cref{lem:NIP_completeness_soundness}, \cref{thm:CV-NIP} follows.

In the following, we give a proof of \cref{lem:NIP_completeness_soundness}.
\begin{proof}[Proof of \cref{lem:NIP_completeness_soundness}]
We prove this lemma by considering virtual protocols that do not change completeness and soundness.
An alternative direct proof is given later. 
First, we consider the virtual protocol 1 described in \cref{fig:virtual1_CV-NIP}.
The difference  from the original protocol is that the setup algorithm generates $N$ Bell pairs and gives each halves to the prover and verifier, and the verifier obtains $(m_1,...,m_n)$ by measuring his halves in either standard or Hadamard basis. 

Because verifier's measurement and the prover's measurement
commute with each other, 
in the virtual protocol 1, 
verifier's acceptance probability does not change even if the verifier chooses $h$ and measures 
$k_V$ (i.e., the $\regV$ register of the $N$ Bell-pairs) in the corresponding basis to obtain outcomes $(m_1,...,m_N)$ before $k_P$ (i.e, the $\regP$ register of the $N$ Bell-pairs) is given to the prover.
Moreover, conditioned on the above measurement outcomes, the state in $\regP$ collapses to $\bigotimes_{j=1}^N(H^{h}|m_j\rangle)$. (See Lemma~\ref{lem:statecollapsing}.) 
Therefore, the virtual protocol 1 is exactly the same as the original protocol from the prover's view, and verifier's acceptance probability of the virtual protocol 1
is the same as that of the original protocol $\Pi_\NIP$ for any possibly malicious prover.

Next, we further modify the protocol to define the virtual protocol 2 described in \cref{fig:virtual2_CV-NIP}.
The difference from the virtual protocol 1 is that instead of setting $m_{j}':= m_{j}\oplus (hz_{j}+(1-h)x_{j})$, the verification algorithm applies a corresponding Pauli operator to $(x,z)$ on $\vkey$, and then measures it to obtain $m_{j}'$. 
Since $X$ and $Z$ before the measurement has the effect of flipping the measurement outcome for $Z$ and $X$ basis measurements, respectively, this does not change the distribution of $(m_1',...,m_N')$.
(See Lemma~\ref{lem:XZ_before_measurement}.)
Therefore, verifier's acceptance probability of the virtual protocol 2
is the same as that of the virtual protocol 1 for any possibly malicious prover.

Therefore, it suffices to prove $(1-\alpha)$-completeness and 
$(1-\beta)$-soundness for the virtual protocol $2$. 
When $\statement\in L_\yes$ and $\pi$ is honestly generated, then $k'_V$  is the history state $\rho_\hist$, which satisfies $\Tr(\rho_\hist \ham_\statement)\leq \alpha$, by the correctness of quantum teleportation (Lemma \ref{lem:teleportation}). 
Therefore, 
by \cref{lem:prob_and_energy} and \cref{lem:two_local_Hamiltonian}, 
verifier's acceptance probability is $1-\Tr(\rho_\hist \ham_\statement)\geq 1-\alpha$. 

Let $\A$ be an adaptive adversary against soundness of virtual protocol $2$.
That is, $\A$ is given $\pkey$ and outputs $(\statement,\pi)$.
We say that $\A$ wins if $\statement \in L_\no$ and $\verify(\vkey,\statement,\pi)=\top$.
For any $\statement$, let 
$\event_\statement$ be the event that the statement output by $\A$ is $\statement$, and
$k'_{V,\statement}$ be the state in $\regV$ right before the measurement by $\verify$ conditioned on $\event_\statement$.
Similarly to the analysis for the completeness, 
by \cref{lem:prob_and_energy} and \cref{lem:two_local_Hamiltonian}, we have 
\begin{align*}
    \Pr[\A\text{~wins}]=\sum_{\statement\in L_\no}\Pr[\event_\statement]\left(1-\Tr(k'_{V,\statement} \ham_\statement)\right)\leq 
    \sum_{\statement\in L_\no}\Pr[\event_\statement]\left(1-\beta\right)\leq 1-\beta.
\end{align*}

\end{proof}

\begin{proof}[Another proof of \cref{lem:NIP_completeness_soundness}]
We first show the soundness.
Let us define
$H^h:= \prod_{j=1}^N H_j^h$
and $|m\rangle :=\bigotimes_{j=1}^N|m_j\rangle$.
Let
$\{\Lambda_{x,z,\statement}\}_{x,z,\statement}$ be the POVM that the adversary $\A$ does
on $k_P$.
Then,
\begin{eqnarray*}
&&    \Pr\left[
    \statement \in L_\no \land \verify(\vkey,\statement,\pi)=\top 
    :(\pkey,\vkey)\sample \setup(1^\secpar),(\statement,\pi) \sample \A(\pkey)\right]\\
&=&\frac{1}{2}\sum_{h\in\{0,1\}}
\frac{1}{2^N}\sum_{m\in\{0,1\}^N}
\sum_{x,z}
\sum_{\statement\notin L}
\langle m|H^h \Lambda_{x,z,\statement} H^h|m\rangle
\sum_{j_1,j_2} p_{j_1,j_2}^\statement
\frac{1-s_{j_1,j_2}^\statement(-1)^{m_{j_1}'\oplus m_{j_2}'}}{2}\\
&=&\frac{1}{2}\sum_{h\in\{0,1\}}
\frac{1}{2^N}\sum_{m\in\{0,1\}^N}
\sum_{x,z}
\sum_{\statement\notin L}
\sum_{j_1,j_2} p_{j_1,j_2}^\statement
\langle m|H^h \Lambda_{x,z,\statement} H^h
H^h X^xZ^z H^h
\frac{I-s^\statement_{j_1,j_2}Z_{j_1} Z_{j_2}}{2}
H^h Z^zX^x H^h
|m\rangle\\
&=&\frac{1}{2}\sum_{h\in\{0,1\}}
\frac{1}{2^N}
\sum_{x,z}
\sum_{\statement\notin L}
\sum_{j_1,j_2} p_{j_1,j_2}^\statement
\mbox{Tr}\Big[
H^h \Lambda_{x,z,\statement} H^h
H^h X^xZ^z H^h
\frac{I-s_{j_1,j_2}^\statement Z_{j_1} Z_{j_2}}{2}
H^h Z^zX^x H^h\Big]\\
&=&
\frac{1}{2^N}
\sum_{x,z}
\sum_{\statement\notin L}
\mbox{Tr}\Big[
Z^zX^x
\Lambda_{x,z,\statement} X^xZ^z 
(I-\ham_\statement)
\Big]\\
&=&
\mbox{Tr}[\sigma (I- \ham_\statement) ]\\
&\le&
\mbox{Tr}\Big[\frac{\sigma}{\Tr\sigma} (I- \ham_\statement) \Big]\\
&=&
1-\mbox{Tr}\Big[\frac{\sigma}{\Tr\sigma} \ham_\statement \Big]\\
&\le&
1-\beta,
\end{eqnarray*}
where $\sigma:=\frac{1}{2^N}\sum_{x,z}\sum_{\statement\notin L}
Z^zX^x\Lambda_{x,z,\statement}X^xZ^z$. Note that $\frac{\sigma}{\Tr\sigma}$ is a quantum state for any
POVM $\{\Lambda_{x,z,\statement}\}_{x,z,\statement}$. 

Next we show the completeness. The POVM corresponding to $\prove$ is
$\{\Lambda_{x,z}=\frac{1}{2^N}Z^zX^x \rho_\hist X^xZ^z\}_{x,z}$. 
Note that this is a POVM, because $\Lambda_{x,z}\ge0$, and
\begin{eqnarray*}
\sum_{x,z}\Lambda_{x,z}&=&2^N\times\frac{1}{2^{2N}}\sum_{x,z}Z^zX^x \rho_\hist X^xZ^z
=2^N\frac{I}{2^N}=I.
\end{eqnarray*}
The reason why such $\{\Lambda_{x,z}\}_{x,z}$ is the POVM done by $\prove$ algorithm is as follows.
The $\prove$ algorithm first prepares $\rho_\hist\otimes H^h|m\rangle\langle m|H^h$,
and then measures $j$th qubit of the history state and the $j$th qubit of
$H^h|m\rangle$ in the Bell basis for all $j=1,2,...,N$.
Then,
\begin{eqnarray*}
&&\Big(\bigotimes_{j=1}^N\langle \phi_{x_j,z_j}|\Big) 
\Big(\rho_\hist\otimes
H^h|m\rangle\langle m|H^h\Big)
\Big(\bigotimes_{j=1}^N|\phi_{x_j,z_j}\rangle\Big) \\
&=&\mbox{Tr}
\Big[
\frac{1}{2^N}Z^zX^x \rho_\hist X^xZ^z
\times H^h|m\rangle\langle m|H^h
\Big].
\end{eqnarray*}
Hence
\begin{eqnarray*}
&&    \Pr\left[
    \verify(\vkey,\statement,\pi)=\top 
    :(\pkey,\vkey)\sample \setup(1^\secpar),\pi \sample \prove(\pkey,\statement,\witness^{\otimes k})\right]\\
&=&
\frac{1}{2^N}
\sum_{x,z}
\mbox{Tr}\Big[ Z^zX^x\Big(\frac{1}{2^N}Z^zX^x \rho_\hist X^xZ^z\Big)X^xZ^z
(I-\ham_\statement)\Big]\\
&=&
\mbox{Tr}\Big[ \rho_\hist (I-\ham_\statement)\Big]\\
&=&1-
\mbox{Tr}\Big[ \rho_\hist \ham_\statement\Big]\\
&\ge&1-\alpha.
\end{eqnarray*}
\end{proof}

\paragraph{Impossibility of classical setup.}
In our protocol,
the setup algorithm sends a quantum proving key to the prover.
Can it be classical?
It is easy to see that such a protocol can exist only for languages in $\mathbf{AM}$.\footnote{A similar observation is also made in \cite{C:PasShe05}.}
In fact, assume that we have a CV-NIP for $L$ in the SP model where the proving key is classical. Then, we can construct a
2-round interactive proof for $L$ where
the verifier  runs the setup by itself and sends the proving key to the prover, and then the prover replies as in the original protocol.
Since $\mathbf{IP}(2)=\mathbf{AM}$, the above implies $L\in \mathbf{AM}$. 
Since it is believed that $\BQP$ is not contained in $\mathbf{AM}$~\cite{STOC:RazTal19}, it is highly unlikely that there is a CV-NIP even for $\BQP$ in the SP model with classical setup.

%% file: OT.tex
\section{Construction of Dual-Mode $k$-out-of-$n$ Oblivious Transfer}\label{sec:OT}
In this section, we prove \cref{lem:OT}.  
That is, we give a construction of a dual-mode $k$-out-of-$n$ oblivious transfer defined in \cref{def:OT} based on the LWE assumption. 

\subsection{Building Block}
We introduce dual-mode encryption that is used as a building block for our construction.
We refer to \cite{C:PeiVaiWat08} for the intuition of this primitive.
\begin{definition}[Dual-Mode Encryption \cite{C:PeiVaiWat08,SCN:Quach20}\footnote{This definition is based on the definition in \cite{SCN:Quach20}, which has several minor differences from that in \cite{C:PeiVaiWat08}.}]
A dual-mode encryption scheme over the message space $\mathcal{M}$ consists of PPT algorithms $\Pi_{\DEnc}=(\setup,\keygen,\enc,\dec,\findmessy,\trapkeygen)$ with the following syntax.
\begin{description}
\item[$\setup(1^\secpar,\mode)$:]
The setup algorithm takes the security parameter $1^\secpar$ and a mode $\mode\in\{\messymode,\decmode\}$ as input, and outputs a common refernece string $\crs$ and a trapdoor $\td_\mode$.
\item[$\keygen(\crs,\sigma)$:]
The key generation algorithm takes the common reference string $\crs$ and a branch value $\sigma\in \bit$ as input, and outputs a public key $\pk$ and a secret key $\sk$.
\item[$\enc(\crs,\pk,b,\mu)$:]
The encryption algorithm takes the common reference string $\crs$, a public key $\pk$, a branch value $b\in \bit$, and a message $\mu\in\mathcal{M}$ as input, and outputs a ciphertext $\ct$.
\item[$\dec(\crs,\sk,\ct)$:]
The decryption algorithm takes the common reference string $\crs$, a secret key $\sk$, and a ciphertext $\ct$ as input, and outputs a message $\mu\in\mathcal{M}$
\item[$\findmessy(\crs,\td_\messymode,\pk)$:]
The messy branch finding algorithm takes the common reference string $\crs$, trapdoor $\td_\messymode$ in the messy mode, and a public key $\pk$ as input, and outputs a branch value $b\in \bit$.
\item[$\trapkeygen(\crs,\td_\decmode)$:]
The trapdoor key generation algorithm takes the common reference string $\crs$ and a trapdoor $\td_\decmode$ in the decryption mode as input, and outputs a public key $\pk_0$ and two secret keys $\sk_0$ and $\sk_1$ that correspond to branches $0$ and $1$, respectively.
\end{description}
We require $\Pi_\DEnc$ to satisfy the following properties.

\medskip
\noindent \underline{
\textbf{Correctness for Decryptable Branch}
}
For all $\mode\in \{\messymode,\decmode\}$, $\sigma\in \bit$, and $\mu\in\mathcal{M}$,
we have 
\begin{align*}
    \Pr\left[
    \dec(\crs,\sk_\sigma,\ct,\mu)=\mu
    :
    \begin{array}{ll}
    (\crs,\td_\mode)\sample \setup(1^\secpar,\mode)\\
    (\pk,\sk_\sigma)\sample \keygen(\crs,\sigma)\\
    \ct\sample \enc(\crs,\pk,\sigma,\mu)
    \end{array}
    \right]\geq 1-\negl(\secpar).
\end{align*}

\medskip
\noindent \underline{
\textbf{Statistical Security in the Messy Mode}
}
With overwhelming probability over $(\crs,\td_\messymode)\sample \setup(1^\secpar,\messymode)$, 
for all possibly malformed $\pk$, all messages $\mu_0,\mu_1\in \bit^{\ell}$, and all unbounded-time distinguisher $\dist$,  we have 
\begin{align*}
    &\left|\Pr\left[
    \dist(\ct)=1:
    \begin{array}{ll}
    b\sample \findmessy(\crs,\td_\messymode,\pk)\\
        \ct\sample \enc(\crs,\pk,b,\mu_0) 
    \end{array}
    \right]\right.\\
    &\left.-
        \Pr\left[
    \dist(\ct)=1:
    \begin{array}{ll}
    b\sample \findmessy(\crs,\td_\messymode,\pk)\\
        \ct\sample \enc(\crs,\pk,b,\mu_1) 
    \end{array}
    \right]\right|
    \leq \negl(\secpar).
\end{align*}

\medskip
\noindent \underline{
\textbf{Statistical Security in the Decryption Mode}
}
With overwhelming probability over $(\crs,\td_\decmode)\sample \setup(1^\secpar,\decmode)$, 
for all $\sigma\in \bit$ and all unbounded-time distinguisher $\dist$,  we have 
\begin{align*}
    &\left|\Pr\left[
    \dist(\pk,\sk_\sigma)=1:
    \begin{array}{ll}
    (\pk,\sk_\sigma)\sample \keygen(\crs,\sigma)
    \end{array}
    \right]\right.\\
    &\left.-
        \Pr\left[
    \dist(\pk,\sk_\sigma)=1:
    \begin{array}{ll}
     (\pk,\sk_0,\sk_1)\sample \trapkeygen(\crs,\td_\decmode)
    \end{array}
    \right]\right|
    \leq \negl(\secpar).
\end{align*}

\medskip
\noindent \underline{
\textbf{Computational Mode Indistinguishability}
}
For any non-uniform QPT distinguisher $\dist$, we have
\begin{align*}
    &\left|\Pr\left[\dist(\crs)=1:(\crs,\td_\messymode)\sample \crsgen(1^\secpar,\messymode)\right]\right.\\
    &\left.- \Pr\left[\dist(\crs)=1:(\crs,\td_\decmode)\sample \crsgen(1^\secpar,\decmode)\right]
    \right|\leq \negl(\secpar).
\end{align*}
\end{definition}

Quach \cite{SCN:Quach20} gave a construction of a dual-mode encryption scheme based on the LWE assumption.
\begin{lemma}[\cite{SCN:Quach20}]\label{lem:dual_mode_encryption}
If the LWE assumption holds, then there exists a dual-mode encryption scheme. 
\end{lemma}
\begin{remark}
Peikert, Vaikuntanathan, and Waters \cite{C:PeiVaiWat08} gave a construction of a relaxed variant of dual-mode encrytption scheme based on the LWE assumption.
Their construction is more efficient than that of Quach \cite{SCN:Quach20} since they only rely on LWE with polynomial size modulus whereas Quach's construction relies on  LWE with super-polynomial modulus. 
However, their scheme does not suffice for our purpose due to the following two reasons.
\begin{enumerate}
    \item The security in the decryption mode holds only against computationally bounded adversaries.
    \item $\crs$ can be reused only for bounded number of times. 
\end{enumerate}
\end{remark}

\subsection{$1$-out-of-$n$ Oblivious Transfer}
In this section, we construct a dual-mode $1$-out-of-$n$ oblivious transfer based on dual-mode encryption.
That is, we prove the following lemma.
\begin{lemma}\label{lem:one_n_OT}
If there exists a dual-mode encryption scheme, then there exists a dual-mode $1$-out-of-$n$ oblivious transfer.
\end{lemma}

Let $\Pi_{\DEnc}=(\setup,\keygen,\enc,\dec,\findmessy,\trapkeygen)$ be a dual-mode encryption scheme over the message space $\mathcal{M}=\bit^{\ell}$. 
Then our construction of  a dual-mode $1$-out-of-$n$ oblivious transfer 
$\OT_{\onen}=(\crsgen_{\onen}, \receiver_{\onen},\allowbreak  \sender_{\onen},\derive_{\onen})$ over the message space $\mathcal{M}$ is given in \cref{fig:one-n-OT}.
This can be seen as a protocol obtained by applying the conversion of \cite{FOCS:BraCreRob86} to the dual-mode $1$-out-of-$2$ oblivious transfer of \cite{SCN:Quach20}. 

\begin{figure}[t]
\rule[1ex]{\textwidth}{0.5pt}
\begin{description}
\item[$\crsgen_{\onen}(1^\secpar,\mode)$:]
Let $\mode':=\decmode$ if $\mode=\binding$ and  $\mode':=\messymode$ if $\mode=\hiding$.
Then it  generates $(\crs,\td_{\mode'})\sample \setup(1^\secpar,\mode')$ and outputs $\crs$. 
\item[$\receiver_{\onen}(\crs,j)$:]
It generates $(\pk_i,\sk_{i,\sigma_i})\sample \keygen(\crs,\sigma_i)$ for all $i\in [N]$ where $\sigma_j:=1$ and $\sigma_i:=0$ for all $i\in [n]\setminus \{j\}$. 
It outputs $\ot_1:=\{\pk_i\}_{i\in [n]}$ and $\otst:=\left(j,\{\sigma_i,\sk_{i,\sigma_i}\}_{i\in [n]}\right)$.
\item[$\sender_{\onen}(\crs,\ot_1,\vecmu)$:]
It parses $\{\pk_i\}_{i\in [n]}\la \ot_1$ and $(\mu_1,...,\mu_n)\la \vecmu$, 
generates 
$(r_1,...,r_{N-1})\sample \bit^{\ell\times (N-1)}$,
sets $\mu'_{i,0}:=\mu_i \oplus r_{i-1}$ and 
$\mu'_{i,1}:=r_{i}\oplus r_{i-1}$ for all $i\in [n]$ where $r_0$ is defined to be $0^{\ell}$.
Then it generates
$\ct_{i,b}\sample\enc(\pk_i,b,\mu'_{i,b})$ for all $i\in[n]$ and $b\in\bit$, 
and outputs $\ot_2:=\{\ct_{i,b}\}_{i\in[n],b\in \bit}$. 

\item[$\derive_{\onen}(\otst,\ot_2)$:]
It parses $\left(j,\{\sigma_i,\sk_{i,\sigma_i}\}_{i\in [n]}\right)\la \otst$ and  $\{\ct_{i,b}\}_{i\in[n],b\in \bit}\la \ot_2$, 
computes $\mu'_{i,\sigma_i}\sample \dec(\sk_{i,\sigma_i},\ct_{i,\sigma_i})$ for all $i\in [j]$ and outputs $\mu_j:=\bigoplus_{i=1}^{j}\mu'_{i,\sigma_i}$.
\end{description} 
\rule[1ex]{\textwidth}{0.5pt}
\hspace{-10mm}
\caption{Our $1$-out-of-$n$ oblivious transfer $\Pi_{\onen}$}
\label{fig:one-n-OT}
\end{figure}

Then we prove the following lemmas.
\begin{lemma}\label{lem:onen_correcness}
$\Pi_{\onen}$ satisfies  correctness
\end{lemma}
\begin{proof}
This easily follows from correctnes of $\Pi_{\DEnc}$. 
\end{proof}
\begin{lemma}\label{lem:onen_mode_ind}
$\Pi_{\onen}$ satisfies  the computational mode indistinguishability.
\end{lemma}
\begin{proof}
This can be reduced to the computational mode indistinguishability of $\Pi_{\DEnc}$ in a straightforward manner.
\end{proof}

\begin{lemma}\label{lem:onen_receiver_security}
$\Pi_{\onen}$ satisfies  statistical receiver's security in the binding mode.
\end{lemma}
\begin{proof}
We construct $\siml_\rec$  as follows.
    \begin{description}
\item[$\siml_\rec(\crs)$:]
It generates $(\pk_i,\sk_{i,0})\sample \keygen(\crs,0)$ for all $i\in[n]$, and outputs $\ot_1:=\{\pk_i\}_{i\in[n]}$.
\end{description}
By statistical security in the decryption mode of $\Pi_{\DEnc}$, 
with overwhelming probability over $(\crs,\td_\decmode)\sample \setup(1^\secpar,\decmode)$, 
the distribution of $\pk$ generated as $(\pk,\sk_\sigma)\sample \keygen(\crs,\sigma)$ for any fixed $\sigma\in \bit$ is statistically close to that generated as $(\pk,\sk_0,\sk_1)\sample \trapkeygen(\crs,\td_\decmode)$, which does not depend on $\sigma$.
Therefore, the distributions of $\pk_i$ generated by $\siml_\rec(\crs)$ and $\receiver(\crs,j)$ are statistically close for any $j\in[n]$. 
Then statistical receiver's security in the binding mode of $\Pi_{\onen}$ follows by a standard hybrid argument. 
\end{proof}

\begin{lemma}\label{lem:onen_sender_security}
$\Pi_{\onen}$ satisfies  the statistical sender's security  in the hiding mode.
\end{lemma}
\begin{proof}
We construct 
$\siml_\CRS$, 
$\Open_\rec$, and $\siml_\sen$ as follows.
\begin{description}
\item[$\siml_\CRS(1^\secpar)$:]
It generates $(\crs,\td_\messymode)\sample \setup(1^\secpar,\messymode)$ and outputs $\crs$ and $\td:=\td_\messymode$.
\item[$\Open_\rec(\td,\ot_1)$:]
It parses 
$\td_\messymode \la \td$ and $\{\pk_i\}_{i\in[n]}\la \ot_1$, 
computes $\sigma_i\sample \findmessy(\td_\messymode,\pk_i)$ for all $i\in[n]$, 
and outputs the minimal $j\in[n]$ such that 
$\sigma_j=1$.
\item[$\siml_\sen(\crs,\ot_1,j,\mu_j)$:]
It 
generates $\mu_i\sample \mathcal{M}$ for $i\in [n]\setminus \{j\}$, and outputs $\ot_2\sample \sender_{\onen}(\crs,\ot_1,(\mu_1,...,\mu_n))$.
\end{description}
The first item of statistical sender's security in the hiding mode is clear because $\siml_\CRS(1^\secpar)$ generates $\crs$ in exactly the same manner as $\crsgen(1^\secpar,\hiding)$.
In the following, we prove the second item is also satisfied. 
For any 
unbounded-time adversary $\A=(\A_0,\A_1)$ and 
fixed $\vecmu=(\mu_1,...,\mu_n)$, we consider the following sequence of games between $\A$ and the challenger.
We denote by $\event_i$ the event that $\A_1$ returns $1$ in $\game_i$.
\begin{description}
\item[$\game_1$:]
This game works as follows.
\begin{enumerate}
    \item  The challenger generates $(\crs,\td_\messymode)\sample \setup(1^\secpar,\messymode)$ and sets $\td:=\td_\messymode$.
    \item $\A_0$ takes $(\crs,\td)$ as input and
    outputs $\ot_1=\{\pk_i\}_{i\in[n]}$ and $\st_\A$.
    \item The challenger computes $j:=\Open_\rec(\td,\ot_1)$. That is, it computes $\sigma_i\sample \findmessy(\td_\messymode,\pk_i)$ for all $i\in[n]$ and 
let $j$ be the minimal value such that 
$\sigma_j=1$. 
    \item The challenger 
    sets $\widetilde{\mu}_j:=\mu_j$, 
    generates $\widetilde{\mu}_i\sample \mathcal{M}$ for $i\in [n]\setminus \{j\}$ and $(r_1,...,r_{N-1})\sample \bit^{\ell\times (N-1)}$, and
    sets $\mu'_{i,0}:=\widetilde{\mu}_i \oplus r_{i-1}$ and 
$\mu'_{i,1}:=r_{i}\oplus r_{i-1}$ for all $i\in [n]$ where $r_0$ is defined to be $0^{\ell}$.
Then it generates
$\ct_{i,b}:=\enc(\pk_i,b,\mu'_{i,b})$ for all $i\in[n]$ and $b\in\bit$ and sets $\ot_2:=\{\ct_{i,b}\}_{i\in[n],b\in \bit}$.
\item $\A_1$ takes $\st_\A$ and $\ot_2$ as input and outputs a bit $\beta$. 
\end{enumerate}
\item[$\game_2$:]
This game is identical to the previous game except that $\mu'_{i,\sigma_i}$ is replaced with $0^{\ell}$ for all $i\in[n]$. 

By the statistical security in the messy mode of $\Pi_{\DEnc}$, it is easy to see that we have $|\Pr[\event_2]-\Pr[\event_1]|\leq \negl(\secpar)$. 

\item[$\game_3$:]
This game is identical to the previous game except that $\mu'_{i,0}$ is replaced with an independently and uniformly random element of $\mathcal{M}$ for all $i> j$.
We note that this game does not use $\{\widetilde{\mu}_i\}_{i\neq j}$ at all.

By an easy information theoretical argument, we can see that the distribution of $\{\mu'_{i,b}\}_{i\in[n],b\in\bit}$ does not change from the previous game, and thus we have  $\Pr[\event_3]=\Pr[\event_2]$. 

\item[$\game_4$:]
This game is identical to the $\game_1$ except that the challenger uses $\vecmu$ instead of $\widetilde{\vecmu}$.

By considering similar game hops to those from $\game_1$ to $\game_3$ in the reversed order, by the statistical security in the messy mode of $\Pi_{\DEnc}$,  we have $|\Pr[\event_4]-\Pr[\event_3]|\leq \negl(\secpar)$. 
\end{description}
Combining the above, we have $|\Pr[\event_4]-\Pr[\event_1]|\leq \negl(\secpar)$.
This is exactly the second item of statistical sender's security in the hiding mode. 
\end{proof}

By combining
\cref{lem:onen_correcness,lem:onen_mode_ind,lem:onen_receiver_security,lem:onen_sender_security}, we obtain \cref{lem:one_n_OT}.

\subsection{$k$-out-of-$n$ Oblivious Transfer}
In this section, we construct a dual-mode $k$-out-of-$n$ oblivious transfer based on dual-mode $1$-out-of-$n$ oblivious transfer by $k$ parallel repetitions.
That is, we prove the following lemma.
\begin{lemma}\label{lem:k_n_OT}
If there exists a dual-mode $1$-out-of-$n$ oblivious transfer, then there exists a dual-mode $k$-out-of-$n$ oblivious transfer.
\end{lemma}

By combining \cref{lem:dual_mode_encryption,lem:one_n_OT,lem:k_n_OT}, we obtain \cref{lem:OT}.

What is left is to prove \cref{lem:k_n_OT}.
Let $\Pi_{\onen}=(\crsgen_{\onen}, \receiver_{\onen}, \sender_{\onen},\allowbreak \derive_{\onen})$ be a dual-mode $1$-out-of-$n$ oblivious transfer over the message space $\mathcal{M}$.
Then our dual-mode $k$-out-of-$n$ oblivious transfer 
$\Pi_{\kn}=(\crsgen_{\kn}, \receiver_{\kn},\allowbreak  \sender_{\kn},\derive_{\kn})$ is described in \cref{fig:k-n-OT}.

\begin{figure}[t]
\rule[1ex]{\textwidth}{0.5pt}
\begin{description}
\item[$\crsgen_{\kn}(1^\secpar,\mode)$:]
It generates $\crs\sample \crsgen_{\onen}(1^\secpar,\mode)$ and outputs $\crs$.
\item[$\receiver_{\kn}(\crs,J)$:]
It parses 
$(j_1,...,j_k)\la J$, 
generates $(\ot_{1,i},\otst_i)\sample \receiver_{\onen}(\crs,j_i)$ for all $i\in[k]$, 
and outputs 
$\ot_1:=\{\ot_{1,i}\}_{i\in [k]}$ and $\otst:=\{\otst_{i}\}_{i\in [k]}$. 
\item[$\sender_{\kn}(\crs,\ot_1,\vecmu)$:]
It parses 
$\{\ot_{1,i}\}_{i\in [k]}\la \ot_1$, 
generates 
$\ot_{2,i}\sample \sender_{\onen}(\crs,\ot_{1,i},\vecmu)$
for all $i\in [k]$, and 
outputs $\ot_2:=\{\ot_{2,i}\}_{i\in[k]}$. 

\item[$\derive_{\kn}(\crs,\otst,\ot_2)$:]
It parses 
$\{\otst_{i}\}_{i\in [k]}\la \otst$, 
computes $\mu_{j_i}\sample \derive_{\onen}(\crs,\otst_i,\ot_{2,i})$ for $i\in[k]$, 
and outputs $(\mu_{j_1},...,\mu_{j_k})$.
\end{description} 
\rule[1ex]{\textwidth}{0.5pt}
\hspace{-10mm}
\caption{Our $k$-out-of-$n$ oblivious transfer $\Pi_{\kn}$}
\label{fig:k-n-OT}
\end{figure}

Then we prove the following lemmas.
\begin{lemma}\label{lem:kn_correcness}
$\Pi_{\kn}$ satisfies  correctness.
\end{lemma}
\begin{proof}
This can be reduced to correctness of $\Pi_{\onen}$ in a straightforward manner.
\end{proof}
\begin{lemma}\label{lem:kn_mode_ind}
$\Pi_{\kn}$ satisfies  the computational mode indistinguishability.
\end{lemma}
\begin{proof}
This can be reduced to the computational mode indistinguishability of $\Pi_{\onen}$ in a straightforward manner.
\end{proof}

\begin{lemma}\label{lem:kn_receiver_security}
$\Pi_{\kn}$ satisfies  statistical receiver's security in the binding mode.
\end{lemma}
\begin{proof}
Let $\siml_{\rec,\onen}$  be the corresponding algorithm for statistical receiver's security in the binding mode of $\Pi_{\onen}$.
Then  We construct $\siml_{\rec,\kn}$ for $\Pi_{\kn}$ as follows.
\begin{description}
\item[$\siml_{\rec,\kn}(\crs)$:]
It parses $(\crs,\pk)\la\crs$, 
computes $\ot_{1,i}\sample \siml_{\rec,\onen}(\crs)$ for all $i\in[k]$,
and outputs $\ot_1:=\{\ot_{1,i}\}_{i\in[k]}$.
\end{description}
Statistical receiver's security in the binding mode of $\Pi_{\kn}$ follows from that of $\Pi_{\onen}$ by a straightforward hybrid argument. 
\end{proof}

\begin{lemma}\label{lem:kn_sender_security}
Let 
$\Pi_{\kn}$ satisfies  the statistical sender's security  in the hiding mode.
\end{lemma}
\begin{proof}
Let $\siml_{\CRS,\onen}$, $\Open_{\rec,\onen}$, and $\siml_{\sen,\onen}$ be the corresponding algorithms for statistical sender's security in the hiding mode of $\Pi_{\onen}$.
Then  We construct $\siml_{\CRS,\kn}$, $\Open_{\rec,\kn}$, and $\siml_{\sen,\kn}$  
for $\Pi_{\kn}$ as follows.
\begin{description}
\item[$\siml_{\CRS,\kn}(1^\secpar)$:] 
This is exactly the same as $\siml_{\CRS,\onen}(1^\secpar)$.
\item[$\Open_\rec(\td,\ot_1)$:]
It parses 
$\{\ot_{1,i}\}_{i\in[k]}\la \ot_1$, 
computes $j_i:= \Open(\td,\ot_{1,i})$ for all $i\in[k]$, and outputs $J=(j_1,...,j_k)$. 
\item[$\siml_\sen(\crs,\ot_1,J,\vecmu_J)$:]
It parses $(j_1,...,j_k)\la J$ and $(\mu_{j_1},...,\mu_{j_k})\la \vecmu_J$, 
generates $\mu_i\sample \mathcal{M}$ for $i\in [n]\setminus \{j_1,...,j_k\}$, and outputs $\ot_2\sample \sender_{\kn}(\crs,\ot_1,(\mu_1,...,\mu_n))$.
\end{description}
The first item of statistical sender's security in the hiding mode of $\Pi_{\kn}$ immediately follows from that of $\Pi_{\onen}$.
In the following, we prove the second item. 
For any 
unbounded-time adversary $\A=(\A_0,\A_1)$ and 
fixed $\mu=(\mu_1,...,\mu_n)$, we consider the following sequence of games between $\A$ and the challenger.
We denote by $\event_i$ the event that $\A_1$ returns $1$ in $\game_i$.
\begin{description}
\item[$\game_1$:]
This game works as follows.
\begin{enumerate}
    \item  The challenger generates $(\crs,\td)\sample \siml_{\CRS,\kn}(1^\secpar)$.
    \item $\A_0$ takes $(\crs,\td)$ as input and
    outputs $\ot_1=\{\ot_{1,i}\}_{i\in[k]}$ and $\st_\A$.
    \item The challenger computes $J:=\Open_{\rec,\kn}(\td,\ot_1)$. That is, it computes $j_i:=\Open_{\rec,\onen}(\td,\ot_{1,i})$ for all $i\in[k]$ and 
lets $J:=(j_1,...,j_k)$. 
    \item 
    The challenger 
    generates 
    $\ot_{2,i}\sample \siml_{\sen,\onen}(\crs,\ot_{1,i},j_i,\mu_{j_i})$ for $i\in [k]$ and
    sets $\ot_2:=\{\ot_{2,i}\}_{i\in[k]}$.
\item $\A_1$ takes $\st_\A$ and $\ot_2$ as input and outputs a bit $\beta$. 
\end{enumerate}
\item[$\game_2$:]
This game is identical to the previous game except that $\ot_{2,i}$ is generated as 
$\ot_{2,i}\sample \sender(\crs,\ot_{1,i},\vecmu)$ for $i\in [k]$.

By the second item of statistical sender's security in the hiding mode of $\Pi_{\onen}$, we have $|\Pr[\event_2]-\Pr[\event_1]|\leq \negl(\secpar)$ by a standard hybrid argument. 
This is exactly the second item of statistical sender's security in the hiding mode. 




\end{description}
\end{proof}

By combining
\cref{lem:kn_correcness,lem:kn_mode_ind,lem:kn_receiver_security,lem:kn_sender_security}, we obtain \cref{lem:k_n_OT}.

%% file: Appendix_bell.tex
\section{QRO + Shared Bell pair model}\label{app:bell}

\begin{definition}[CV-NIZK in the QRO + Shared Bell pair Model]\label{def:bell_nizk}
A CV-NIZK for a $\QMA$ promise problem $L=(L_\yes,L_\no)$ in the QRO + shared Bell pair model 
w.r.t. a random oracle distribution $\ROdist$ 
consists of algorithms $\Pi=(\setup,\allowbreak\prove,\verify)$ with the following syntax:
\begin{description}
\item[$\setup(1^\secpar)$:]
This algorithm generates $\poly(\secpar)$ Bell pairs (a state $\frac{1}{\sqrt{2}}\left(\ket{0}\ket{0}+\ket{1}\ket{1}\right)$) and sends the first and second halves to the prover and verifier as proving key $\pkey$ and verification key $\vkey$, respectively.  
\item[$\prove^{H}(\pkey,\statement,\witness^{\otimes k})$:] This is a QPT algorithm that 
is given quantum oracle access to the random oracle $H$. It
takes  
the proving key $\pkey$, a statement $\statement$, and $k=\poly(\secpar)$ copies $\witness^{\otimes k}$ of a witness $\witness\in R_L(\statement)$ as input, and outputs a classical proof $\pi$.
\item[$\verify^{H}(\vkey,\statement,\pi)$:]
This is a QPT algorithm that 
is given quantum oracle access to the random oracle $H$.
It takes 
the verification key $\vkey$, a statement $\statement$, and a proof $\pi$ as input, and outputs $\top$ indicating acceptance or $\bot$ indicating rejection. 
\end{description}
We require $\Pi$ to satisfy the following properties. 

\medskip
\noindent \underline{
\textbf{Completeness.}
}
For all
$\statement\in L_\yes\cap \bit^\secpar$, and $\witness\in R_L(\statement)$, we have 
\begin{align*}
    \Pr\left[
    \verify^H(\vkey,\statement,\pi)=\top 
    :
    \begin{array}{c}
          H\sample \ROdist \\
         (\pkey,\vkey) \sample \setup(1^\secpar)\\
         \pi \sample \prove^H(\pkey,\statement,\witness^{\otimes k})
    \end{array}
    \right]
    \geq 1-\negl(\secpar).
\end{align*}

\medskip
\noindent \underline{
\textbf{Adaptive Statistical Soundness.}
}
For all adversaries $\A$ that make at most $\poly(\secpar)$ quantum random oracle queries, we have 
\begin{align*}
    \Pr\left[
    \statement\in L_\no \land \verify^H(\vkey,\statement,\pi)=\top
    :
    \begin{array}{c}
          H\sample \ROdist \\
         (\pkey,\vkey) \sample \setup(1^\secpar)\\
       (\statement,\pi) \sample \A^H(\pkey)
    \end{array}
    \right]
    \leq \negl(\secpar).
\end{align*}

\medskip
\noindent \underline{
\textbf{Adaptive Multi-Theorem Zero-Knowledge.}
}
For defining the zero-knowledge property in the QROM, we define the syntax of a simulator in the QROM following \cite{EC:Unruh15}. A simulator is given quantum access to the random oracle $H$ and classical access to reprogramming oracle $\reprogram$.
When the simulator queries $(x,y)$ to $\reprogram$, the random oracle $H$ is reprogrammed so that $H(x):=y$ while keeping the values on other inputs unchanged. 
Then the adaptive multi-theorem zero-knowledge property is defined as follows:

There exists  a QPT simulator $\siml$ with the above syntax such that for any QPT distinguisher $\dist$, we have 
\begin{align*}
    &\left|\Pr\left[\dist^{H,\ora_P^H(\cdot,\cdot)}(1^\secpar)=1:
    \begin{array}{c}
          H \sample \ROdist
    \end{array}
    \right]\right.\\     &-\left.\Pr\left[\dist^{H,
    \ora_S^{H,\reprogram}(\cdot,\cdot)}(1^\secpar)=1:
    \begin{array}{c}
          H \sample \ROdist
    \end{array}
    \right]\right|\leq \negl(\secpar)
\end{align*}
where  
$\dist$'s queries to the second oracle should be of the form
$(\statement,\witness^{\otimes k})$ 
where $\witness\in R_L(\statement)$ and $\witness^{\otimes k}$ is unentangled with $\dist$'s internal registers, 
$\ora_P^H(\statement,\witness^{\otimes k})$ generates $(\pkey,\vkey)\sample \setup(1^\secpar)$ and 
returns $(\vkey,\prove^H(\pkey,\statement,\witness^{\otimes k}))$, and
$\ora_S^{H,\reprogram}(\statement,\witness^{\otimes k})$ 
generates 
 $(\pkey,\vkey)\sample \setup(1^\secpar)$ and returns 
 $\siml^{H,\reprogram}(\pkey,\statement)$.
\end{definition}
\begin{remark}
The difference from the zero-knowledge property in the \QROVP model is that the malicious verifier is not allowed to maliciously generate $\pkey$. This is because the setup is supposed to be run by a trusted third party in this model.  
\end{remark}

%% file: differences.tex
\section{Differences from Previous Versions}\label{sec:differences}
Previous versions of this paper was submitted to CRYPTO 2020 and Asiacrypt 2020 before. Compared to those versions, the current version includes an additional result on Fiat-Shamir. The other results are basically unchanged. (The previous versions also included a construction of not necessarily zero-knowledge CV-NIP in the QSP model, but we omit it from the main claims in this version since that is subsumed by CV-NIZK in the QSP model. The construction is given in \cref{sec:CV-NIP} in the current version.)

We believe that the additional result on Fiat-Shamir makes this submission much stronger for the following reasons:
\begin{enumerate}
    \item This resolves the open problem explicitly raised by Broadbent and Grilo \cite{FOCS:BroGri20}, which is to construct NIZK for $\QMA$ in the shared Bell pair model via Fiat-Shamir transform. 
    \item Fiat-Shamir is a very important and well-studied paradigm in cryptography. To the best of our knowledge, we are the first to demonstrate its applicability to quantum protocols (rather than post-quantum protocols). 
    \item The result demonstrates further usefulness of classicalizing quantum protocols via quantum teleportation, which is the basic idea underlying our results.
\end{enumerate}

We would appreciate it if the reviewers who were negative in previous conferences re-evaluate the submission taking the above points into account. 

%% file: Omitted_Related_Works.tex
\section{Omitted Related Works}\label{sec:omitted_related_works}
\paragraph{More related works on quantum NIZKs.}
Kobayashi \cite{Kobayashi03} studied (statistically sound and zero-knowledge) NIZKs in a model where the prover and verifier share Bell pairs, and gave a complete problem in this setting.
It is unlikely that the complete problem contains (even a subclass of) $\NP$ \cite{MW18} and thus even a NIZK for all $\NP$ languages is unlikely to exist in this model.
Note that if we consider the prover and verifier sharing Bell pairs in advance like this model,
the verifier's preprocessing message 
of our protocols (and the protocol of~\cite{C:ColVidZha20}) 
becomes
classical.
Chailloux et al. \cite{TCC:CCKV08} showed that there exists a (statistically sound and zero-knowledge) NIZK for all languages in $\mathbf{QSZK}$ in the help model where a trusted party generates a pure state \emph{depending on the statement to be proven} and gives copies of the state to both prover and verifier.

\paragraph{Interactive zero-knowledge for $\QMA$.}
There are several works of interactive zero-knowledge proofs/arguments for $\QMA$.
The advantage of these constructions compared to non-interactive ones is that they do not require any trusted setup. 
Broadbent, Ji, Song, and Watrous \cite{BJSW20} gave the  first construction of a zero-knowledge proof for $\QMA$.
Broadbent and Grilo \cite{FOCS:BroGri20} gave an alternative simpler construction.
Bitansky and Shmueli \cite{STOC:BitShm20} gave the first constant round zero-knowledge argument for $\QMA$ with negligible soundness error. 
Brakerski and Yuen \cite{BraYue} gave a construction of $3$-round \emph{delayed-input} zero-knowledge proof for $\QMA$ where  the prover needs to know the statement and witness only for generating its last message. 
By considering the first two rounds as preprocessing, we can view this construction as a NIZK in a certain kind of preprocessing model. 
However, their protocol has a constant soundness error, and it seems difficult to prove the zero-knowledge property for the parallel repetition version of it.
\mor{Chardouvelis and Malavolta [iacr:2021/918] recently constructed 4-round statistical zero-knowledge arguments for $\QMA$ and
2-round zero-knowledge for $\QMA$ in the timing model.}

%% file: Appendix_Preliminaries.tex
\section{Omitted Contents in \cref{sec:preliminary}}
\subsection{Notations}\label{sec:notations}
\paragraph{Notations.}
We use $\secpar$ to denote the security parameter throughout the paper. 
For a positive integer $N$, $[N]$ means the set $\{1,2,...,N\}$.
For a probabilistic classical  or quantum algorithm $\A$, we denote by $y\sample \A(x)$ to mean $\A$ runs on input $x$ and outputs $y$. 
For a finite set $S$ of classical strings, $x\sample S$ means that $x$ is uniformly randomly chosen from $S$.
For a classical string $x$, $x_i$ denotes the $i$-th bit of $x$. 
For classical strings $x$ and $y$, $x\concat y$ denotes the concatenation of $x$ and $y$. 
We write $\poly$ to mean an unspecified polynomial and $\negl$ to mean an unspecified negligible function. 
We use PPT to stand for (classical) probabilistic polynomial time and QPT to stand for quantum polynomial time.
When we say that an algorithm is non-uniform QPT, it is expressed as a family of polynomial size quantum circuits with quantum advice. 
\subsection{Proof of \cref{lem:prob_and_energy}}\label{sec:proof_prob_and_energy}
\begin{proof}[Proof of \cref{lem:prob_and_energy}]
Let us define 
$V:=
(\prod_{j\in S_X}V_j(X))
(\prod_{j\in S_Y}V_j(Y))
(\prod_{j\in S_Z}V_j(Z))$, 
and $|m\rangle:=\bigotimes_{j=1}^N|m_j\rangle$.
Then,
\begin{eqnarray*}
\Pr\left[(-1)^{\bigoplus_{j\in S_X \cup S_Y\cup S_Z}m_j}=-s\right]
&=&\sum_{m\in\{0,1\}^N}\langle m|V^\dagger\rho V|m\rangle
\frac{1-s(-1)^{\bigoplus_{j\in S_X\cup S_Y\cup S_Z}m_j}}{2}\\
&=&\sum_{m\in\{0,1\}^N}\langle m|V^\dagger\rho V
\frac{I-s\prod_{j\in S_X\cup S_Y\cup S_Z}Z_j}{2}
|m\rangle\\
&=&\Tr\Big[V^\dagger\rho V
\frac{I-s\prod_{j\in S_X\cup S_Y\cup S_Z}Z_j}{2}\Big]\\
&=&\Tr\Big[\rho V
\frac{I-s\prod_{j\in S_X\cup S_Y\cup S_Z}Z_j}{2}V^\dagger\Big]\\
&=&\Tr\Big[\rho (I-\ham)\Big]\\
&=&1-\Tr(\rho \ham).
\end{eqnarray*}
\end{proof}
\subsection{Omitted Basic Lemmas}\label{sec:basic_lemmas}
Here, we review basic lemmas that are used in our security proofs.

\begin{lemma}[Pauli Mixing]\label{lem:Pauli_mixing}
Let $\rho$ be an arbitrary quantum state over registers $\regA$ and $\regB$, and let $N$ be the number of qubits in $\regA$. 
Then we have 
\[
\frac{1}{2^{2N}}
\sum_{x\in \bit^{N},z_\in \bit^{N}}
\left(X^xZ^z \otimes I_\regB\right)\rho \left(Z^zX^x \otimes I_\regB\right)=\frac{1}{2^{N}} I_\regA \otimes \Tr_{\regA}(\rho).
\]
\end{lemma}

This is well-known, and one can find a proof in e.g., \cite{FOCS:Mahadev18b}.

\begin{lemma}[Quantum Teleportation]\label{lem:teleportation}
Suppose that we have $N$ Bell pairs between registers $\regA$ and $\regB$, i.e., $\frac{1}{2^{N/2}}\sum_{s\in \bit^{N}}\ket{s}_\regA\otimes \ket{s}_\regB$, and let $\rho$ be an arbitrary $N$-qubit quantum state in register $\regC$. 
Suppose that we measure $j$-th qubits of $\regC$ and $\regA$ in the Bell basis for all $j\in [N]$, and let $((x_1,z_1),...,(x_N,z_N))$ be the measurement outcomes. 
Let $x:=x_1\concat x_2 \concat...\concat x_N$ and $z:=z_1\concat z_2 \concat...\concat z_N$. 
Then the $(x,z)$ is uniformly distributed over $\bit^N \times \bit^N$.
Moreover, conditioned on the measurement outcome $(x,z)$, the resulting state in $\regB$ is
$X^{x}Z^{z}\rho Z^{z}X^{x}$.
\end{lemma}
This is also well-known, and one can find a proof in e.g., \cite{NC00}.

\subsection{$\QMA$}\label{sec:appendix_QMA}
\begin{definition}[$\QMA$]\label{def:QMA}
We say that a promise problem $L=(L_\yes,L_\no)$ is in $\QMA$ if there is 
a polynomial $\ell$ and a QPT algorithm $V$ such that the following is satisfied: 
\begin{itemize}
    \item For any $\statement \in L_\yes$, there exists a quantum state $\witness$ of $\ell(|\statement|)$-qubit (called a witness) such that we have $\Pr[V(\statement,\witness)=1]\geq 2/3$.
    \item For any $\statement\in L_\no$ and any quantum state $\witness$ of $\ell(|\statement|)$-qubit, we have $\Pr[V(\statement,\witness)=1]\leq 1/3$.
\end{itemize}
\end{definition}

%% file: Appendix_proof_amplification.tex
\subsection{Proof of \cref{lem:CV-NIZK_amplification}}\label{sec:proof_amplification}
\begin{proof}[Proof of \cref{lem:CV-NIZK_amplification}]
Let $\Pi=(\setup,\prove,\verify)$ be a CV-NIZK for $L$ in the QSP model that satisfies $c$-completeness, $s$-soundness, and the zero-knowledge property for  some $0<s<c<1$ such that $c-s>1/\poly(\secpar)$.
Let $k$ be the number of copies of a witness $\prove$ takes as input.  
For any polynomial $N=\poly(\secpar)$, $\Pi^{N}=(\setup^{N},\prove^{N},\verify^{N})$ be the $N$-parallel version of $\Pi$. That is, $\setup^{N}$ and $\prove^{N}$ run $\setup$ and $\prove$ $N$ times parallelly and outputs tuples consisting of outputs of each execution, respectively where $\prove^{N}$ takes $Nk$ copies of the witness as input.
$\verify^{N}$ takes $N$-tuple of the verification key and proof, runs $\verify$ to verify each of them separately, and outputs $\top$ if the number of executions of $\verify$ that outputs $\top$ is larger than $\frac{N(\alpha+\beta)}{2}$.
By Hoeffding's inequality, it is easy to see that we can take $N=O\left(\frac{\log^{2} \secpar}{(\alpha-\beta)^2}\right)$ so that $\Pi^{N}$ satisfies $(1-\negl(\secpar))$-completeness and $\negl(\secpar)$-soundness. 

What is left is to prove that $\Pi^{N}$ satisfies the zero-knowledge property.
This can be reduced to the zero-knowledge property of $\Pi$ by a standard hybrid argument.
More precisely, for each $i\in \{0,...,N\}$, let $\ora_i$ be the oracle that works as follows where $\pkey'$ and $\vkey'$ denote the proving and verification keys of $\Pi^{N}$, respectively. 
\begin{description}
\item[$\ora_i(\pkey'=(\pkey^1,...,\pkey^N),\vkey'=(\vkey^1,...,\vkey^N),\statement,\witness^{\otimes Nk})$:]
It works as follows:
\begin{itemize}
    \item For $1\leq j\leq i$, it computes  $\pi_j\sample \siml(\vkey^j,\statement)$.
    \item For $i<j\leq N$,  it computes $\pi_j\sample \prove(\pkey^j,\statement,\witness^{\otimes k})$ where it uses the $(k(j-1)+1)$-th to $kj$-th copies of $\witness$.
    \item Output $\pi:=(\pi_1,...,\pi_N)$.
\end{itemize}
\end{description}
Clearly, we have $\ora_0(\pkey',\vkey',\cdot,\cdot)=\ora_P(\pkey',\cdot,\cdot)$  
and $\ora_N(\pkey',\vkey',\cdot,\cdot)=\ora_S(\vkey',\cdot,\cdot)$.\footnote{$\ora_P(\pkey',\cdot,\cdot)$ and $\ora_S(\vkey',\cdot,\cdot)$ mean the corresponding oracles for $\Pi^{N}$.
} 
Therefore, it suffices to prove that no  distinguisher can distinguish $\ora_i(\pkey',\vkey',\cdot,\cdot)$ 
and 
$\ora_{i+1}(\pkey',\vkey',\cdot,\cdot)$
for any $i\in \{0,1,...,N-1\}$. 
For the sake of contradiction, suppose that there exists a distinguisher $\dist'$ that 
distinguishes $\ora_i(\pkey',\vkey',\cdot,\cdot)$ and $\ora_{i+1}(\pkey',\vkey',\cdot,\cdot)$ with a non-negligible advantage by making one query of the form $(\statement,\witness^{\otimes Nk})$.
Then we construct a distinguisher $\dist$ that breaks the zero-knowledge property of $\Pi$ as follows:
\begin{itemize}
    \item[$\dist^{\ora}(\vkey)$:]
    $\dist$ takes $\vkey$ as input and is given a single oracle access to $\ora$, which is either 
    $\ora_P(\pkey,\cdot,\cdot)$
    or  $\ora_S(\vkey,\cdot,\cdot)$ where $\pkey$ is the proving key corresponding to $\vkey$.\footnote{ $\ora_P(\pkey,\cdot,\cdot)$ and $\ora_S(\vkey,\cdot,\cdot)$ mean the corresponding oracles for $\Pi$ by abuse of notation.
}  (Remark that $\dist$ is not given $\pkey$.)
    It sets $\vkey^{i+1}:=\vkey$ (which implicitly defines $\pkey^{i+1}:=\pkey$) and generates $(\pkey^j,\vkey^j)\sample \setup(1^\secpar)$ for all $j\in [N]\setminus \{i+1\}$.
    It sets $\vkey':=(\vkey^1,...,\vkey^N)$ and runs $\dist'^{\ora'}(\vkey')$ where when $\dist'$ makes a query $(\statement,\witness^{\otimes Nk})$ to $\ora'$, $\dist$ simulates the oracle $\ora'$ for $\dist'$ as follows:
    \begin{itemize}
          \item For $1\leq j\leq i$, $\dist$ computes  $\pi_j\sample \siml(\vkey^j,\statement)$.
    \item For $j=i+1$, $\dist$ queries $(\statement, \witness^{\otimes k})$ to the external oracle $\ora$ where it uses the $(ki+1)$-th to $k(i+1)$-th copies of $\witness$ as part of its query, and lets $\pi_{i+1}$ be the oracle's response.
    \item For $i+1<j\leq N$,  it computes $\pi_j\sample \prove(\pkey^j,\statement,\witness^{\otimes k})$ where it uses the $(k(j-1)+1)$-th to $kj$-th copies of $\witness$.
    We note that this can be simulated by $\dist$ since it knows $\pkey^j$  for $j\neq i+1$. 
    \item $\dist$ returns $\pi':=(\pi_1,...,\pi_N)$ to $\dist'$ as a response from the oracle $\ora'$.
    \end{itemize}
    Finally, when $\dist'$ outputs $b$, $\dist$ also outputs $b$.
\end{itemize}
We can see that the oracle $\ora'$ simualted by $\dist$ works similarly to $\ora_i(\pkey',\vkey',\cdot,\cdot)$ 
when $\ora$ is $\ora_P(\pkey,\cdot,\cdot)$
and works similarly to $\ora_{i+1}(\pkey',\vkey',\cdot,\cdot)$ when $\ora$ is $\ora_S(\vkey,\cdot,\cdot)$ where $\pkey'=(\pkey^1,...,\pkey^N)$.
Therefore, by the assumption that $\dist'$ distinguishes $\ora_i(\pkey',\vkey',\cdot,\cdot)$ and $\ora_{i+1}(\pkey',\vkey',\cdot,\cdot)$ with a non-negligible advantage, $\dist$ distinguishes $\ora_P(\pkey,\cdot,\cdot)$ and $\ora_S(\vkey,\cdot,\cdot)$ with a non-negligible advantage.
However, this contradicts the zero-knowledge property of $\Pi$.
Therefore, such $\dist'$ does not exist, which completes the proof of \cref{lem:CV-NIZK_amplification}. 
\end{proof}

%% file: Appendix_Omitted_Dual-Mode.tex
\section{Omitted Contents in \cref{sec:Dual-mode}}\label{sec:omitted_dual}
\subsection{Lossy Encryption}\label{sec:definition_lossy}
We give a definition of lossy encryption.

\begin{definition}[Lossy Encryption]\label{def:lossy}
A lossy encryption scheme 
over the message space $\mathcal{M}$ and the randomness space $\mathcal{R}$
consists of PPT algorithms $\Pi_\LE=(\injgen,\lossygen,\enc,\dec)$ with the following syntax.
\begin{description}
\item[$\injgen(1^\secpar)$:]
The injective key generation algorithm takes the security parameter $1^\secpar$ as input and ouputs an \emph{injective} public key $\pk$ and a secret key $\sk$.
\item[$\lossygen(1^\secpar)$:]
The lossy key generation algorithm takes the security parameter $1^\secpar$ as input and ouputs a \emph{lossy} public key $\pk$.
\item[$\enc(\pk,\mu)$:]
The encryption algorithm takes the public key $\pk$ and a message $\mu \in \mathcal{M}$ as input and outputs a ciphertext $\ct$.
This algorithm uses a randomness $R\in \mathcal{R}$.
We denote by  $\enc(\pk,\mu;R)$ to mean that we run $\enc$ on input $\pk$ and $\mu$ and randomness $R$ when we need to clarify the randomness.  
\item[$\dec(\sk,\ct)$:]
The decryption algorithm takes the secret key $\sk$ and a ciphertext $\ct$ as input and outputs a message $\mu$.
\end{description}
We require $\Pi_\LE$ to satisfy the following properties.
\end{definition}

\medskip
\noindent \underline{
\textbf{Correctness on Injective Keys}
}
For all $\mu\in \mathcal{M}$, we have 
\begin{align*}
    \Pr\left[
    \dec(\sk,\ct)=\mu:
    \begin{array}{ll}
    (\pk,\sk)\sample \injgen(1^\secpar)\\
    \ct\sample \enc(\pk,\mu)
    \end{array}
    \right]= 1.
\end{align*}

\medskip
\noindent \underline{
\textbf{Lossiness on Lossy Keys}
}
With overwhelming probability over $\pk\sample \lossygen(1^\secpar)$, for all $\mu_0,\mu_1 \in \mathcal{M}$ and all unbounded-time distinguisher $\dist$, we have  
\begin{align*}
    \left|\Pr\left[
    \dist(\ct)=1:
    \begin{array}{ll}
        \ct\sample \enc(\pk,\mu_0) 
    \end{array}
    \right]-
        \Pr\left[
    \dist(\ct)=1:
    \begin{array}{ll}
        \ct\sample \enc(\pk,\mu_1) 
    \end{array}
    \right]\right|
    \leq \negl(\secpar).
\end{align*}

\medskip
\noindent \underline{
\textbf{Computational Mode Indistinguishability}
}
For any non-uniform QPT distinguisher $\dist$, we have
\begin{align*}
    \left|\Pr\left[\dist(\pk_{\mathsf{inj}})=1\right]
    - \Pr\left[\dist(\pk_{\mathsf{lossy}})=1\right]
    \right|\leq \negl(\secpar)
\end{align*}
where $(\pk_{\mathsf{inj}},\sk)\sample \injgen(1^\secpar)$ and 
$\pk_{\mathsf{lossy}}\sample \lossygen(1^\secpar)$.

It is well-known that Regev's encryption \cite{JACM:Regev09} is lossy encryption under the LWE assumption with a negligible correctness error.
We can modify the scheme to achieve perfect correctness by a standard technique.   
Then we have the following lemma.
 
\begin{lemma}\label{lem:lossy_encryption}
If the LWE assumption holds, then there exists a lossy encryption scheme.
\end{lemma}

\subsection{Remarks on Definition of Dual-Mode CV-NIZK in the \CRSVP model.}\label{sec:remark_definition_dual}
Here we give remarks on the definition of dual-mode CV-NIZK in the \CRSVP model (\cref{def:dual-mode}).
\begin{remark}[On definition of  zero-knowledge property]
By considering a combination of $\crsgen$ (for a fixed $\mode$) and $\preprocess$ as a setup algorithm, (dual-mode) CV-NIZK in the \CRSVP~model can be seen as a CV-NIZK in the QSP model in a syntactical sense.    
However, it seems difficult to prove that this  satisfies (even a computational variant of) the zero-knowledge property  defined in \cref{def:CV-NIZK} due to the following reasons:
\begin{enumerate}
    \item In \cref{def:dual-mode}, $\siml_1$ is quantum, whereas a simulator is required to be classical in \cref{def:CV-NIZK}. 
    We observe that this seems unavoidable in the above model: If $\pkey$ is quantum, then a classical simulator cannot even take $\pkey$ as input. On the other hand, if $\pkey$ is classical, then that implies $L\in\mathbf{AM}$ similarly to the final paragraph of \cref{sec:CV-NIP}. 
    \item A simulator in \cref{def:dual-mode} can embed a trapdoor $\td$ behind the common reference string $\crs$  whereas a simulator in  \cref{def:CV-NIZK} just takes an honestly generated verification key $\vkey$ as input. 
    We remark that this also seems unavoidable since $\vkey$ may be maliciously generated when the verifier is malicious, in which case just taking $\vkey$ as input would be useless for the simulation. 
\end{enumerate}
On the other hand, the definition in \cref{def:dual-mode} allows a distinguisher (that plays the role of a malicious verifier)  to maliciously generate $\pkey$, which is  a stronger capability than that of a distinguisher in  \cref{def:CV-NIZK}. 
Therefore, the zero-knowledge properties in \cref{def:dual-mode} and  \cref{def:CV-NIZK} are incomparable. 
We believe that the definition of the zero-knowledge property in \cref{def:dual-mode} ensures meaningful security.
It 
roughly means that any malicious verifier cannot learn anything beyond what could be computed in quantum polynomial time by itself even if it is allowed to interact with many sessions of honest provers under maliciously generated proving keys and the reused honestly generated common reference string.
While this does not seem very meaningful when $L\in \BQP$, we can ensure a meaningful privacy of the witness when $L\in \QMA$.  
Finally we remark that our definition is essentially the same as that in \cite{C:ColVidZha20}  (except for the dual-mode property).  
\end{remark}
\begin{remark}[Comparison to NIZK in the malicious designated verifier model]
A CV-NIZK for $\QMA$ in the \CRSVP model as defined above is syntactically very similar to the NIZK for $\QMA$ in the malicious designated verifier model as introduced in \cite{Shmueli20}.
However, a crucial difference is that the proving key $\pkey$ is a quantum state in our case and cannot be reused whereas that is classical and can be reused for proving multiple statements in \cite{Shmueli20}.
On the other hand, a CV-NIZK in the \CRSVP model has two nice features  that the NIZK of \cite{Shmueli20} does not have: one is that verification can be done classically in the online phase and the other is the dual-mode property.
\end{remark}

\subsection{Soundness in Hiding Mode and Zero-Knowledge in Binding Mode}\label{sec:dual_mode_transfer}
We prove that a dual-mode CV-NIZK satisfies computational soundness in the hiding mode and computational zero-knowledge in the binding mode.  
\begin{lemma}
If a dual-mode CV-NIZK $\Pi=(\crsgen,\preprocess,\prove,\verify)$ for a $\QMA$ promise problem
$L=(L_\yes,L_\no)$ in the \CRSVP model satisfies statistical $s$-soundness in the binding mode, statistical zero-knowledge property in the hiding mode, and computational mode indistinguishability, then it also satisfies the following properties.

\medskip
\noindent \underline{
\textbf{(Exclusive-Adaptive) Computational $(s+\negl(\secpar))$-Soundness in the Hiding Mode}
}
For all non-uniform QPT adversaries $\A$, we have 
\begin{align*}
    \Pr\left[
  \verify(\crs,\vkey,\statement,\pi)=\top
    :
    \begin{array}{c}
          \crs\sample \crsgen(1^\secpar,\hiding) \\
         (\pkey,\vkey) \sample \preprocess(\crs)\\
       (\statement,\pi) \sample \A(\crs,\pkey)
    \end{array}
    \right]
    \leq s+\negl(\secpar).
\end{align*}
where $\A$'s output must always satisfy $\statement \in L_\no$.  


\medskip
\noindent \underline{
\textbf{(Adaptive Multi-Theorem) Computational Zero-Knowledge in the Binding Mode}
}
There exists a PPT simulator $\siml_0$ and QPT simulator $\siml_1$ such that for any  non-uniform QPT distinguisher $\dist$, we have 
\begin{align*}
    &\left|\Pr\left[\dist^{\ora_P(\crs,\cdot,\cdot,\cdot)}(\crs)=1:
    \begin{array}{c}
          \crs\sample \crsgen(1^\secpar,\binding) 
    \end{array}
    \right]\right.\\     &-\left.\Pr\left[\dist^{\ora_S(\td,\cdot,\cdot,\cdot)}(\crs)=1:
    \begin{array}{c}
          (\crs,\td)\sample \siml_0(1^\secpar)
    \end{array}
    \right]\right|\leq \negl(\secpar)
\end{align*}
where 
$\dist$ can make $\poly(\secpar)$ queries, which should be of the form $(\pkey,\statement,\witness^{\otimes k})$ where $\witness\in R_L(\statement)$ and $\witness^{\otimes k}$ is unentangled with $\dist$'s internal registers, $\ora_P(\crs,\pkey,\statement,\witness^{\otimes k})$ returns $\prove(\crs,\pkey,\statement,\witness^{\otimes k})$, and
  $\ora_S(\td,\pkey,\statement,\witness^{\otimes k})$ returns $\siml_1(\td,\pkey,\statement)$.
\end{lemma}  
Intuitively, the above lemma holds because soundness and zero-knowledge should transfer from one mode to the other by the mode indistinguishability since otherwise we can  distinguish the two modes. 
Here, security degrades to  computational ones as the mode indistinguishability only holds against QPT distinguishers. 
We omit a formal proof since this is easy and can be proven similarly to a similar statement for dual-mode NIZKs for $\NP$, which has been folklore and formally proven recently \cite{INDOCRYPT:ArtBel20}. 

\begin{remark}
Remark that soundness in the hiding mode is 
defined in the ``exclusive style" where $\A$ should always output $\statement \in L_\no$. 
This is weaker than  soundness in the ``penalizing style" as in \cref{def:dual-mode} where $\A$ is allowed to also output $\statement \in L_\yes$ and we add $\statement \in L_\no$ as part of the adversary's winning condition.  
This is because the adaptive soundness in the penalizing style does not transfer well through the mode change while the adaptive soundness in the exclusive style does.
This was formally proven for NIZK for $\NP$ in the common reference string model  in \cite{INDOCRYPT:ArtBel20}, and easily extends to CV-NIZK for $\QMA$ in the \CRSVP model.
This is justified by the impossibility of  penalizing-adaptively (computational) sound and statistically zero-knowledge NIZK for $\NP$ in the classical setting (under falsifiable assumptions)  \cite{TCC:Pass13}.
We leave it open to study if a similar impossibility holds for dual-mode CV-NIZK for $\QMA$ in the \CRSVP model.
\end{remark}

\subsection{Remark on Security Definition of Dual-Mode Oblivous Transfer}\label{sec:remark_definition_OT}
We give a remark on the definition of a dual-mode oblivious transfer (\cref{def:OT}). 
\begin{remark}[On security definition of dual-mode oblivious transfer]
We remark that security of a $k$-out-of-$n$  dual-mode  oblivious transfer as defined in \cref{def:OT}
does not imply UC-security \cite{Canetti20,C:PeiVaiWat08,SCN:Quach20} or even full-simulation security  in the standard stand-alone simulation-based definition \cite{RSA:Lindell08a}.
This is because the receiver's security in \cref{def:OT} only ensures privacy of $J$ and does not prevent a malicious sender from generating $\ot_2$ so that he can manipulate the message derived on the receiver's side depending on $J$.  
The security with such a weaker receiver's security is often referred to as half-simulation security \cite{EC:CamNevshe07}.  
We define the security in this way due to the following reasons:
\begin{enumerate}
    \item This definition is sufficient for constructing a dual-mode CV-NIZK in the \CRSVP model given in \cref{sec:construction_dual-mode}
    by additionally relying on lossy encryption. 
    \item We are not aware of an efficient construction of a $k$-out-of-$n$ oblivious transfer  that satisfies full-simulation security under a post-quantum assumption (even if we ignore the dual-mode property).  
    We note that Quach \cite{SCN:Quach20} gave a construction of a $1$-out-of-$2$ oblivious transfer with full-simulation security based on LWE and we can extend it to $1$-out-of-$n$ one.\footnote{His construction further satisfies UC-security, which is stronger than full-simulation security.}
    However, we are not aware of an efficient way to convert this into $k$-out-of-$n$ one without losing the full-simulation security. 
    We note that a conversion from $1$-out-of-$n$ to $k$-out-of-$n$ oblivious transfer by a simple $k$-parallel repetition  does not work if we require the  full-simulation security since a malicious sender can send different inconsistent messages in different sessions, which should be considered as an attack against full-simulation security.
One possible way to prevent such an inconsistent message attack is to let the sender prove that the messages in all sessions are consistent by using (post-quantum) NIZK for $\NP$ in the common reference string model \cite{C:PeiShi19}.
However, such a construction is very inefficient since it uses the underlying $1$-out-of-$n$ oblivious transfer in a non-black-box manner.
On the other hand, the half-simulation security is preserved under parallel repetitions as shown in \cref{sec:OT}, and thus we can achieve this much more efficiently.  
\end{enumerate}
\end{remark}

\subsection{Omitted Security Proofs}\label{sec:omitted_sec_proof_dual_mode_NIZK}
Here, we prove \cref{lem:DM_soundness,lem:DM_ZK}, i.e., prove soundness and the zero-knowledge property of $\Pi_\DM$ described in \cref{fig:dual}.

\input{security_proof_dual-mode}

%% file: Appendix_Omitted_Fiat-Shamir.tex
\section{Omitted Contents in \cref{sec:Fiat-Shamir}}
\subsection{Definition of non-interactive commitment scheme}\label{sec:def_commitment}

\input{main/non-interactive_commitment}
\subsection{Omitted remarks}\label{sec:remark_definition_sigma}

\input{main/remark_definition_sigma}
\subsection{Proof of \cref{thm:sigma}}\label{sec:proof_sigma}

\input{main/proof_sigma}
\subsection{Proof of \cref{lem:QRO_completeness_soundness}}\label{sec:proof_QRO_completeness_soundness}

\input{main/FS_completeness_soundness}
\subsection{Proof of \cref{lem:QRO_ZK}}\label{sec:proof_QRO_ZK}

\input{main/FS_ZK}